\documentclass[
  preprint,
  10pt,
  3p,
  fleqn,
]{elsarticle}

\usepackage{main}

\usepackage[english]{babel}

\begin{document}

\begin{frontmatter}
  \title{Probabilistic unifying relations for modelling epistemic and aleatoric uncertainty: semantics and automated reasoning with theorem proving\tnoteref{t1,t2}}
  \tnotetext[t1]{This document is the results of the research project RoboTest (\url{https://robostar.cs.york.ac.uk/}) funded by EPSRC.}
  \author[1]{Kangfeng Ye\corref{cor1}}
  \ead{kangfeng.ye@york.ac.uk}
  
  \author[1]{Jim Woodcock}
  \ead{jim.woodcock@york.ac.uk}
  
  \author[1]{Simon Foster}
  \ead{simon.foster@york.ac.uk}
  
  \cortext[cor1]{Corresponding author}
  
  \affiliation[1]{organization={Department of Computer Science, University of York}, 
    addressline={Deramore Lane, Heslington},
    postcode={YO10 5GH},
    city={York},
    country={United Kingdom}}
  
  \begin{abstract}
    Probabilistic programming combines general computer programming, statistical inference, and formal semantics to help systems make decisions when facing uncertainty. Probabilistic programs are ubiquitous, including having a significant impact on machine intelligence. While many probabilistic algorithms have been used in practice in different domains, their automated verification based on formal semantics is still a relatively new research area. In the last two decades, it has attracted much interest. Many challenges, however, remain. The work presented in this paper, probabilistic unifying relations (ProbURel), takes a step towards our vision to tackle these challenges.
    
    Our work is based on Hehner's predicative probabilistic programming, but there are several obstacles to the broader adoption of his work. Our contributions here include (1) the formalisation of its syntax and semantics by introducing an Iverson bracket notation to separate relations from arithmetic; (2) the formalisation of relations using Unifying Theories of Programming (UTP) and probabilities outside the brackets using summation over the topological space of the real numbers; (3) the constructive semantics for probabilistic loops using Kleene's fixed-point theorem; (4) the enrichment of its semantics from distributions to subdistributions and superdistributions to deal with the constructive semantics; (5) the unique fixed-point theorem to simplify the reasoning about probabilistic loops; and (6) the mechanisation of our theory in Isabelle/UTP, an implementation of UTP in Isabelle/HOL, for automated reasoning using theorem proving.
    
    We demonstrate our work with six examples, including problems in robot localisation, classification in machine learning, and the termination of probabilistic loops.
  \end{abstract}
  
  \begin{keyword}
    formal semantics \sep 
    fixed-point theorems \sep 
    predicative programming \sep
    UTP \sep
    probabilistic models \sep 
    probabilistic programs \sep 
    probability distributions \sep 
    formal verification \sep 
    quantitative verification \sep 
    automated reasoning \sep 
    theorem proving \sep 
    Isabelle/HOL \sep 
    robotics \sep 
    machine learning \sep 
    classification
  \end{keyword}
\end{frontmatter}


\section{Introduction}
\label{sec:intro}

\paragraph{\textbf{Motivations}}
Probabilistic programming combines general computer programming, statistical inference, and formal semantics to help systems make decisions when facing uncertainty. Probabilistic programs are ubiquitous but are particularly important in machine intelligence applications. Probabilistic algorithms have been used in practice for a long time in autonomous robots, self-driving cars, and artificial intelligence. There are tools for automated formal verification, particularly model checkers. However, many challenges remain, such as the following. What is the mathematical meaning of a probabilistic program? Are probabilistic programming languages expressive enough to capture rich features in real-world applications, such as epistemic and aleatoric uncertainty, discrete and continuous distributions, and real-time? How can we compare two programs? How can we implement a probabilistic specification as a probabilistic program? Can formal verification be largely automated and scaled to large systems without sacrificing accuracy? Does a probabilistic program almost surely terminate? What is the expected runtime of this program? The work presented in this paper, probabilistic unifying relations (ProbURel), takes a step towards our vision to tackle these challenges.

Uncertainty is essential to cyber-physical systems, particularly in autonomous robotics where we work. Such systems are subjected to various uncertainties, including real-world environments and physical robotic platforms, which present significant challenges for robots. Robots are usually equipped with probabilistic control algorithms to deal with these uncertainties. For example, a modern robot may use SLAM for localisation and mapping, the value iteration algorithm for probabilistic planning~\cite{Thrun2005}.

Probabilistic algorithms are intrinsically more difficult to program and analyse than non-probabilistic algorithms. For a specific input, the output of a probabilistic algorithm might be a distribution of possible outputs, not just a single output. Outputs with low probabilities are rare, which makes testing difficult. Probabilistic behaviour may be challenging to capture and model correctly because assumptions may not be obvious and may be left implicit. We need more precision in understanding an autonomous robot that may impose safety-critical issues on humans and their environments. One approach to addressing these challenges is providing probabilistic programs with formal syntax, semantics, and verification techniques to ensure they behave as expected in a real-world environment.

{
According to Gordon et al.~\cite{Gordon2014}, probabilistic programming includes two basic constructs to draw values from probabilistic distributions such as uniform distributions and condition values of variables. 
Probabilistic inference is the problem in probabilistic programming to compute explicit probability distributions or to compute relevant probabilities for particular events from probabilistic programs. 
}

\paragraph{\textbf{{Examples}}}
We illustrate a few examples and informally discuss their modelling and inference.

\begin{example}[The (forgetful) Monty Hall problem]
    \label{ex:monty}
The problem~\cite{Wikipedia2023} is a puzzle based on an American television game show \emph{Let's Make a Deal}. It is named after its original host, Monty Hall. Suppose you are the contestant and are given a choice of opening one of three doors. Behind one door is a car, and behind the others are goats. You pick a door, say No. 1, and the host, who knows what is behind each door, opens another door, say No. 3 and reveals a goat. He then asks, ``Do you want to pick door No. 2 instead of your original choice?'' The problem is simple: should you change your choice to maximise your chance to win a car?

To model this problem, we define three variables $p$, $c$, and $m$ for the door number having the prize (car), the contestant chooses, and the Monty chooses. We model the problem below. 
\begin{lstlisting}[language=PPL,]
p := rand({0..2}); c := rand({0..2});//Prize and contestant's choice are random
if(p=c) { // If the contestant's choice is the prize, 
  m := (c+1)%3 pc{1/2} m := (c+2)%3; // Monty randomly chooses other two doors
} else { 
  m := 3-c-p; // Monty chooses another door which has no prize
}/*
c := c // (without change of choice) */
c := 3-c-m; // If the contestant changes the choice
\end{lstlisting}
We use two constructs: \lstinppl{rand(S)} to draw a random value from set $S$ and \lstinppl{P pc\{r\} Q}, a binary probabilistic choice, to choose $P$ with probability $r$ and $Q$ with probability $1-r$.

The last line corresponds to the strategy for the contestant to change the initial choice. The question becomes ``which strategy will have the higher winning (c=p) probability?''

Suppose now that Monty forgets which door has the prize behind it. He opens either of the doors not chosen by the contestant. The contestant switches their choice to that door if the prize is revealed ($m=p$).  So the contestant will surely win. However, should the contestant switch if the prize is not revealed ($m \neq p$)? In this forgetful Monty, the new knowledge ($m \neq p$) is learned. 

Accordingly, we model the forgetful Monty problem below. 
\begin{lstlisting}[language=PPL,]
{
  p := rand({0..2)}; c := rand({0..2}); 
  m := (c+1)%3 pc{1/2} m := (c+2)%3; // Monty has no knowledge of p
} || (m != p) 
\end{lstlisting}
We introduce another construct \lstinppl{P || Q} to model the new knowledge (encoded in $Q$, e.g. \lstinppl{m != p} denoting the prize is not revealed) learned after $P$ is executed. So the question becomes how the distribution is updated after learning the new evidence? What is the winning probability?  
\end{example}

\begin{example}[Classification - COVID-19 diagnosis]
    \label{ex:covid}
We consider people using a COVID-19 test to diagnose if they may or may not have contracted COVID-19. The test result is binary and could be positive or negative. The test, however, is imperfect. It doesn't always give a correct result. 

We model the prior, the test, and the first test result positive below.
\begin{lstlisting}[language=PPL,]
{
  c := True pc{p1} c := False; // Prior probability of a person having COVID 
  // A test 
  if c { ct := Pos pc{p2} ct := Neg; } // True positive and True negative
  else { ct := Pos pc{p3} ct := Neg; } // False positive and False negative
} || (ct = Pos) // Learn the result is positive 
\end{lstlisting}
In the program, \lstinppl{c} and \lstinppl{ct} denotes if a person has COVID or not, and the test result; and \lstinppl{p1}, \lstinppl{p2}, and \lstinppl{p3} are parameters of this program.
We are interested in several questions. How likely is a randomly selected person to have COVID-19 if the first test result is positive?  Is it necessary to have the second test to reassure the result? 

Taken the second test into account, the new program is as follows.
\begin{lstlisting}[language=PPL,]
{
  { ... } // The previous program
  // The second test 
  if c { ct := Pos pc{p2} ct := Neg; }
  else { ct := Pos pc{p3} ct := Neg; }
} || (ct = Pos) // Learn the result is positive again 
\end{lstlisting}

But how much can the second test contribute to the diagnosis? How the result changes if the parameters are changed? 
\end{example}

\begin{example}[Robot localisation (RL)]
    \label{ex:robot}
A circular room has two doors and a wall. A robot with a noisy door sensor maps position to $\clz door$ or $\clz wall$. Doors are at positions 0 and 2; position 1 is a blank wall. 
We introduce a program variable $bel \in \{0 \upto 2\}$ to denote the position of the robot that we believe. 
When the reading of the door sensor is $\clz door$, it {is} four times more likely to be right than wrong {and likewise when the reading is wall}. 

The following program models two sensor readings and one movement in between the readings.
\begin{lstlisting}[language=PPL,]
{
  {
    bel := rand({0..2}); // Prior for bel is uniformly distributed 
  } || (3*door(bel)+1) ; // Likelihood function for the sensor result door
  bel := (bel + 1) % 3; // Move to the right
} || (3*wall(bel)+1) ; // Likelihood function for the sensor result wall 
\end{lstlisting}
The \lstinppl{door(bel)} (or \lstinppl{wall(bel)}) is a function returning 1 if the \lstinppl{bel} is 0 or 2 (or 1) and returning 0 otherwise.

We are interested in questions like how many measurements and moves are necessary to estimate the robot's location accurately.
\end{example}

\begin{example}[Flip a coin till heads]
    \label{ex:coin}
We consider the simplest probabilistic program with a loop: flip a coin until the outcome is heads, defined as follows.
\begin{lstlisting}[language=PPL,]
c := heads;
while (c=tail) { 
  c := heads pc{p} c:= tail; 
}
\end{lstlisting}
The \lstinppl{p} above is a parameter denoting the probability of getting a heads for a coin flip. It is 1/2 for a fair coin.
Does this program terminate? What is the probability distribution on termination? How is the distribution related to \lstinppl{p}? What is the semantics of this program? What is its expected runtime?
\end{example}

\begin{example}[Throw a pair of dice]
    \label{ex:dice}
This example~\cite{Hehner2011} is about throwing a pair of dice till they have the same outcome. We model it below.
\begin{lstlisting}[language=PPL,]
while (d1 != d2) {
  d1 := rand({1..6}); 
  d2 := rand({1..6}); 
}
\end{lstlisting}
This is slightly complex than the coin program in Example~\ref{ex:coin} because two variables are declared. 
Does this program terminate? What is the probability distribution on termination? What is the semantics of this program? Is it still as simple as the coin example? What is its expected runtime?
\end{example}

\begin{example}[One-dimensional simple random walk (SRW)]
    \label{ex:srw}
Grimmett and Welsh~\cite{Grimmett1986} defined various random walks. A random walk is \emph{simple} if at each time step it can move only to its next (or neighbouring) positions randomly in one of the lattice directions. A \emph{symmetric} simple random walk has the equal probability for each direction. 
It is also the Gambler's Ruin Problem with an absorbing barrier at 0. 
We model it as a probabilistic program below.
\begin{lstlisting}[language=PPL,]
x := m; // m is the initial position of $x$ 
while (x > 0) { 
  x := x + 1 pc{p} x := x - 1 
}
\end{lstlisting}
In the program, \lstinppl{m} and \lstinppl{p} are parameters.
The program with \lstinppl{p=1/2} (that is, symmetric) is widely studied, for example, in \cite{Hurd2003,McIver2005an,McIver2017,Chatterjee2020}. Unlike the coin and dice examples where each experiment (flip a coin or throw a pair of dice) is independent, each experiment in this example is not independent because the value of $x$ is updated. 

Does this program terminate? How does the termination relate to the parameter \lstinppl{p}? What is the probability distribution on termination? What is the semantics of this program? What is its expected runtime?
\end{example}

{
There are several challenges to model and answer the questions of these programs: 
\begin{enumerate*}[label={(\arabic*)}]
    \item the capability to model the learning process using conditional probability and joint probability as used in the Bayesian approach,
    \item the reasoning about probabilistic loops to give them a precise semantics (probability invariant) and their termination,
    \item the inference to get exact probability distributions and exact expected runtime, especially for programs with loops, and
    \item the guarantee of the correctness of the analysis. 
\end{enumerate*}
A considerable amount of literature has been published on addressing these problems, but none of them can address all these challenges.

McIver and Morgan's weakest pre-expectation~\cite{McIver2005bn} is based on the \emph{pGCL}~\cite{Morgan1999,McIver2005bn}, an extension of Dijkstra's Guarded Command Language (GCL)~\cite{Dijkstra1976} with a probabilistic choice construct. It is mechanised in High-Order Logic (HOL)~\cite{Gordon1993} by Hurd et al.~\cite{Hurd2005}, enabling verification of partial correctness of probabilistic programs. However, pGCL does not support conditional probability, so it cannot model the examples: the forgetful Monty, COVID, and the robot localisation, presented in Examples~\ref{ex:monty} to~\ref{ex:robot}. 

Based on the weakest pre-expectation semantics, Kaminski~\cite{Kaminski2019a} developed an advanced weakest precondition calculus which supports conditioning in cpGCL~\cite{Olmedo2018} using an \textbf{observe} statement and expected runtimes in the calculus~\cite{Kaminski2016}. The observe statement, however, conditions only boolean expressions. It cannot support the general likelihood functions as discussed in Example~\ref{ex:robot} where the two functions are real-valued expressions. The expected runtime analysis~\cite{Kaminski2018} is based on upper-bounds. It cannot reason about the exact expected runtimes which requires the reasoning of exact probability distributions.

Barthe et al.~\cite{Barthe2018} presented ELLORA, an assertion-based logic for probabilistic programs, implemented in the EasyCrypt theorem prover~\cite{Barthe2014}. However, pGCL does not support conditional probability, so it cannot model the examples: the forgetful Monty, COVID, and the robot localisation.

Schr\"oer et al.~\cite{Schroeer2023} developed expectation-based reasoning using a deductive verification infrastructure, based on the weakest pre-expectation semantics. Similarly, it cannot reason about the exact probability distributions and the expected runtimes. For example, the random walk is verified to be almost-surely terminated, but without its semantics or exact distributions. It is also not able to model the general likelihood functions.

Hehner's probabilistic predicate programming (PPP)~\cite{Hehner2004,Hehner2011} can model and reason about all these examples. But PPP is not formalised and implemented in any tool for automated verification.
}

{
The work presented in this paper aims to support modelling and analysis of these probabilistic programs and answer the questions which we are interested in. Additionally, as discussed later in Sect.~\ref{sec:concl:vision}, we aim to pursue a probabilistic semantic framework 
\begin{enumerate*}[label={(\arabic*)}]
    \item having an expressive language with rich semantics to model systems not only from abstract specification level but also concrete implementation level;
    \item able to unify different probabilistic models and programmings, so their tools can be integrated;
    \item extendible to support more features like more discrete distributions, nondeterminism, continuous distributions, time, communication and concurrency, because these features are essential in modelling robotic applications;
    \item providing a practical and decidable method to approximate the semantics for probabilistic loops because it is non-trivial to construct an invariant and prove it for a loop; and
    \item most importantly, supporting theorem proving because these programs usually have unbounded variables and infinite state space.
\end{enumerate*}
}

\paragraph{\textbf{Our approach}}
Our previous work probabilistic RoboChart~\cite{Ye2022} and probabilistic designs~\cite{Ye2021} model aleatoric uncertainty describing the natural randomness of physical processes. Another category is epistemic uncertainty due to the lack of knowledge of information, which is reducible by gaining more knowledge. {Usual probabilistic choice can model aleatoric uncertainty but not epistemic uncertainty because it requires the capability to update distributions or beliefs after learning new knowledge. To model this process, for example, in the Bayesian approach, conditional and joint probabilities should be supported.} 
In this paper, we present a probabilistic programming language, called \emph{probabilistic unifying relations} (ProbURel), based on Hehner's probabilistic predicative programming~\cite{Hehner2004,Hehner2011}, to model both aleatoric and epistemic uncertainty. This programming uses the subjective Bayesian approach to reason about epistemic uncertainty.

In Hehner's original work~\cite{Hehner2011}, a probabilistic program is given relational semantics, and its syntax is a mixture of relations and arithmetic. The presentation of syntax and semantics in the paper is not formal. For example, the semantics of a probabilistic \emph{ok} (skip) is given as $ok = \left(x' = x\right) \times \left(y' = y\right)$. There is a benefit to introducing semantics using examples, but it lacks formalisation. The operators like $=$ and $\times$ are not formally defined, and the types for variables and expressions are not given. The lack of this information makes the paper not easily accessible to readers, particularly for researchers aiming to use the work for automated reasoning of probabilistic programs. Therefore, our first contribution to this paper is formalising its syntax and semantics. We introduce a notation called Iverson brackets, such as $\ibracket{r}$, to establish a correspondence between relations $r$ and arithmetic ($0$ or $1$). For \emph{ok}, we could formalise it as $\ibracket{v' = v}$ where $v$ denotes the state space (composed of all variables) of a program. This notation separates relations ($v'=v$) with arithmetic, so expressions and operators in a program all have clear meanings or definitions depending on their contexts (relations or arithmetic) where the contexts can be easily derived because of the separation. 

In addition to syntax and semantics, the semantics for probabilistic loops are not formally presented and argued. Hehner proposed a more straightforward but more potent (than total correctness) approach~\cite{Hehner1999}: partial correctness + time to deal with the termination of loops (for conventional programs) with extra information about run time. His approach introduces a time variable $t$ with a healthiness condition, strict incremental for each iteration. The variable $t$ can be discrete or continuous and is an extended natural or real number to have $\infty$ for nontermination. The same approach is also applied to probabilistic programs~\cite{Hehner2011}, but the partial correctness of probabilistic loops is not formally reasoned about. 
For this reason, our second contribution in this paper is to bridge the semantic gap for probabilistic loops by establishing its semantics using fixed-point theorems: specifically Kleene's fixed-point theorem, to construct the fixed points using iterations. The advantage of having an iterative (or constructive) fixed point includes both theoretical semantics and practical computation or approximation. A hint of this would be possible to verify probabilistic loops using both theorem proving and model checking (based on approximation). 

To give the semantics using the fixed-point theorem, we define a complete lattice $\left([0,1], \leq\right)$ over the real unit interval (the real numbers between 0 and 1 inclusive). We restrict to the unit interval simply because probability values are between 0 and 1. To apply Kleene's theorem, we prove the loop function is Scott-continuous for the state space in which only finite states have positive probabilities. Then we define the semantics of loops as the least fixed point (\emph{lfp}) of the function where \emph{lfp} can be calculated iteratively as the supremum of the ascending Kleene chain (from the bottom of the complete lattice). This bridges the semantics gap, but it is still challenging to calculate \emph{lfp} because the chain is infinite. We, therefore, also present the strongest fixed point (\emph{gfp}) of the function, calculated iteratively as the infimum of the descending Kleene chain (from the top of the complete lattice) of the function. We prove a unique fixed point theorem where \emph{lfp} and \emph{gfp} are the same based on particular assumptions. The unique fixed point theorem makes reasoning about loops much more accessible because it is unnecessary to calculate \emph{lfp}. Instead, a fixed point must be constructed and proved with the unique fixed point theorem. This, eventually, is consistent with the loop semantics using Hehner's more straightforward approach. In particular, our semantics can be mechanised and automated.

In Kleene's theorem, the ascending and descending chains start from a pointwise constant function 0 and 1 (the bottom and the top of the complete lattice). The two pointwise functions are not distributions (where probabilities of the state space sum to 1). Indeed, the pointwise function 0 is a subdistribution (the probabilities sum to less than or equal to 1), and the pointwise function 1 is a superdistribution (the probabilities sum to larger than 1).  For this reason, our third contribution is to extend the semantic domain of the probabilistic programming language from distributions to subdistributions and superdistributions. Eventually, constructs like conditional, probabilistic choice, and sequential composition will not be restricted to programs that are distributions. This brings us to the required semantics to use Kleene iterations for the semantics of loops. 

The introduction of Iverson brackets also has another benefit: the relations inside the brackets can be easily characterised using alphabetised relations in UTP because relations in both Hehner's work and UTP are of the predicative style. Relations in our probabilistic programs are indeed UTP alphabetised relations, which allows us to reason about probabilistic programs using the existing theorem prover Isabelle/UTP for UTP. Our final contribution, therefore, is to mechanise the semantics of the probabilistic programming language in Isabelle/UTP. Our reasoning is primarily automated thanks to the various relation tactics in Isabelle/UTP. Six examples presented in this paper are all verified.
All definitions and theorems in this paper are mechanised, and accompanying icons (\isalogo) {encoding a hyperlink (available only in the electronic version of this paper)} to corresponding repository artefacts.

\paragraph{\textbf{Paper structure}}
The remainder of this paper is organised as follows. We review related work in Sect.~\ref{sec:relwork}. Section~\ref{sec:prelim} provides the necessary background for further presentation of our work in the subsequent sections. In Sect.~\ref{sec:ureal}, we define the complete lattice over the unit interval and then lift it to a complete lattice over pointwise functions. Section~\ref{sec:program} formalises our probabilistic programming with Iverson brackets defined. We also present proven algebraic laws for each construct. In Sect.~\ref{sec:rec}, we present the semantics of probabilistic loops and the fixed point theorem. Afterwards, we illustrate our reasoning approach using six examples. Two are classification problems in machine learning, and two contain probabilistic loops (see Sect.~\ref{sec:ex_cases}). Finally, we discuss future work in Sect.~\ref{sec:concl}. 

\section{Related work}
\label{sec:relwork}

\paragraph{\textbf{Imperative probabilistic sequential programming languages and their semantics}}
Imperative probabilistic programs are the extension of conventional imperative programs with the capability to model randomness (typically from \emph{random number generators}), usually using a binary probabilistic choice construct~\cite{McIver2005bn,Hehner2011} or a random number generator (\emph{rand})~\cite{Hehner2004,Dahlqvist2020} sampling from the uniform distribution over a set (either finite or infinite). 
McIver and Morgan~\cite{McIver2020} show any discrete distribution, including uniform distributions, can be achieved through a binary probabilistic choice using a fair coin. Our work presented in this paper, ProbURel, uses both a probabilistic choice and a construct to draw a discrete uniform distribution from a finite set (similar to \emph{rand}).

The ``predicate style'' semantics (for example, weakest precondition~\cite{Dijkstra1976}, Hoare logic~\cite{Hoare1969}, and predicative programming~\cite{Hehner1984a}) for conventional imperative programs are boolean functions over state space. The semantics is sufficient to reason about these programs in the qualitative aspect, including termination or total correctness. Still, reasoning about probabilistic programs with natural quantitative measurements is insufficient. For this reason, boolean functions are generalised to real-valued functions over state space~\cite{Kozen1981,Kozen1985,McIver2005bn,Hehner2004}. 
One exception is the relational semantics~\cite{He2004} that embeds standard programs into the probabilistic world. In this semantics, probability distributions over states are captured in a special program variable ($prob$, a total function) in a standard program, and, therefore, the semantics of such probabilistic programs are still boolean functions over state space. Programs in our ProbURel are real-valued functions over state space.

Kozen's extension~\cite{Kozen1981,Kozen1985} replaced nondeterministic choice in conventional imperative programs with probabilistic choice, while McIver and Morgan~\cite{Morgan1999} added a probabilistic choice construct (and so it has both nondeterministic and probabilistic choice). McIver and Morgan's weakest pre-expectation or expectation transformer semantics~\cite{McIver2005bn}, the real-valued expressions over state space are called expectations (indeed random variables). The weakest pre-expectation expressed as $pre=wp\left(P, [post]\right)$ (where the square bracket $[\varg]$ converts a boolean-valued predicate to an arithmetic value, especially $[\ptrue]=1$ and $[\pfalse]=0$), is the least pre-expectations (evaluated in the initial state of $P$) to ensure that the probabilistic program $P$ terminates with post-expectation $[post]$ in its final state. 
For example,  
\begin{align*}
    wp(x := x+1 _{(1/3)}\oplus x := x-1, [x \geq 0]) 
= (1/3)*[x = -1 \lor x = 0] + [x \geq 1] 
\end{align*}
means that in order for the program to establish $x \geq 0$, the probability of $x$ being $-1$ or $0$ (or $x \geq 1$, or $x < -1$) in its initial state is at least $1/3$ (or $1$, or $0$).
An extensive set of algebraic laws has been presented for reasoning about probabilistic programs, including loops and termination in~\cite[Appendix B]{McIver2005an}. 
Using this semantics, Kaminski~\cite{Kaminski2019a} developed an advanced weakest precondition calculus. {Based on the work, Schr\"oer et al.~\cite{Schroeer2023} developed a deductive verification infrastructure for verifying discrete probabilistic programs in terms of bounded expectations and expected runtimes, and termination probabilities using an intermediate verification language from which verification conditions are generated and verified in SMT solvers.} 
Our ProbURel uses a similar notation to the square bracket, called Iverson brackets. The predicates in ProbURel are UTP's alphabetised relations which have been used to establish program correctness using the weakest precondition calculus~\cite{Woodcock2004}. For this reason, our ProbURel can also describe the weakest pre-expectation semantics. 

The \emph{pGCL}~\cite{Morgan1999,McIver2005bn}, an extension of Dijkstra's Guarded Command Language (GCL)~\cite{Dijkstra1976} with a probabilistic choice construct, is a widely studied imperative probabilistic programming language. 
The weakest pre-expectation semantics is based on \emph{pGCL}. 
It is formalised in High-Order Logic (HOL)~\cite{Gordon1993} by Hurd et al.~\cite{Hurd2005} (based on the quantitative logic~\cite{Morgan1996a}), enabling verification of partial correctness of probabilistic programs, and also formalised in Isabelle/HOL by Cock~\cite{Cock2012} using shallow embedding (where probabilities are just primitive real numbers) to achieve improved proof automation.
The \emph{pGCL} has simple operational semantics~\cite{Gretz2014} using (parametric) Markov Decision Processes (MDPs) to establish a semantic connection with the weakest pre-expectation semantics and has relational semantics~\cite{Jifeng1997,He2004,Woodcock2019}, which is based on the theory of designs in UTP and mechanised in Isabelle/UTP~\cite{Ye2021}. The \emph{pGCL} contains a nondeterministic choice construct, but ProbURel in this paper does not include it. It is part of our future work to introduce nondeterminism. We also use Isabelle/UTP for automated verification, but the theory we use here is the theory of relations in UTP, which is more general than the theory of designs.

Probabilistic programs can be modelled as functions using a monadic interpretation. Hurd~\cite{Hurd2003} developed a formal HOL framework for modelling and verifying probabilistic algorithms using theorem proving. The work uses mathematical measure theory to represent probability space to model a random bit generator (an infinite stream of independent coin-flips). It uses a monadic state transformer to model probabilistic programs with higher-order logic functions. A probabilistic program consumes some bits from the front of the stream for randomisation and returns the remains. Audebaud et al.~\cite{Audebaud2009} use the monadic interpretation of randomised programs for probabilistic distributions (instead of measure theory) and mechanise their work in the Coq theorem prover~\cite{coq}. They consider probabilistic choice (without nondeterminism) in a functional language with recursion (instead of an imperative language). Programs in ProbURel are interpreted in imperative instead of monadic, and probabilistic loops are reasoned using fixed-point theories.

Dahlqvist et al.'s simple imperative probabilistic language~\cite{Dahlqvist2020} uses two constructs, $coin()$ and $rand()$, to introduce discrete and continuous uniform distributions, and both operational and denotational semantics are presented. In its operational semantics, a probabilistic program is assumed to start in two fixed infinite streams (one for $coin$ and one for $rand$), and the execution of each random sampling reads and removes the head from its corresponding stream. Eventually, the program is deterministic, and randomness is present in the infinite streams. This is similar to Hehner's probabilistic predicative programming~\cite{Hehner2004} where each call to $rand$ is stored in a mathematical variable (not a program variable). The denotational semantics of the simple language is given in terms of probability distributions. In ProbURel, we use a similar notation to $rand$ to draw a discrete uniform distribution from a finite set. The semantics of ProbURel are denotational.

Hehner~\cite{Hehner2004,Hehner2011} also generalises boolean functions for predicative programming to real-valued functions for probabilistic predicative programming. In his language, conditional and joint probability are modelled through sequential and parallel composition. One unique feature of the language is its capability to model epistemic uncertainty, due to the lack of knowledge of information and reducible after gaining more knowledge, and aleatoric uncertainty, due to the natural randomness of physical processes. Epistemic uncertainty is modelled through parallel composition using the subjective Bayesian approach. Our work, presented here, is based on Hehner's work. We formalise the syntax and semantics of the work, introduce UTP's alphabetised relations, bridge the semantics gap in dealing with probabilistic loops and mechanise it in Isabelle/UTP for automated reasoning.

Researchers also use Hoare logic to reason about probabilistic programs, such as the work presented in \cite{Ramshaw1979, 
denHartog2002, 
Chadha2007}, and 
VPHL~\cite{Rand2015} uses a weighted tree structure to represent probability distributions in its semantics and can reason about the partial correctness of probabilistic programs. The probabilistic relational Hoare logic (pRHL)~\cite{Barthe2009} is a Hoare quadruple that establishes the equivalence of two programs and the usual Hoare logic to relate programs as pre- and post-conditions. {ELLORA~\cite{Barthe2018} is an assertion-based program logic for probabilistic programs and mechanised in the EasyCrypt theorem prover~\cite{Barthe2014}. The logic is presented in both abstract and concrete. The abstract logic is used for reasoning about loops and adversaries while the concrete logic facilitates formal verification. The logic features the reasoning of the broad class of loops for absolute termination, AST, and general termination using different assertions.} ProbURel uses UTP's alphabetised relations, which have been used to establish program correctness using Hoare logic~\cite{Woodcock2004}. For this reason, our ProbURel can also describe probabilistic Hoare logic semantics. The semantics of ProbURel are denotational, and two programs are equivalent if they are equal functions over the same state space.

\paragraph{\textbf{Recursion and Almost-sure termination}}
Reasoning about recursion is usually hard and is especially harder~\cite{Kaminski2019} for probabilistic programs because semantically probabilistic programs associate states with probability distributions other than merely boolean information for conventional programs. 
From this aspect, conventional programs can be regarded as a particular case of probabilistic programs where probability is always 1 or 0, so probabilistic programs are a more general paradigm. 

Morgan and McIver's early work~\cite{Morgan1996b,McIver2005an} uses general techniques invariants and variants for reasoning about loops. Invariants for probabilistic loops are now expectations. Variants {($V$)} are still integer-valued expressions but 
\begin{inparaenum}[a)]
\item {they are bounded below and above ($L \leq V < H$) by fixed integer constants ($L$ and $H$) if the states that satisfy the conjunction ($G \land Inv$) of the loop guard condition $G$ and the invariant $Inv$, are infinite}; 
\item for every iteration, there is a fixed non-zero probability $\varepsilon$ that the invariants are strictly decreased. 
\end{inparaenum}
A variant is not required always to be strictly decreased now; it could also be increased.
The variant rule is strengthened later by McIver et al. to remove the need to bound above, allow (quasi-) variants to be real-valued expressions, and $\varepsilon$ to vary~\cite[Theorem~4.1]{McIver2017}. Their new variant rule relies on a \emph{supermartingale}, a sequence of random variables (RVs) for which the expected value of the current random variable is larger than or equal to that of the subsequent random variable. Two parametric antitone functions characterise the quasi-variant called $p$ and $d$ for a lower bound $d$ on how much a program must decrease the variant with at least probability $p$. The new rule enables them to reason about the two-dimensional random walk, which is believed to be hard. 
Chakarov et al.'s \emph{expectation invariants}~\cite{Chakarov2014} and Kaminski's \emph{sub-} and superinvariants~\cite{Kaminski2019a} are similar to Morgan and McIver's probabilistic invariants to use expectations for invariants. 
Our approach to reasoning about probabilistic loops is based on fixed-point theories, and the semantics of a loop is its unique fixed point when additional conditions are satisfied. We also use an iterative way to construct the fixed point, which is the supremum of an ascending chain.

As a consequence of the generality of probabilistic programs from conventional programs, it is much harder~\cite{Kaminski2015} to analyse the termination of probabilistic programs because now the program terminates with probability (instead of absolute termination~\cite{McIver2005an} for conventional programs). 
The usual knowledge for conventional programs, such as termination or nontermination, finite run-time, and compositionality, is not valid for probabilistic programs. 
This problem has attracted a lot of interest in recent decades, such as~\cite{McIver2005an,Bournez2005,Esparza2012,Chakarov2013,FerrerFioriti2015,McIver2017,Agrawal2017,Huang2018,Chatterjee2018}. The research area of interest for probabilistic programs is a weakened termination, called \emph{almost-sure termination} (AST), or termination with probability one. In other words, a probabilistic program may not always terminate, but the probability of divergence is 0. For example, flipping a fair coin until the outcome is heads is such a program. 
As Esparza et al.~\cite{Esparza2012} pointed out, (conventional) termination is a purely topological property, namely the absence of cycles, while AST requires arithmetic reasoning. 
Some recent studies have investigated the \emph{positive} almost-sure termination~\cite{Bournez2005} where probabilistic programs terminate in the finite expected time, and several studies have also assessed the \emph{null} almost-sure termination where probabilistic programs terminate almost-surely but not in the finite expected time. Hehner inspires our approach to reason about termination.

In Hehner's semantics~\cite{Hehner2011}, a time variable $t$ is introduced and is strictly increased in each iteration of a loop. For example, it can be an extended integer number (including $\infty$ for nontermination) and is used to count iterations. His approach~\cite{Hehner1999} for termination of (probabilistic) loops is stronger than total correctness (equal to partial correctness plus termination) because the time variable allows reasoning about not only whether a loop terminates but also when it terminates (run-time analysis). One example from his paper is a probabilistic loop about throwing a pair of dice till they have the same outcome.
The invariant (or hypothesis) of this loop, shown below and proved in Sect.~\ref{sec:ex_cases:dice}, gives the distribution of final states (primed variables).
\begin{align*}
    & H \defs {\ibracket{d_1' = d_2'} * \ibracket{t' \geq t + 1} * {\left(\frac{5}{6}\right)^{\left(t' - t - 1\right)}} * \left(\frac{1}{36}\right)}
\end{align*} 
We are interested in the probability distribution in terms of iterations or $t$, and so we substitute $\ibracket{d_1' = d_2'}$ with 6 because there are 6 possible combinations of $d_1'$ and $d_2'$ in each experiment to have them equal. 
\begin{align*}
    & Ht \defs {6 * \ibracket{t' \geq t + 1} * {\left(\frac{5}{6}\right)^{\left(t' - t - 1\right)}} * \left(\frac{1}{36}\right)} = {\ibracket{t' \geq t + 1} * {\left(\frac{5}{6}\right)^{\left(t' - t - 1\right)}} * \left(\frac{1}{6}\right)}
\end{align*} 
Provided the initial value of $t$ is 0, we plot the program's termination probability using this invariant in Fig.~\ref{fig:dice_t_prob}. Because $t$ counts iterations, the diagram also shows the probability of the termination in the exact iteration $t$. The probability for $t'=1$ is $1/6$ ($\approx 0.167$), so the program terminates in the first iteration with probability $1/6$ as expected ($6$ among total $6*6=36$ combinations).

\begin{figure}[!ht]
    \begin{center}
\begin{tikzpicture}
\begin{axis}[
    xlabel={$t'$},
    ylabel={Probability of termination},
    y tick label style={/pgf/number format/.cd,fixed,fixed zerofill,precision=2},
]
    \addplot table[header=false,col sep=&,row sep=\\,y expr={\thisrowno{1}}] {
        1 & 1/6\\
        2 & (1/6)*(5/6)\\
        3 & (1/6)*(5/6)^2\\
        4 & (1/6)*(5/6)^3\\
        5 & (1/6)*(5/6)^4\\
        6 & (1/6)*(5/6)^5\\
        7 & (1/6)*(5/6)^6\\
        8 & (1/6)*(5/6)^7\\
        9 & (1/6)*(5/6)^8\\
        10& (1/6)*(5/6)^9\\
        11& (1/6)*(5/6)^10\\
        12& (1/6)*(5/6)^11\\
        13& (1/6)*(5/6)^12\\
        14& (1/6)*(5/6)^13\\
        15& (1/6)*(5/6)^14\\
        16& (1/6)*(5/6)^15\\
        17& (1/6)*(5/6)^16\\
        18& (1/6)*(5/6)^17\\
        19& (1/6)*(5/6)^18\\
        20& (1/6)*(5/6)^19\\
        21& (1/6)*(5/6)^20\\
        22& (1/6)*(5/6)^21\\
        23& (1/6)*(5/6)^22\\
        24& (1/6)*(5/6)^23\\
        25& (1/6)*(5/6)^24\\
        26& (1/6)*(5/6)^25\\
        27& (1/6)*(5/6)^26\\
        28& (1/6)*(5/6)^27\\
        29& (1/6)*(5/6)^28\\
        30& (1/6)*(5/6)^29\\
    };
  \end{axis}
\end{tikzpicture}
    \end{center}
    \caption{Termination probability over $t'$ for dice.}
    \label{fig:dice_t_prob}
\end{figure}
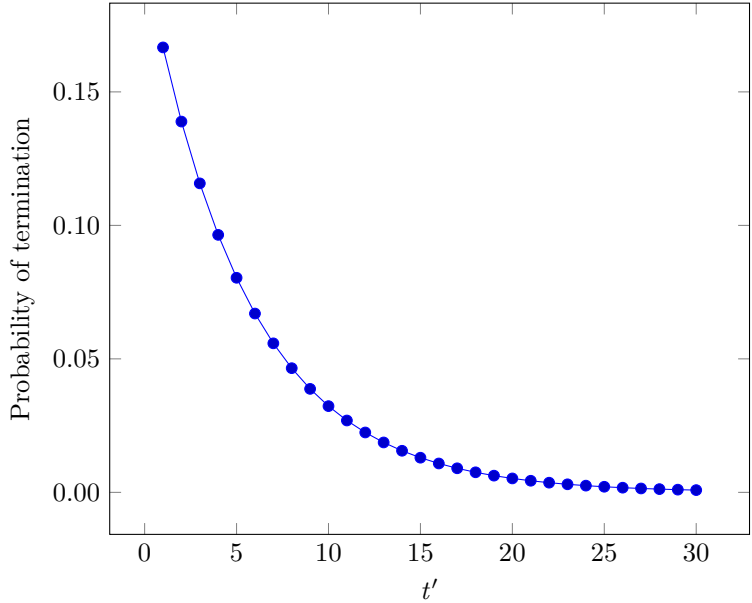

Reasoning about almost-sure termination becomes an arithmetic summation of this distribution, as shown below, which sums to 1.
\begin{align*}
\displaystyle\sum_{t'=0}^{\infty} \ibracket{t' \geq 1} * {\left(\frac{5}{6}\right)^{\left(t' - 1\right)}} * \left(\frac{1}{6}\right) 
    = \sum_{t'=1}^{\infty} {\left(\frac{5}{6}\right)^{\left(t' - 1\right)}} * \left(\frac{1}{6}\right) 
    =  \left(\frac{1}{6}\right) * \sum_{t'=1}^{\infty} {\left(\frac{5}{6}\right)^{\left(t' - 1\right)}}
    = 1.0
\end{align*}
\noindent
The expected run-time is the expectation of $t'$, simply the sequential composition of $\left(Ht ; t'\right) = t + 6$, denoting, on average, it takes six throws to have their outcomes equal.

%
%
{
\paragraph{\textbf{Summary}}
\begin{table*}[ht]
    \caption{{Comparison of different probabilistic semantics}}
    \label{tab:related_work}\centering
\bgroup
\def\arraystretch{1.0}
\setlength\tabcolsep{.5mm}
    \begin{tabular}{@{}l c cccccc c ccccc c cccc c c c @{}}
        \toprule
        \multirow{2}{*}{Approach} & 
        & \multicolumn{6}{c}{Modelling features} & \phantom{a} & \multicolumn{3}{c}{Advanced} & \phantom{a} & \multicolumn{5}{c}{Reasoning and Verification} & 
        & \multirow{2}{*}{Auto} \\
        \cmidrule{3-8} \cmidrule{10-12} \cmidrule{14-18}
        &
        & PrCh & UniD 
        & Cond
        & Jnt & NonD 
        & Expr
        & 
        & Para & Unf & Refn & 
        & WP & ExaPD & ExpRT 
        & Inv & 
        Para &
        \\
        \midrule
        PPDL\cite{Kozen1985} & 
             & \checkmark & \checkmark & &  &  & Bsc & & &  &  & 
             & & \checkmark & \checkmark & \checkmark & & & \\
        WPE\cite{McIver2005} & 
             & \checkmark & & &  & \checkmark & Bsc & 
             & &  & \checkmark &  
             & \checkmark & & & \checkmark & & & \autotp \\
        AWPE\cite{Kaminski2019a} & 
             & \checkmark & \checkmark & \checkmark &  & \checkmark & Bsc & 
             & &  & \checkmark &  
             & \checkmark & & \checkmark & \checkmark & & & \checkmark \\
        DVWPE\cite{Schroeer2023} & 
             & \checkmark & \checkmark & \checkmark &  & \checkmark & Bsc & 
             & &  & &  
             & \checkmark & & \checkmark & \checkmark & & & \checkmark \\
        ELLORA\cite{Barthe2018} & 
             & \checkmark & \checkmark & & & & Bsc & 
             & & & & 
             & ASS & \checkmark & \checkmark & \checkmark & & & \autotp \\
        PDs\cite{He2004} & 
             & \checkmark & & &  & \checkmark & Rich & 
             & \checkmark & \checkmark & \supportnoimpl  & 
             & CTr & \checkmark &  & \checkmark  & \checkmark & & \autotp \\
        PPP\cite{Hehner2004} & 
            & \checkmark & \checkmark & \checkmark & \checkmark&  & Bsc & 
            & & & \supportnoimpl & 
            & & \checkmark & \checkmark & \checkmark & & &  \\ 
        Our & 
            & \checkmark & \checkmark & \checkmark & \checkmark&  & Rich & 
            & \checkmark & \checkmark & \supportnoimpl & 
            & \supportnoimpl & \checkmark & \checkmark & \checkmark & \checkmark & & \autotp \\ 
        \midrule
        \multicolumn{20}{p{1.0\linewidth}}{\textbf{Acronym}: 
            ASS: assertion-based;
            Auto: automation;
            AWPE: advanced WPE;
            Bsc: basic;
            CTr: contract-based;
            Cond: conditional probability;
            DVWPE: WPE-based deductive verification infrastructure;
            ExaPD: exact probability distributions;
            Exp: expression and type system; 
            ExpRT: expected runtime;
            Inv: loop invariant; 
            Jnt: joint probability, especially supporting general likelihood functions; 
            NonD: nondeterministic choice;
            Para: parametric models or reasoning;
            PDs: probabilistic designs; 
            PrCh: probabilistic choice; 
            Refn: refinement; 
            TP: theorem proving;
            Unf: unification; 
            UniD: construct for uniform distributions;
            \supportnoimpl: supported but not yet developed;
            \autotp: automation through theorem proving; 
        } \\
        \bottomrule
\end{tabular}
\egroup
\end{table*}

In Table~\ref{tab:related_work}, we summarise the comparison of different probabilistic semantics in terms of four perspectives: modelling features, advanced features, reasoning and verification, and automation. In modelling features, we consider the support of constructs for usual probabilistic choice, uniform distributions, conditioning, joint probability, nondeterministic choice, and expression and type systems. Our work supports all but nondeterminism. Notably, our language has a rich expression and type system which is based on the Z notation~\cite{Spivey1992,Woodcock1996} and the mechanised Z mathematical toolkit\footnote{\url{https://github.com/isabelle-utp/Z_Toolkit}.} in Isabelle/HOL, which entitles us to model abstract probabilistic programs and capture rich semantics. Our language and semantic framework can support parametric models, is able to unify other semantics and support refinement thanks to our UTP relations. Probabilistic designs are also based on UTP and so the theory supports similar features as ours. In terms of verification, our work can be used to unify the WPE semantics (but this needs the new development), and support parametric verification. 
}

\section{Preliminaries}
\label{sec:prelim}
\subsection{Unifying Theories of Programming}
In UTP~\cite{Hoare1998,Woodcock2004}, the meaning (denotational semantics) of programs is given as predicates, called ``programs-as-predicates''~\cite{Hoare1984}. In this approach, the alphabetised relational calculus~\cite{Woodcock2004}, a combination of standard predicate calculus operators and Tarski's relation algebra~\cite{Tarski1941}, is used as the basis for its semantic model to denote programs as binary relations between initial observations and their subsequent observations. An alphabetised relation is an alphabet-predicate pair $\left(\alpha P, P\right)$ where the accompanying alphabet of $P$, $\alpha P$, is composed of undashed variables ($x$) and dashed variables ($x'$), representing observations made initially and subsequently. 
For example, a program $\left(x' := x + 1\right)$ with two observable variables $x$ and $y$ can be modelled as a relational predicate $\left(x' = x + 1 \land y' = y\right)$ with its alphabet $\{x,x',y,y'\}$.

Alphabetised predicates are presented in this representation of UTP through alphabetised expressions $[V,S] \uexpr$, parametric over the value type $V$ and the observation space $S$ and defined as total functions $S \fun V$. Predicates\footnote{In the rest of the paper, we use expressions, predicates, and relations to refer to alphabetised counterparts for simplicity.} are boolean expressions: $[S] \upred \defs [\bool, S] \uexpr$, an expression whose value type is boolean. Relations are predicates over a product space: $[S_1, S_2] \urel \defs [S_1 \cross S_2] \upred$, where $S_1$ and $S_2$ are the initial and final observation space, corresponding to undashed variables (input alphabet) and dashed variables (output alphabet). {Here $\urel$ means the UTP relations.} \emph{Homogeneous} relations have the same initial and final observation space: $[S] \hrel \defs [S, S] \urel$.

The denotational semantics of a sequential program is given as relations by the composition of constructors, including conditional, assignment, skip, sequential composition, and nondeterministic choice. These constructors are defined below.
\begin{definition}[Constructs of sequential programs]
    \label{def:utp_relation}
    \begin{align*}
        & \left(P \lhd b \rhd Q\right) \defs \left(b \land P\right) \lor \left(\lnot b \land Q\right)\tag*{(conditional)} \label{def:ucond} \\
        & \left(x :=_A e\right) \defs \left(x' = e \land w' = w \right)\tag*{(assignment)} \label{def:uassign} \\
        & \II \defs \left({v' = v}\right) \tag*{(skip)} \label{def:uskip} \\
        & P ; Q \defs \left(\exists s_0 @ P[s_0/s'] \land Q[s_0/s]\right)\tag*{(sequential composition)} \label{def:useq} \\
        & P \sqcap Q \defs \left(P \lor Q\right) \qquad{\text{ if } \alpha P = \alpha Q}\tag*{(nondeterminism)} \label{def:unondeter}
    \end{align*}
\end{definition}

Particularly, we need to emphasise the type of programs and their alphabets. In \ref{def:ucond}, {$b$ is type $[S_1]\upred$, $P$ and $Q$ are of type $[S_1,S_2]\urel$}, and $\alpha b \subseteq \alpha P = \alpha Q$. In \ref{def:uassign}, $A$ is the observation space of a program, including variable $x$ and a set $w$ of other variables. We also use a simple syntax $w'=w$ here to denote conjunctions of equations over each variable in $w$. For example, if $w = \{y, z\}$, then $\left(w' = w\right) \defs \left(y' = y \land z' = z\right)$. The subscript $A$ in $:=_A$ is usually omitted because it can automatically be derived from its context. Skip $\II$ \ref{def:uskip} is a special assignment where no variable changes {(here $v$ denotes the set of all variables)}, so the observation space stays the same. In sequential composition $P ; Q$ \ref{def:useq}, two relations $P$ of type $[S_1, S_2]\urel$ and $Q$ of type $[S_2, S_3]\urel$ are composed because the output alphabet $S_2$ of $P$ is the same as the input alphabet $S_2$ of $Q$. The relational composition gives a program of type $[S_1, S_3]\urel$ with $s_0:S_2$ denotes the entire state, $P[s_0/s']$ for the substitution of the final observation $s'$ of $P$ by $s_0$, and $Q[s_0/s]$ for the substitution of the initial observation $s$ of $Q$ by $s_0$. Nondeterministic choice $P \sqcap Q$ \ref{def:unondeter} is simply a disjunction of relations {if they have the same alphabets. The introduction of the new notation $\sqcap$ emphasises this condition.}

UTP uses refinement to deal with program development or correctness. A specification $S$ is refined by a program $P$, denoted as $S \refinedby P$, if and only if that $P$ implies $S$ is universally closed. For example, $\left(x := x + 1\right)$ is a refinement of $\left(x' > x\right)$ because for any $x$ and $x'$, $\left(x'=x+1\right) \implies \left(x' > x\right)$. Relations of type $[S_1, S_2]\urel$, for any given $S_1$ and $S_2$, are partially ordered by $\refinedby$ where $\ufalse$ and $\utrue$, special relations whose predicates are $\pfalse$ and $\ptrue$, are at its extremes: $\utrue \refinedby P \refinedby \ufalse$.

\subsection{Isabelle/UTP}
Isabelle/UTP~\cite{Foster2020} provides a shallow embedding of UTP's alphabetised relational calculus on top of Isabelle/HOL. In Isabelle/UTP, variables are modelled as algebraic structures using Lenses~\cite{Foster2009,Foster2016} to represent observations. Each observable variable $x$ is a lens ($\lens{\mathcal{V}}{\mathcal{S}}$), equipped with a pair of functions $\lget{x}: \mathcal{S} \fun \mathcal{V}$ and $\lput{x}: \mathcal{S} \fun \mathcal{V} \fun \mathcal{S}$, to query and update a view (of type $\mathcal{V}$) of an observation space (of type $\mathcal{S}$).
In this model, an alphabetised predicate $\left(x = y + 1\right)$ with two observable variables $x$ and $y$ is expressed as $\left(\lambda s. \lget{x} s = \lget{y} s + 1\right)$. A relation is a predicate over a product space, manipulated through a product lens $S_1 \lprod S_2$. The early relational example $\left(x' := x + 1\right)$, therefore, can be expressed as 
\begin{align*}
\lambda s. \lget{x} \left(\lget{\lsnd} s\right) = \lget{x}\left(\lget{\lfst} s\right) + 1 \land \lget{y}\left(\lget{\lsnd} s\right) = \lget{y}\left(\lget{\lfst} s\right) \tag*{[Lens representation]} \label{eqn:lens_representation} 
\end{align*}
where $\lfst$ and $\lsnd$ are the lenses to project the first and the second element of a product space. By substituting the state $s$ with a pair $(s, s')$, this expression can be simplified to 
\begin{align*}
    \lambda (s, s'). \lget{x} s' = \lget{x} s + 1 \land \lget{y} s' = \lget{y} s \tag*{[Simplified lens representation]} \label{eqn:simp_lens_representation}
\end{align*}
However, writing UTP expressions this way is tedious and not very useful and intuitive for good programming practice because too many implementation details are presented. For this reason, Isabelle/UTP implemented a lifted parser to provide a transparent conversion between the lens's representation and the programming syntax like $\left(x' := x + 1\right)$. We denote this representation of UTP expressions as $\usexpr{expr}$, such as $\usexpr{x' : = x + 1}$, which is converted to \ref{eqn:lens_representation}.

In Isabelle/UTP, we use $\vv$ to denote the universe alphabet of a program. In other words, it is the set of all observable variables. We also use $\vv'$ to denote the set of all dashed observable variables. For the previous example, $\vv$ denotes $\{x,y\}$ and $\vv'$ denotes $\{x', y'\}$.

\subsection{Probabilistic predicative programming}
Predicative programming~\cite{Hehner1984a,Hehner1984}, {or programs-as-predicates~\cite{Hehner1993}}, describes programs using first-order semantics or relational semantics as boolean expressions (predicates). A program has its input denoted by undashed variables and output denoted by dashed variables. Predicative programming also uses refinement for program correctness.

Probabilistic predicative programming~\cite{Hehner2004,Hehner2011} generalises predicative programming from boolean to probabilistic. Notations are introduced for probabilistic programming, such as skip, assignment, conditional choice, probabilistic choice, sequential composition (conditional probability), parallel composition (joint probability), and recursion. Except for parallel composition, these constructors deal with probabilistic programs whose outputs are distributions or distribution programs. Parallel composition can deal with probabilistic programs whose outputs might not be distributions (non-distribution programs), and uses normalisation to give a distribution program. 

This programming supports the subjective Bayesian approach through parallel composition. From a given distribution program, we can learn a new fact by placing the fact in parallel with the distribution program to allow beliefs to be updated. 

To reason about the termination of loops, a time variable is introduced to count iterations. This gives more information (time) than just termination~\cite{Hehner1999}. In this programming, the expected value of a number expression $e$ according to a distribution program $P$ is just the sequential composition of $P$ and $e$. If $e$ is a boolean expression, the sequential composition gives the probability that $e$ is valid after the execution of $P$. 
With the time variable, the programming allows reasoning about the average termination time. For example, on average, it takes 
two flips of a fair coin to see heads or tails. 
The termination probability of a loop (\isopbf{while} $b$ \isopbf{do} $P$) can be specified using the sequential composition of the solution (of the loop) and the negation of the loop condition ($\lnot b$). If the result is 1, it means the loop almost surely terminates. If the result is not 1, the loop may diverge.

\subsection{Complete lattices and fixed-point theorems}

A partially ordered set (poset) $(X, \leq)$ is a complete lattice if every subset of $X$ has a supremum and an infimum.
\begin{align*}
	& \forall A \subseteq X \bullet \left(\thinf{} A\right) \in X  \tag*{(Inf exists)} \label{law:inf_exists} \\ 
	& \forall A \subseteq X \bullet \left(\thsup{} A\right) \in X  \tag*{(Sup exists)} \label{law:sup_exists}
\end{align*}

We use a tuple $\left(X, \leq, <, \bot, \top, \tinf, \tsup, \thninf, \thnsup\right)$ to represent a complete lattice $\left(X, \leq\right)$ with a strict binary relation $<$, the bottom element $\bot$, the top element $\top$, the infimum $\tinf$ of two elements, the supremum $\tsup$ of two elements, the infimum $\thninf$ of a (finite or infinite) set of elements, and the supremum $\thnsup$ of a (finite or infinite) set of elements. 

A complete lattice satisfies more laws below. 
\begin{align*}
	& x \leq x \tag*{(reflexive)} \label{law:reflexive} \\
	& x \leq y \land y \leq z \implies x \leq z \tag*{(transitive)} \label{law:transitive} \\
	& x \leq y \land y \leq x \implies x = y	 \tag*{(antisym)} \label{law:antisym}\\
	& x \sqcap y \leq x \tag*{(inf\_le1)} \label{law:inf_le1} \\
	& x \sqcap y \leq y \tag*{(inf\_le2)} \label{law:inf_le2} \\
	& x \leq y \land x \leq z \implies x \leq y \sqcap z \tag*{(inf\_greatest)} \label{law:inf_greatest} \\
	& \left(x \leq y\right) \equiv \left(x \sqcap y = x\right) \tag*{(inf\_iff)} \label{law:inf_iff} \\
	& x \leq x \sqcup y \tag*{(sup\_ge1)} \label{law:sup_ge1} \\
	& y \leq x \sqcup y \tag*{(sup\_ge2)} \label{law:sup_ge2} \\
	& y \leq x \land z \leq x \implies y \sqcup z \leq x \tag*{(sup\_least)} \label{law:sup_least} \\
	& \left(x \leq y\right) \equiv \left(x \sqcup y  = y\right) \tag*{(sup\_iff)} \label{law:sup_iff} \\
	& x \in A \implies \thinf{}A \leq x \tag*{(Inf\_lower)} \label{law:Inf_lower} \\
	& (\forall x. x \in A \implies z \leq x) \implies z \leq \thinf{} A \tag*{(Inf\_greatest)} \label{law:Inf_greatest} \\
	& x \in A \implies x \leq \thsup{}A  \tag*{(Sup\_upper)} \label{law:Sup_upper} \\
	& (\forall x. x \in A \implies x \leq z) \implies \thsup{} A \leq z\tag*{(Sup\_least)} \label{law:Sup_least} \\
	& \thinf{} \{\} = \top  \tag*{(Inf\_empty)} \label{law:Inf_empty} \\
	& \thsup{} \{\} = \bot  \tag*{(Sup\_empty)} \label{law:Sup_empty}
\end{align*}

%
%
%
%
%


Monotonic and antimonotonic functions in order theory are characterised using $\mono$ and $\antimono$ defined below.
\begin{definition}[Monotone and anti-monotone]
    Provided $(X, \leq)$ and $(X', \leq')$ are posets and $f$ is a function of type $X \fun X'$, then
    \begin{align*}
        & \mono(f) \defs \forall x @ \forall y @ x \leq y \implies f(x) \leq' f(y) \tag*{(monotone)} \label{def:mono}\\ 
        & \antimono(f) \defs \forall x @ \forall y @ x \leq y \implies f(y) \leq' f(x) \tag*{(anti-monotone)} \label{def:antimono}
    \end{align*}
\end{definition}

Ascending and descending chains are monotonic and antimonotonic functions whose domain is natural numbers. 
\begin{definition}[Chains]
    Provided $(X, \leq)$ is a {complete lattice} and $f$ is a function of type $\nat \fun X$, then
    \begin{align*}
        & \incseq(f) \defs \mono(f) \tag*{(ascending chain)} \label{def:incseq}\\ 
        & \decseq(f) \defs \antimono(f) \tag*{(descending chain)} \label{def:decseq}
    \end{align*}
\end{definition}
{We particularly define $\incseq$ and $\decseq$ to be over complete lattices which we are interested in this paper. This is to simplify the specification of premises in lemmas and theorems because $\incseq$ and $\decseq$ impose a type restriction to complete lattices directly. Otherwise, we need additional premises if we use the more general $\mono$ and $\antimono$.}

We show the application of a monotonic function $f$ to an ascending chain $c$ is also an ascending chain.
\begin{thm}
    We fix $c:\nat \fun X$ and $f: X \fun Y$, then 
    $\incseq(c) \land \mono(f) \implies \incseq(\lambda n @ f(c(n)))$
\end{thm}

We show the application of a monotonic function $f$ to a descending chain $c$ is also a descending chain.
\begin{thm}
    We fix $c:\nat \fun X$ and $f: X \fun Y$, then 
    $\decseq(c) \land \mono(f) \implies \decseq(\lambda n @ f(c(n)))$
\end{thm}

If $f$ is an ascending or descending chain, its limit is the supremum or infimum of the chain.
\begin{thm}[Limit as supremum and infimum]
    \label{thm:limit_as_sup_inf}
    Provided $(X, \leq)$ is a complete lattice and also totally ordered, and we fix $f:\nat \fun X$, then 
    \begin{align*}
        & incseq(f) \implies f \tendsto \left({\thnsup n @ f\left(n\right)} \right) \tag*{(limit as supremum)}\label{thm:limit_as_sup}\\
    & decseq(f) \implies f \tendsto \left({\thninf n @ f\left(n\right)} \right)\tag*{(limit as infimum)}\label{thm:limit_as_inf}
    \end{align*}
\end{thm}
Here we use $f \tendsto v$ to denote the limit of $f$ is $v$: ${\displaystyle \lim_{n \to \infty} f(n) = v}$. The definition of the limit of a sequence is given below.
\begin{definition}[Limit of a sequence]\label{def:tendsto}
    A sequence $f$ converges to $v$ if and only if
    \begin{align*}
        & \forall \varepsilon:\real > 0 \bullet \exists N:\nat \bullet \forall n \geq N \bullet |f(n) - v| < \varepsilon
    \end{align*}
\end{definition}

\begin{thm}[Knaster–Tarski fixed-point theorem]
    \label{thm:tarski_fixed_point}
    Provided $(X, \leq)$ is a complete lattice and $F: X \fun X$ is monotonic, the set of fixed points of $F$ also forms a complete lattice. The least fixed point is the infimum of the pre-fixed points. 
    \begin{align*}
        &\thlfp{}~F \defs \thninf{} \left\{u : X | F(u) \leq u\right\} \tag*{(least fixed point)} \label{def:tarski_lfp}
    \end{align*}

The greatest fixed point is the supremum of the post-fixed points. 
    \begin{align*}
        &\thgfp{}~F \defs \thnsup{} \left\{u : X | u \leq F(u)\right\} \tag*{(great fixed point)} \label{def:tarski_gfp}
    \end{align*}
\end{thm}

\begin{definition}[Scott continuity~\cite{Abramsky1995}]
Suppose $(X, \leq$) and $(X', \leq'$) are complete lattices. A function $F: X \to X'$ is \textbf{Scott-continuous} or \textbf{continuous} if, for every non-empty chain $S \subseteq X$,
\begin{align*}
	F\left(\thsup{X} S\right) = \thsup{X'} F(S) \tag*{(continuous)} \label{law:continuity}
\end{align*}
Here we use $F(S)$ to denote the set $\left\{d \in S @ F(d)\right\}$, the relational image of $S$ under $F$ or the range of $F$ domain restricted to $S$.
\end{definition}

In the original definition of Scott continuity, both $X$ and $X'$ are directed-complete partial orders (dcpo). We fix them to be complete lattices because every complete lattice is a dcpo~\cite{Abramsky1995}, and we only consider complete lattices in this paper. If $X$ and $X'$ are identical lattices, the subscript of $\thsup{}$ in Definition~\ref{law:continuity} can be omitted.

\begin{thm}[Monotonicity]
    \label{thm:cont_mono}
   A continuous function is also monotonic. 
\end{thm}

\begin{thm}[Kleene fixed-point theorem]
    \label{thm:kleene_fixed_point_theorem}
    Provided $(X, \leq)$ is a complete lattice with a least element $\bot$ and a top element $\top$, and $F: X \fun X$ is continuous, then $F$ has a least fixed point $\thlfp{}~F$ and a greatest fixed point $\thgfp{}~F$.
    \begin{align*}
        \thlfp{}~F &= \thsup{$n \geq 0$} F^n(\bot) \tag*{(least fixed point)} \label{def:kleene_lfp}\\
        \thgfp{}~F &= \thinf{$n \geq 0$} F^n(\top)\tag*{(greatest fixed point)} \label{def:kleene_gfp}
    \end{align*}
    Here we use $\thsup{$n \geq 0$} F^n(\bot)$ to denote $\thnsup\left\{n : \nat @ F^n(\bot)\right\}$
\end{thm}

\begin{proof}
    We prove $\thlfp{}~F$ is a fixed point first and then prove $\thlfp{}~F$ is the least one.
\begin{align*}
& F \left(\thlfp{}~F\right) \\
 = & \quad \cmt{Definition~\ref{def:kleene_lfp}} \\
& F \left(\thsup{$n \geq 0$} F^n(\bot)\right) \\
 = & \quad \cmt{Continuity Definition~\ref{law:continuity}} \\
& \thsup{$n \geq 0$} F \left(F^n(\bot)\right) \\
= & \quad \cmt{Defintion of $F^n$: $F(F^m(x)) = F^{m+1}(x)$} \\
& \thsup{$n \geq 0$} \left( F^{n+1}(\bot)\right) \\
= & \quad \cmt{Rewrite index} \\
& \thsup{$m \geq 1$} \left( F^{m}(\bot)\right) \\
= & \quad \cmt{Law~\ref{law:sup_iff} and $\bot$ is the least element} \\
& \bot \sqcup \left(\thsup{$m \geq 1$} \left( F^{m}(\bot)\right)\right) \\
= & \quad \cmt{Definition of $F^0$: $F^0(\bot)=\bot$} \\
& F^{0}(\bot) \sqcup \left(\thsup{$m \geq 1$} \left( F^{m}(\bot)\right)\right) \\
= & \quad \cmt{Definition of $\thsup{$m \geq 0$}$} \\
& \thsup{$m \geq 0$} \left( F^{m}(\bot)\right) \\
= & \quad \cmt{Rewite index} \\
& \thsup{$n \geq 0$} \left( F^{n}(\bot)\right) \\
= & \quad \cmt{Definition~\ref{def:kleene_lfp}} \\
& \thlfp{}~F
\end{align*}
So $\thlfp{}~F$ is a fixed point of $F$.

Suppose $fb$ is also a fixed point of $F$ that is, $F (fb) = fb$.
\begin{align*}
& \quad \cmt{$\bot$ is the least element} \\
& \bot \leq fb \\
\implies & \quad \cmt{$F$ is continuous and so is monotonic by Theorem~\ref{thm:cont_mono}} \\
& F (\bot) \leq F (fb) \\
\implies & \quad \cmt{$F(fb)=fb$} \\
& F (\bot) \leq fb \\
\implies & \quad \cmt{$F^{2} (\bot) = F\left(F(\bot)\right)$ and $F$ is monotonic} \\
& F^{2} (\bot) \leq F (fb)\\
\implies & \quad \cmt{$F(fb)=fb$} \\
& F^2 (\bot) \leq fb \\
& ... \\
\implies & \quad \cmt{Induction} \\
& F^{n} (\bot) \leq fb \\
\implies & \quad \cmt{Law~\ref{law:Sup_least}} \\ 
& \thsup{$n \geq 0$}\left(F^n(\bot)\right) \leq fb \\
= & \quad \cmt{Definition~\ref{def:kleene_lfp}} \\
& \thlfp{}~F \leq fb
\end{align*}
%
So $\thlfp{}~F$ is the least fixed point.

Similarly, we prove $\thgfp{}~F$ is a fixed point of $F$ and is also the greatest.
\end{proof}

\subsection{Summation over topological space}
Summation considered in this paper could range over an infinite set, called infinite sums. We use corresponding theories in Isabelle/HOL to deal with convergence and infinite sums. 

We say a function $f$ is \emph{summable} on a (potentially infinite) set $A$, denoted as $\summable(f, A)$ if the sum of $f$ on $A$ converges to a particular value. The convergence is expressed as the existence of a limit of $f$ over finite subsets $B$ of $A$ when $B$ is approaching $A$. In Isabelle/HOL, the limit is generalised to arbitrary topological space using filters~\cite{Hoelzl2013}. Its definition is parametrised over two filters. To the infinite sums, they are the open neighbourhood filter, interpreted as ``for all points in some open neighbourhood of a point'' and the subset inclusion ordered at-top filter, interpreted as ``for sufficiently large finite subsets when it approaches its top $A$''. The infinite sums of $f$ over $A$, denoted as $\infsum x \in A @ f(x)$, is the limit if $\summable(f, A)$ and 0 otherwise. The definitions of summable and finite sums can be found in Isabelle/HOL.

We list some laws of summation below.
\begin{thm}
    \label{thm:summation}
    \begin{align*}
        & c \neq 0 \land \summable(f, A) \implies \summable\left(\lambda x @ f(x)/c, A\right) \tag*{(division by constant summable)} \label{thm:summation_cdiv_summable} \\ 
        & c \neq 0 \land \summable(f, A) \implies \infsum x \in A @ f(x)/c = \left(\infsum x \in A @ f(x)\right)/c  \tag*{(division by constant)} \label{thm:summation_cdiv} \\ 
        & \summable(f, A) \implies \summable\left(\lambda x @ f(x)*c, A\right) \tag*{(multiplication of constant summable)} \label{thm:summation_cmult_right_summable} \\ 
        & \summable(f, A) \implies \infsum x \in A @ f(x)*c = \left(\infsum x \in A @ f(x)\right)*c \tag*{(multiplication of constant)} \label{thm:summation_cmult_right} \\ 
        & \summable(f, A) \implies \summable\left(\lambda x @ c*f(x), A\right) \tag*{(multiplication of constant summable)} \label{thm:summation_cmult_left_summable} \\ 
        & \summable(f, A) \implies \infsum x \in A @ c*f(x) = c*\left(\infsum x \in A @ f(x)\right) \tag*{(multiplication of constant)} \label{thm:summation_cmult_left} \\ 
        & \summable(f, A) \land \summable(g, A) \implies \summable\left(\lambda x @ f(x)+g(x), A\right) \tag*{(addition summable)} \label{thm:summation_add_summable} \\ 
        & \summable(f, A) \land \summable(g, A) \implies \infsum x \in A @ f(x)+g(x) = \infsum x \in A @ f(x) + \infsum x \in A @ g(x) \tag*{(addition)} \label{thm:summation_add} \\ 
        & \summable(f, A) \land \summable(g, A) \implies \infsum x \in A @ f(x)-g(x) = \infsum x \in A @ f(x) - \infsum x \in A @ g(x) \tag*{(subtraction)} \label{thm:summation_minus} 
    \end{align*}
\end{thm}

\section{Unit real interval (complete lattice)}
\label{sec:ureal}
Our probabilistic programs are real-valued functions over state space, and specifically, they are the functions from state space to real numbers between 0 and 1 inclusive, or the \emph{unit real interval} ($\ureal$). We call them $\ureal$-valued functions, denoted as $S \fun \ureal$. To deal with the semantics of probabilistic loops in Section~\ref{sec:rec} using the Knaster–Tarski and Kleene fixed-point theorems~\cite{Tarski1955}, we define a complete lattice containing a set of these functions together with a pointwise comparison relation $\leq$.

Section~\ref{ssec:ureal_def} defines $\ureal$ and constructs a complete lattice containing the set $\ureal$ with relation $\leq$. Then we define $\ureal$-valued functions and the pointwise comparison relation in Section~\ref{ssec:ureal_functions}. With these definitions, we construct the required complete lattice for characterising probabilistic loops. 

\subsection{Definition of \texorpdfstring{$\ureal$}{ureal}}
\label{ssec:ureal_def}
The $\ureal$ is defined below as a set of real numbers between $0$ and $1$. 
\begin{definition}[Unit real interval]
        $\ureal \defs \{0 \upto 1\}$
    $ $\isalink{https://github.com/RandallYe/probabilistic_programming_utp/blob/6a4419b8674b84988065a58696f15093d176594c/probability/probabilistic_relations/utp_prob_rel_lattice.thy\#L17}
\end{definition}

We also define two functions to get the smaller and larger value of two comparable numbers, such as real numbers and $ureal$ numbers. 
\begin{definition}[Maximum and minimum of real numbers]
    \label{def:min_max}
   \begin{align*}
       & \umax\left(x, y\right) \defs \left(\IF x \leq y \THEN y \ELSE x\right) \qquad 
       \umin\left(x, y\right) \defs \left(\IF x \leq y \THEN x \ELSE y\right)
   \end{align*}
\end{definition}

The two functions $\umin$ and $\umax$ entitle us to define conversions between real numbers and $\ureal$ numbers.
\begin{definition}[Conversion between $\ureal$ and $\real$]
    \label{def:u2r_r2u}
    We define functions $\urealreal$ (\isaref{https://github.com/RandallYe/probabilistic_programming_utp/blob/6a4419b8674b84988065a58696f15093d176594c/probability/probabilistic_relations/utp_prob_rel_lattice.thy\#L30}) and $\realureal$ (\isaref{https://github.com/RandallYe/probabilistic_programming_utp/blob/6a4419b8674b84988065a58696f15093d176594c/probability/probabilistic_relations/utp_prob_rel_lattice.thy\#L27}) (notations $\ur{x}$ and $\ru{y}$) to convert $x$ of $\ureal$ to $\real$, and $y$ of $\real$ to $\ureal$.
    \begin{align*}
        & \ur{x} \defs \left(x :: \real\right) \qquad
        \ru{y} \defs \umin\left(\umax\left(0, y\right), 1\right)
    \end{align*}
\end{definition}
The conversion of a $\ureal$ number $x$ to a real number, using the function $\urealreal$, is simply a type cast from $\ureal$ to $\real$. However, the conversion of a real number $y$ to $\ureal$ by the function $\realureal$ needs to deal with the cases when $y$ is out of the unit interval. We use $\umin$ and $\umax$ to bound it to 0 or 1 in these cases and keep its value if $y$ is between 0 and 1. {Based on the conversions, we define the comparison functions over $\ureal$.

\begin{definition}[Comparison functions of $\ureal$]
    Provided both $x$ and $y$ are of type $\ureal$.
\isalink{https://github.com/RandallYe/probabilistic_programming_utp/blob/6a4419b8674b84988065a58696f15093d176594c/probability/probabilistic_relations/utp_prob_rel_lattice.thy\#L46}
    \label{def:ureal_funcs}
    \begin{align*}
        & x = y \defs \ur{x} = \ur{y}\qquad  
        x < y \defs \ur{x} < \ur{y}\qquad  
        x \leq y \defs \ur{x} \leq \ur{y} \\
       & \umax\left(x, y\right) \defs \left(\IF x \leq y \THEN y \ELSE x\right) \qquad 
       \umin\left(x, y\right) \defs \left(\IF x \leq y \THEN x \ELSE y\right)
    \end{align*}
\end{definition}
From these comparisons, we show both conversion functions are monotonic.
}

\begin{lem}
    The function $\urealreal$ is strictly monotonic (\isaref{https://github.com/RandallYe/probabilistic_programming_utp/blob/6a4419b8674b84988065a58696f15093d176594c/probability/probabilistic_relations/utp_prob_rel_lattice_laws.thy\#L42}). That is, if $x < y$, then $\ur{x} < \ur{y}$. The function $\realureal$ is monotonic (\isaref{https://github.com/RandallYe/probabilistic_programming_utp/blob/6a4419b8674b84988065a58696f15093d176594c/probability/probabilistic_relations/utp_prob_rel_lattice_laws.thy\#L49}), but not strictly. 
    For example, $\ru{2} = \ru{3}$ (both equal to $\uone$) though $2 < 3$.
\end{lem}

The function $\realureal$ is the inverse of $\urealreal$. 
\begin{lem}
    $\ru{\left(\ur{x}\right)} = x$ \isalink{https://github.com/RandallYe/probabilistic_programming_utp/blob/6a4419b8674b84988065a58696f15093d176594c/probability/probabilistic_relations/utp_prob_rel_lattice_laws.thy\#L330}
\end{lem}

The function $\urealreal$ is the inverse of $\realureal$ only if the real number to be converted is between 0 and 1.
\begin{lem}
    $\left(x \geq 0 \land x \leq 1\right) \implies \ur{\left(\ru{x}\right)} = x$ 
    \isalink{https://github.com/RandallYe/probabilistic_programming_utp/blob/6a4419b8674b84988065a58696f15093d176594c/probability/probabilistic_relations/utp_prob_rel_lattice_laws.thy\#L324}
\end{lem}

Infimum and supremum of $\ureal$ are defined using Hilbert's $\varepsilon$ operator, {an indefinite description, written $\varepsilon~x \bullet P(x)$ denoting some x such that P(x) is true.} {We note that $\varepsilon$ used below denotes the Hilbert's operator and a real number elsewhere in the paper.}
\begin{definition}[Infimum and supremum of $\ureal$]
    \begin{align*}
       & \thninf A \defs \left(\varepsilon~x \bullet \left(\forall y \in S @ x \leq y \right) \land \left(\forall z \bullet \left(\forall y\in A \bullet z \leq y\right) \implies z \leq x \right)\right)  \\
       & \thnsup A \defs \left(\varepsilon~x \bullet \left(\forall y \in S @ y \leq x \right) \land \left(\forall z \bullet \left(\forall y\in A \bullet y \leq z\right) \implies x \leq z \right)\right) 
    \end{align*}
\end{definition}
The infimum satisfies Laws~\ref{law:Inf_lower} and \ref{law:Inf_greatest}, and the supremum satisfies Laws~\ref{law:Sup_upper} and \ref{law:Sup_least}.

A complete lattice is now formed using these definitions.
\begin{thm}
    The poset $\left(\ureal, \leq\right)$ with the least element $\uzero$ $\left(\ru{0}\right)$, the greatest element $\uone$ $\left(\ru{1}\right)$, $\tinf$ $\left(\umin\right)$, $\tsup$ $\left(\umax\right)$, $\thninf$, and $\thnsup$ forms a complete lattice $\left(\ureal, \leq, <, \uzero, \uone, \tinf, \tsup, \thninf, \thnsup \right)$.
    $ $\isalink{https://github.com/RandallYe/probabilistic_programming_utp/blob/6a4419b8674b84988065a58696f15093d176594c/probability/probabilistic_relations/utp_prob_rel_lattice.thy\#L51}
\end{thm}

This complete lattice is illustrated in the left diagram of Fig.~\ref{fig:ureal_complete_lattices}. Indeed, it is a totally ordered set.
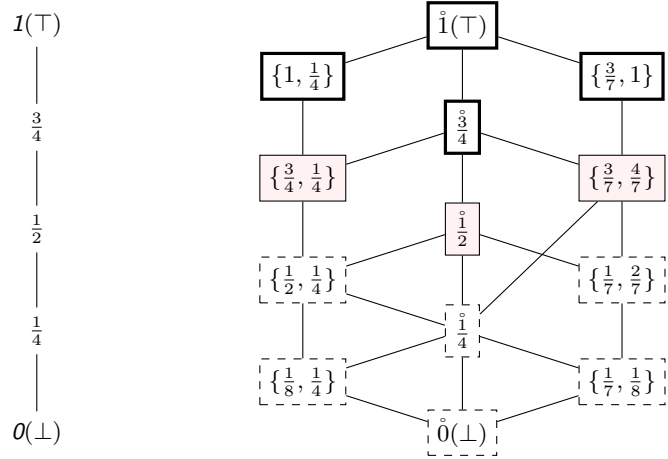
\begin{figure}[!ht]
    \begin{center}
\newcommand\SC{0.75}
\newcommand\SCC{0.75}
\begin{tikzpicture}[x=1.4cm,y=1.8cm,
        roundnode/.style={circle, draw=green!60, fill=green!5, very thick, minimum size=7mm},
        distnode/.style={rectangle, draw=black, fill=red!5},
        supdistnode/.style={rectangle, draw=black, very thick},
        subdistnode/.style={rectangle, draw=black, dashed},
    ]
\node at (0,0)    (UB)  {\small$\uzero(\bot)$};
\node at (0,1*\SC)    (U1)  {\small$\frac{1}{4}$};
\node at (0,2*\SC)   (U2)  {\small$\frac{1}{2}$};
\node at (0,3*\SC)   (U3)  {\small$\frac{3}{4}$};
\node at (0,4*\SC)   (UT)  {\small$\uone(\top)$};
\draw (UB) -- (U1) -- (U2) -- (U3) -- (UT);

\node[subdistnode] at (4,0)    (F0)  {\small$\ufzero(\bot)$};
\node[subdistnode] at (4,1*\SCC)    (F14)  {\small$\mathring{\frac{1}{4}}$};
\node[distnode] at (4,2*\SCC)    (F12)  {\small$\mathring{\frac{1}{2}}$};
\node[supdistnode] at (4,3*\SCC)    (F34)  {\small$\mathring{\frac{3}{4}}$};
\node[supdistnode] at (4,4*\SCC)    (F1)  {\small$\ufone(\top)$};

\node[subdistnode] at (2.5,0.5*\SCC)    (F1814)  {\small$\{ {\frac{1}{8}}, {\frac{1}{4}}\}$};
\node[subdistnode] at (2.5,1.5*\SCC)    (F1214)  {\small$\{ {\frac{1}{2}}, {\frac{1}{4}}\}$};
\node[distnode] at (2.5,2.5*\SCC)    (F3414)  {\small$\{ {\frac{3}{4}}, {\frac{1}{4}}\}$};
\node[supdistnode] at (2.5,3.5*\SCC)    (F0114)  {\small$\{ 1, {\frac{1}{4}}\}$};

\node[subdistnode] at (5.5,0.5*\SCC)    (F1718)  {\small$\{ {\frac{1}{7}}, {\frac{1}{8}}\}$};
\node[subdistnode] at (5.5,1.5*\SCC)    (F1727)  {\small$\{ {\frac{1}{7}}, {\frac{2}{7}}\}$};
\node[distnode] at (5.5,2.5*\SCC)    (F3747)  {\small$\{ {\frac{3}{7}}, {\frac{4}{7}}\}$};
\node[supdistnode] at (5.5,3.5*\SCC)    (F3701)  {\small$\{ {\frac{3}{7}}, 1\}$};

\draw (F0) -- (F14) -- (F12) -- (F34) -- (F1);
\draw (F0) -- (F1814) -- (F14) -- (F1214) -- (F3414) -- (F0114) -- (F1);
\draw (F1214) -- (F12);
\draw (F1814) -- (F1214);
\draw (F3414) -- (F34);
\draw (F0) -- (F1718) -- (F1727) -- (F3747) -- (F3701) -- (F1);
\draw (F1718) -- (F14);
\draw (F1727) -- (F12);
\draw (F3747) -- (F34);
\draw (F14) -- (F3747);
\end{tikzpicture}
    \end{center}
    \caption{Complete lattices: left $\left(\ureal, \leq\right)$ and right $\left(S \fun \ureal, \leq\right)$. {We use $\left\{{\frac{3}{4}}, {\frac{1}{4}}\right\}$ denotes a function $\left\{s_1 \mapsto {\frac{3}{4}}, s_2 \mapsto {\frac{1}{4}}\right\}$ whose domain contains two elements $s_1$ and $s_2$ and their corresponding probabilities are ${\frac{3}{4}}$ and ${\frac{1}{4}}$ respectively. {\small$\mathring{\frac{1}{4}}$} denotes a constant function which maps every element in its domain to ${\frac{1}{4}}$.} Dashed box: subdistributions; normal: distributions; thick: superdistributions.}
    \label{fig:ureal_complete_lattices}
\end{figure}

The addition, real numbers' subtraction, and multiplication operators are lifted for $\ureal$. 
\begin{definition}[Bounded plus and minus]
    \label{def:bounded_plus_minus}
    Provided both $x$ and $y$ are of type $\ureal$.
    $ $\isalink{https://github.com/RandallYe/probabilistic_programming_utp/blob/6a4419b8674b84988065a58696f15093d176594c/probability/probabilistic_relations/utp_prob_rel_lattice.thy\#L95}
    \begin{align*}
        & x + y \defs \ru{\left(\umin\left(1, \ur{x} + \ur{y}\right)\right)} \qquad 
        x - y \defs \ru{\left(\umax\left(0, \ur{x} - \ur{y}\right)\right)} \qquad
        x * y \defs \ru{\left(\ur{x} * \ur{y}\right)} 
    \end{align*}
\end{definition}
The addition $+$ and subtraction $-$ are bounded to $\ureal$ by using $\umin$ and $\umax$. For example, $0.5 + 0.7 = 1$ and $0.5 - 0.7 = 0$.

\subsection{The \texorpdfstring{$\ureal$}{ureal}-valued functions}
\label{ssec:ureal_functions}
We now consider $\ureal$-valued functions and define several constant functions for real-valued and $\ureal$-valued. 
\begin{definition}[Real- and $\ureal$-valued constant functions]
    \label{def:urf_const}
    $ $ \isalink{https://github.com/RandallYe/probabilistic_programming_utp/blob/6a4419b8674b84988065a58696f15093d176594c/probability/probabilistic_relations/utp_prob_rel_lattice.thy\#L207}
    \begin{align*}
        & \rfzero \defs \lambda s @ (0{::\real})  \qquad 
         \rfone \defs \lambda s @ (1::\real) \qquad  
         \ufzero \defs \lambda s @ (0::\ureal)\qquad  
         \ufone \defs \lambda s @ (1::\ureal)
    \end{align*}
\end{definition}
The $\rfzero$ and $\rfone$ are real-valued constant functions, and $\ufzero$ and $\ufone$ are $\ureal$-valued constant functions.

The addition and subtraction operators are also lifted to functions in a pointwise manner, {and relations $\leq$ and $<$ are also lifted to functions.}
\begin{definition}[Operators and relations on functions]
    \label{def:urf_pointwise}
    \begin{align*}
        & f - g \defs \left(\lambda x \bullet f(x) - g(x)\right) \qquad  
         f + g \defs \left(\lambda x \bullet f(x) + g(x)\right) \\
        & f \leq g \defs \left(\forall x \bullet f(x) \leq g(x)\right) \qquad 
         f < g \defs \left(\forall x \bullet f(x) < g(x)\right) 
    \end{align*}
\end{definition}

The complete lattice for $\ureal$-valued functions is formed.
\begin{thm}
    \label{thm:ureal_func_complete}
    The poset $\left(S \fun \ureal, \leq\right)$ with the least element $\ufzero$, the greatest element $\ufone$, the infimum and supremum in a pointwise manner forms a complete lattice $\left(S \fun \ureal, \leq, <, \ufzero, \ufone, \tinf, \tsup, \thninf, \thnsup\right)$. \isalink{https://isabelle.in.tum.de/library/HOL/HOL/Complete_Lattices.html}
\end{thm}

We illustrate an example of the complete lattice $\left(S \fun \ureal, \leq\right)$ in the right diagram of Fig.~\ref{fig:ureal_complete_lattices}. Here we consider $S$ containing two elements and use $\{\frac{1}{2}, \frac{1}{4}\}$ to denote probabilities over $S$: $\frac{1}{2}$ and $\frac{1}{4}$ respectively. Constant functions such as $\mathring{\frac{1}{2}}$ are on the central column. In the diagram, we only show a few functions where a dashed box, a normal box, or a thick box denotes a subdistribution, a distribution, or a superdistribution whose probabilities sum to less than or equal to 1, equal to 1, or larger than 1.

\section{Probabilistic programming}
\label{sec:program}
This section concerns our probabilistic programming language's syntax and denotational semantics. Probabilistic recursion is not considered here, and its syntax and semantics will be introduced in Sect.~\ref{sec:rec}. 

Before presenting the semantics, we define a notation of Iverson brackets in Sect.~\ref{ssec:prob_iverson} and introduce various expression types used in our language to constrain programs in Sect.~\ref{prob:type_abbr}. Our probabilistic programs are functions characterised as probabilistic distributions or subdistributions in Sect.~\ref{prob:dist_functs}. Our definition of Iverson brackets is real-valued functions, but probabilistic programs are $\ureal$-valued functions. We must convert between these functions to use Iverson brackets in our semantics. The conversion is defined in Sect.~\ref{prob:type_abbr}.

After the presentation of syntax and semantics, we show a collection of proved algebraic laws for each construct in Sects.~\ref{ssec:prog_top_bot} to \ref{ssec:prog_parallel_comp}. These laws are used in compositional reasoning to simplify probabilistic programs.

\subsection{Iverson brackets}
\label{ssec:prob_iverson}
Iverson brackets establish a correspondence between the predicate calculus and arithmetic, generalising the Kronecker delta.\footnote{Iverson brackets are a notation for the characteristic function on predicates. The convention was invented by Kenneth Eugene Iverson in 1962. Donald Knuth advocated using square brackets to avoid ambiguity in parenthesised logical expressions.}
\begin{definition}[Iverson bracket]
    The Iverson bracket of a predicate $P$ of type $[S]\upred$ defines a function $S \fun \real$, which gives a real number 0 or 1 if $P$ is false or true for a particular state $s$ (of type $S$). \isalink{https://github.com/RandallYe/probabilistic_programming_utp/blob/6a4419b8674b84988065a58696f15093d176594c/probability/probabilistic_relations/utp_iverson_bracket.thy\#L45}
   \begin{align*}
       & \ibracket{P} \defs \usexpr{\IF P \THEN 1 \ELSE 0}
   \end{align*}
\end{definition}

Several laws follow immediately from this definition.
\begin{thm}
    \label{thm:ib}
\isalink{https://github.com/RandallYe/probabilistic_programming_utp/blob/6a4419b8674b84988065a58696f15093d176594c/probability/probabilistic_relations/utp_iverson_bracket.thy\#L63}
    \begin{align}
        \ibracket{\pfalse} & = \rfzero \label{thm:ib_false} \\
        \ibracket{\ptrue} & = \rfone \label{thm:ib_true} \\
        Q \refinedby P & \implies \ibracket{P} \leq \ibracket{Q} \label{thm:ib_monotone} \\ 
        \ibracket{\lnot P} & = \usexpr{1 - \ibracket{P}} \label{thm:ib_neg}\\
        \ibracket{P \land Q} & = \usexpr{\ibracket{P} * \ibracket{Q}} \label{thm:ib_conj}  \\
        \ibracket{P \lor Q} & = \usexpr{\ibracket{P} + \ibracket{Q} - \ibracket{P} * \ibracket{Q}} \label{thm:ib_disj}\\
        \ibracket{\lambda s@ s \in A \cap B} & = \usexpr{\ibracket{\lambda s@ s \in A} * \ibracket{\lambda s@ s \in B} } \label{thm:ib_inter}\\
        \usexpr{\ibracket{\lambda s@ s \in A} + \ibracket{\lambda s@ s \in B}} & = \usexpr{\ibracket{\lambda s@ s \in A \cap B} + \ibracket{\lambda s@ s \in A \cup B}} \label{thm:ib_plus}\\ 
        \usexpr{\umax\left(x, y\right)} & = \usexpr{x * \ibracket{x > y} + y * \ibracket{x \leq y}} \label{thm:ib_max} \\
        \usexpr{\umin\left(x, y\right)} & = \usexpr{x * \ibracket{x \leq y} + y * \ibracket{x > y}} \label{thm:ib_min} \\
        {\sum_{P(k)} f(k)} & = {\sum_{k} \usexpr{f*\ibracket{P}} (k)} \label{thm:ib_summation} 
    \end{align}
\end{thm}

Laws~(\ref{thm:ib_false}) and (\ref{thm:ib_true}) show the arithmetic representations of {UTP predicates $\pfalse$ and $\ptrue$ of type $[S]\upred$} are simply constant functions $\rfzero$ and $\rfone$. Iverson brackets are monotone, as shown in Law~(\ref{thm:ib_monotone}). 
Laws~(\ref{thm:ib_neg}) to (\ref{thm:ib_plus}) establish direct correspondence between arithmetic, logic, and set operations. Laws~(\ref{thm:ib_max}) and (\ref{thm:ib_min}) show the maximum and minimum operations that can be implemented using the Iverson bracket. Law~(\ref{thm:ib_summation}) shows summation over a subset of indices characterised by $P(k)$ can be expressed as a summation over whole indices with the summation function $f(x)$ multiplied by the Iverson bracket of $P$. According to Donald E. Knuth~\cite{Knuth1992}, it is not easy to make a mistake when dealing with summation indices by using the notation of the right side of the law. 
We omit other properties of Iverson brackets here for simplicity.

\subsection{Type Abbreviations}
\label{prob:type_abbr}
We define several type abbreviations for real-valued and $\ureal$-valued functions used to type constructs in our language.
 \isalink{https://github.com/RandallYe/probabilistic_programming_utp/blob/6a4419b8674b84988065a58696f15093d176594c/probability/probabilistic_relations/utp_prob_rel_lattice.thy\#L175}
\begin{align*}
    [S] \rexpr & = [\real, S] \uexpr \tag*{(Real-valued expression)} \label{typeabb:rexpr}\\
    [S_1, S_2] \rvfun & = [\real, S_1 \times S_2] \uexpr \tag*{(Relational real-valued expression)} \label{typeabb:rvfun}\\
    [S] \rvhfun & = [S, S] \rvfun \tag*{(Homogeneous relational real-valued expression)} \label{typeabb:rvhfun}\\
    [S] \urexpr & = [\ureal, S] \uexpr \tag*{($\ureal$-valued expression)} \label{typeabb:urexpr}\\
    [S_1, S_2] \prfun & = [\ureal, S_1 \times S_2] \urexpr \tag*{(Relational $\ureal$-valued expression)} \label{typeabb:prfun}\\
    [S] \prhfun & = [S, S] \prfun \tag*{(Homogeneous relational $\ureal$-valued expression)} \label{typeabb:prhfun}
\end{align*}

We define two functions $\prrvfun$ (\isaref{https://github.com/RandallYe/probabilistic_programming_utp/blob/6a4419b8674b84988065a58696f15093d176594c/probability/probabilistic_relations/utp_prob_rel_lattice.thy\#L184}) and $\rvprfun$ (\isaref{https://github.com/RandallYe/probabilistic_programming_utp/blob/6a4419b8674b84988065a58696f15093d176594c/probability/probabilistic_relations/utp_prob_rel_lattice.thy\#L181}) to convert $P$ of type $[S_1, S_2] \prfun$ to an expression of type $[S_1, S_2] \rvfun$, and $f$ of type $[S_1, S_2] \rvfun$ to an expression of type $[S_1, S_2] \prfun$.
\begin{definition}[Conversion of relational real-valued and $\ureal$-valued functions]
    \label{def:rvfun2prfun}
    \begin{align*}
        & \prrvfun(P) \defs \usexpr{\ur{P}} \tag*{(Probabilistic programs to real-valued functions)} \label{def:rvfun_of_prfun}\\
        & \rvprfun(f) \defs \usexpr{\ru{f}} \tag*{(Real-valued functions to probabilistic programs)} \label{def:prfun_of_rvfun}
    \end{align*}
\end{definition}
The notations $\ur{p}$ and $\ru{r}$ (Definition~\ref{def:u2r_r2u}) used ablove convert a $\ureal$ number to a real number and a real number to a $\ureal$ number, respectively. In this paper, we also use symbols $\prrvfunsym{P}$ and $\rvprfunsym{f}$ for $\prrvfun(P)$ and $\rvprfun(f)$, the conversions of functions.

\begin{rmk}
   For two reasons, we define two types $\prfun$ and $\rvfun$.
   First, infinite summation and limits are defined over topological space, and so over real numbers, which form a Banach space, but not over $\ureal$. We, therefore, need to convert probabilistic programs into real-valued functions to calculate summation and limits. After calculation, the results are converted back to probabilistic programs. Second, two functions in parallel composition or joint probability, introduced later in Sect.~\ref{ssec:syntax_semantics}, are not necessary to be probabilistic, and they can be more general real-valued functions.
\end{rmk}

\subsection{Distribution functions}
\label{prob:dist_functs}

A real-valued expression $p$ is \emph{nonnegative} if its range is real numbers larger than or equal to 0. 
\begin{definition}[Nonnegative] \label{def:nonneg} 
    \isalink{https://github.com/RandallYe/probabilistic_programming_utp/blob/6a4419b8674b84988065a58696f15093d176594c/probability/probabilistic_relations/utp_distribution.thy\#L15}
\begin{align*}
    & \isnonneg(p) \defs \tautology{p \geq 0}
\end{align*}
where $\tautology{p}$ is a tautology on predicate $p$ and is expanded to $\forall s @ \usexpr{p}(s)$.
\end{definition}

A real-valued expression $p$ is \emph{probabilistic} if its range is real numbers between 0 and 1 inclusive. This expression is characterised by a function $\isprob$ of type $[S]\rexpr \fun \bool$.

\begin{definition}[Probability expression] \label{def:isprob} 
    \isalink{https://github.com/RandallYe/probabilistic_programming_utp/blob/6a4419b8674b84988065a58696f15093d176594c/probability/probabilistic_relations/utp_distribution.thy\#L18}
\begin{align*}
    & \isprob(p) \defs \tautology{p \geq 0 \land p \leq 1}
\end{align*}
\end{definition}

\begin{thm}[Iverson bracket is probabilistic]
    \label{thm:isprob_ibracket}
    $\isprob(\ibracket{p})$ 
    \isalink{https://github.com/RandallYe/probabilistic_programming_utp/blob/6a4419b8674b84988065a58696f15093d176594c/probability/probabilistic_relations/utp_distribution.thy\#L76}
\end{thm}

A probabilistic function is called a \emph{distribution} function if the probabilities of all states sum to 1, which is characterised by a function $\isdist$ of type $[S]\rexpr \fun \bool$.
\begin{definition}[Probabilistic distributions] \label{def:isdist} 
    \isalink{https://github.com/RandallYe/probabilistic_programming_utp/blob/6a4419b8674b84988065a58696f15093d176594c/probability/probabilistic_relations/utp_distribution.thy\#L29}
    \begin{align*}
        & \isdist(p) \defs \isprob(p) \land \infsum s @ p(s) = 1
    \end{align*}
\end{definition}
where $\infsum$ denotes a summation over possible infinite states. 

A probabilistic function is called a \emph{subdistribution} function if the probabilities of all states sum to less than or equal to 1, which is characterised by a function $\issubdist$.
\begin{definition}[Probabilistic subdistributions] \label{def:issubdist} 
    \isalink{https://github.com/RandallYe/probabilistic_programming_utp/blob/6a4419b8674b84988065a58696f15093d176594c/probability/probabilistic_relations/utp_distribution.thy\#L32}
    \begin{align*}
        \issubdist(p) \defs \isprob(p) \land \infsum s @ p(s) > 0 \land \infsum s @ p(s) \leq 1
    \end{align*}
\end{definition}
We note that a probabilistic distribution is also a subdistribution {but $\rfzero$ (probabilities are zero everywhere) is not. We exclude $\rfzero$ in subdistributions because 
    \begin{enumerate*}[label={(\arabic*)}]
        \item $\infsum$ in Isabelle/HOL is defined to be 0 when $p$ is not summable or divergent, and so we cannot differentiate this case from $\rfzero$ from the summation result; and 
        \item the probability summation is the denominator in the definitions of normalisation in Definitions~\ref{def:norm} and \ref{def:normf}, and so subdistributions allow us to characterise the non-zero result to deal with the division-by-zero error. 
    \end{enumerate*}
}
\begin{lem}
   $\isdist(p) \implies \issubdist(p)$ 
    \isalink{https://github.com/RandallYe/probabilistic_programming_utp/blob/6a4419b8674b84988065a58696f15093d176594c/probability/probabilistic_relations/utp_distribution.thy\#L79}
\end{lem}

For relational real-valued expressions $p$ of type $[S_1, S_2]\rvfun$, we define three functions to specify if the final state of a program is characterised by the expression $p$ is probabilistic, distributions, or subdistributions.
To specify these functions, we define a curried operator $\curried{p}$ ($\defs \lambda s~s' @ p(s, s')$) to turn $p$ into lambda terms, and so $\curried{p}(s)$ is a function from the final state $s'$ to real numbers. 
\begin{definition}[Final states are probabilistic, distributions, and subdistributions]
    \isalink{https://github.com/RandallYe/probabilistic_programming_utp/blob/6a4419b8674b84988065a58696f15093d176594c/probability/probabilistic_relations/utp_distribution.thy\#L36}
    \begin{align*}
        \isfinalprob(p) & \defs \tautology{\isprob(\curried{p})} \tag*{(final probabilistic)} \label{def:isfinalprob} \\
        \isfinaldist(p) & \defs \tautology{\isdist(\curried{p})} \tag*{(final distributions)} \label{def:isfinaldist}\\
        \isfinalsubdist(p) & \defs \tautology{\issubdist(\curried{p})} \tag*{(final subdistributions)} \label{def:isfinalsubdist}
    \end{align*}
\end{definition}
For all initial states $s$, if such a curried expression $\curried{p}$ is probabilistic, a distribution, or a subdistribution, then we say $p$ is probabilistic, a distribution, or a subdistribution over the final states, characterised by functions $\isfinalprob$, $\isfinaldist$, and $\isfinalsubdist$.

For an expression $P$ of type $[S_1, S_2]\prfun$, if $\isfinalprob(\prrvfunsym{P})$, we say $P$ is probabilistic. Similarly, we say $P$ is a distribution or a subdistribution (over its final states), if $\isfinaldist(\prrvfunsym{P})$ or $\isfinalsubdist(\prrvfunsym{P})$.

Using the function $\summable$, we introduce convergence for relational expressions and for the product of relational expressions over final states.

\begin{definition}[Summable on final states]
    \isalink{https://github.com/RandallYe/probabilistic_programming_utp/blob/6a4419b8674b84988065a58696f15093d176594c/probability/probabilistic_relations/utp_prob_rel_lattice.thy\#L188}
    \label{def:summable_on_finals}
    \begin{align}
        \summableonfinal (p) &\defs \left(\forall s @ \summable\left(\curried{p}(s), \univ\right) \right) \label{def:summable_on_final}\\
        \summableonfinals (p, q) &\defs \left(\forall s @ \summable\left(\lambda s' @ p(s, s') * q(s, s'), \univ\right)\right) \label{def:summable_on_final2}
    \end{align}
\end{definition}
The function $\summableonfinal$ characterises the relational expression $p(s)$ over final states are summable on the universe $\univ$ of state space for any initial state $s$. The function $\summableonfinals$ characterises the product of the expressions $p(s)$ and $q(s)$ over final states that are summable on $\univ$.

We also define functions below to characterise if the final states of a relational expression are reachable and the final states of two relational expressions are reachable at the same states.
\begin{definition}[Reachable final states]
    \isalink{https://github.com/RandallYe/probabilistic_programming_utp/blob/6a4419b8674b84988065a58696f15093d176594c/probability/probabilistic_relations/utp_prob_rel_lattice.thy\#L194}
    \begin{align}
        \finalreachable(p) & \defs \left(\forall s @ \exists s' @ p(s, s') > 0 \right) \label{def:final_reachable}\\
        \finalreachables(p, q) & \defs \left(\forall s @ \exists s' @ p(s, s') > 0 \land q(s, s') > 0 \right) \label{def:final_reachable2}
    \end{align}
\end{definition}
The final states of $p$ are reachable, $\finalreachable(p)$, if from any initial state $s$ there exists at least one final state $s'$ such that $p(s,s')$ is larger than 0, or reaching $s'$ from $s$ is possible.
The $\finalreachables(p,q)$ characterises the possibility for $p(s)$ and $q(s)$ to reach the same state $s'$ from any initial state $s$.

Convergence and reachability of a relational expression $p$ can be derived from whether $p$ is a distribution or subdistribution over its final states, as shown below.
\begin{thm}\label{thm:final_distribtion}
    \isalink{https://github.com/RandallYe/probabilistic_programming_utp/blob/6a4419b8674b84988065a58696f15093d176594c/probability/probabilistic_relations/utp_prob_rel_lattice_laws.thy\#L375}
    \begin{align}
        \isfinaldist(p) & \implies \left(
        \begin{array}[]{l}
            \isprob(p) \land \left(\forall s @ \infsum s' @ p(s, s') = 1\right) \land \\
            \summableonfinal(p)  \land \finalreachable(p)
        \end{array}
        \right) \label{thm:final_distribtion_1} \\
        \isfinalsubdist(p) & \implies \left(
        \begin{array}[]{l}
            \isprob(p) \land \left(\forall s @ \infsum s' @ p(s, s') > 0\right) \land \left(\forall s @ \infsum s' @ p(s, s') \leq 1\right) \\
            \summableonfinal(p)  \land \finalreachable(p)
        \end{array}
        \right) \label{thm:final_subdistribtion}
    \end{align}
\end{thm}
Law~\ref{thm:final_distribtion_1} shows if $p$ is a distribution over its final states, then $p$ is probabilistic, summable over its final states, and reachable. The second conjunct restates $p$ as a distribution over its final states.
Law~\ref{thm:final_subdistribtion} shows if $p$ is a subdistribution over the final states, then $p$ is probabilistic, the summation of $p$ over its final states is larger than 0, and less than or equal to 1, and $p$ is summable over its final states and reachable. The second and third conjuncts restate $p$ as a subdistribution over its final states.

Normalisation $\norm(p)$ is the distribution whose values are in the same proportion as the values of $p$. Here, $p$ is not required to be a distribution, but the result of normalisation {is a distribution.}

\begin{definition}[Normalisation]
    \label{def:norm}
    We fix $p$ of type $[S]\rexpr$,
    \isalink{https://github.com/RandallYe/probabilistic_programming_utp/blob/6a4419b8674b84988065a58696f15093d176594c/probability/probabilistic_relations/utp_distribution.thy\#L55}
\begin{align*}
    \norm(p) \defs \usexpr{p / \left(\infsum s: S @ p(s)\right)}
\end{align*}
\end{definition}
{In the definition, the division operator $/$ in Isabelle/HOL is implemented as $inverse\_divide$\footnote{\url{https://isabelle.in.tum.de/library/HOL/HOL/Fields.html}.} in the division ring. The result $a/b$ is 0 if either $a$ or $b$ is 0.}
The $\norm(p)$ gives the distribution of the state space (both the initial and final states for a relational expression) in $p$. For example, suppose that $x$ is in the $1..n$ range. 

\begin{example}
\begin{align*}
    & \norm\left(\ibracket{x'=x+1 \lor x'=x+2}\right) \\
   = & \cmt{~\text{Definition~\ref{def:norm}}~} \\
   & \usexpr{\ibracket{x'=x+1 \lor x'=x+2} / \left(\infsum (x,x'): (1..n)\cross (1..n) @     \ibracket{x'=x+1 \lor x'=x+2}\right)} \\ 
   = & \cmt{~\text{Theorem~\ref{thm:ib} Law~\ref{thm:ib_summation}}~} \\
   & \usexpr{\left(\ibracket{x'=x+1 \lor x'=x+2}\right) / \left(\infsum (x,x') @ \left(\ibracket{x'=x+1 \lor x'=x+2}\right)*\ibracket{(x,x') \in (1..n)\cross (1..n)}\right)} \\
   = & \cmt{~\text{Theorem~\ref{thm:ib} Law~\ref{thm:ib_disj}}~} \\
   & \usexpr{
       \begin{array}[]{l}
       \left(x'=x+1 \lor x'=x+2\right) / \\
       \left(\infsum (x,x') @ \left(
       \begin{array}[]{l}
        \ibracket{x'=x+1} + \ibracket{x'=x+2} - \\
        \ibracket{x'=x+1}*\ibracket{x'=x+2}
       \end{array}
    \right)*\ibracket{(x,x') \in (1..n)\cross (1..n)}\right)
       \end{array}
    } \\
   = & \cmt{~Theorem~\ref{thm:ib} Law~\ref{thm:ib_conj}, $\ibracket{x'=x+1}*\ibracket{x'=x+2}=\rfzero$, and summation distributes through sum~} \\
   & \usexpr{
   \begin{array}[]{l}
        \left(\ibracket{x'=x+1 \lor x'=x+2}\right) / \\
        \left(
       \begin{array}[]{l}
\infsum (x,x') @ \ibracket{x'=x+1}*\ibracket{(x,x') \in (1..n)\cross (1..n)} + \\
\infsum (x,x') @ \ibracket{x'=x+2}*\ibracket{(x,x') \in (1..n)\cross (1..n)}\\
       \end{array}
    \right)
       \end{array}
} \\
   = & \cmt{~Theorem~\ref{thm:ib} Law~\ref{thm:ib_conj} and Iverson bracket summation one-point rule~} \\
   & \usexpr{\left(\ibracket{x'=x+1 \lor x'=x+2}\right) / \left(
       \begin{array}[]{l}
\infsum x @ \ibracket{(x,x+1) \in (1..n)\cross (1..n)} + \\
\infsum x @ \ibracket{(x,x+2) \in (1..n)\cross (1..n)}\\
       \end{array}
    \right)} \\
   = & \cmt{~arithmetic~} \\
   & \usexpr{\left(\ibracket{x'=x+1 \lor x'=x+2}\right) / \left(
\infsum x @ \ibracket{x \in 1..n-1} + \infsum x @ \ibracket{x \in 1..n-2}
    \right)} \\
   = & \cmt{~Iverson summation~} \\
   & \usexpr{\left(\ibracket{x'=x+1 \lor x'=x+2}\right) / \left(
    n - 1 + n - 2
    \right)} \\
   = & \cmt{~arithmetic~} \\
   & \usexpr{\left({\ibracket{x'=x+1 \lor x'=x+2}}\right) / \left(2*n - 3\right)} 
\end{align*}

\end{example}

Often we want the distribution of just the final state if $p$ is a relational expression.

\begin{definition}[Normalisation of the final state]
    \label{def:normf}
    We fix $p$ of type $[S_1,S_2]\rvfun$,
    \isalink{https://github.com/RandallYe/probabilistic_programming_utp/blob/6a4419b8674b84988065a58696f15093d176594c/probability/probabilistic_relations/utp_distribution.thy\#L58}
\begin{align*}
    \normf(p) \defs \usexpr{p / \left(\infsum v_0: S_2 @ p[v_0/\vv']\right)}
\end{align*}
\end{definition}
The value $p(s,s')$ of $p$ for a pair $(s, s')$ of the initial state $s$ and the final state $s'$ is divided by the summation of $p(s, v_0)$ over the all final states $v_0$ for the initial state $s$.
The normalisation of $p$ is a distribution over its final states, given $p$ is nonnegative and reachable, and $\curried{p}(s)$ is convergent for any state $s$. 
\begin{thm}
    $\isnonneg(p) \land \finalreachable(p) \land \summableonfinal(p) \implies \isfinaldist \left(\normf(p)\right)$
    \isalink{https://github.com/RandallYe/probabilistic_programming_utp/blob/6a4419b8674b84988065a58696f15093d176594c/probability/probabilistic_relations/utp_prob_rel_lattice_laws.thy\#L2266}
\end{thm}

A probabilistic program $P$ is a relational $\ureal$-valued expression of type $[S_1, S_2] \prfun$. We show the conversion of $P$ to a real-valued function $\rvprfunsym{P}$ is probabilistic and pointwise subtraction of $\rvprfunsym{P}$ from the constant 1 function is also probabilistic.

\begin{thm}
    \label{thm:prrvfun_prob}
Provided $P$ is an expression of type $[S_1, S_2] \prfun$, then $\isprob\left(\prrvfunsym{P}\right)$ and $\isprob\left(\rfone - \prrvfunsym{P}\right)$.
\isalink{https://github.com/RandallYe/probabilistic_programming_utp/blob/6a4419b8674b84988065a58696f15093d176594c/probability/probabilistic_relations/utp_prob_rel_lattice_laws.thy\#L593}
\end{thm}

The conversion of $P$ to a real-valued function and then back to the $\ureal$-valued function is still $P$.
\begin{thm}
    \label{thm:rvprfun_inverse}
    $\rvprfun$ is the inverse of $\prrvfun$, that is, $\rvprfunsym{\left(\prrvfunsym{P}\right)} = P$. 
\isalink{https://github.com/RandallYe/probabilistic_programming_utp/blob/6a4419b8674b84988065a58696f15093d176594c/probability/probabilistic_relations/utp_prob_rel_lattice_laws.thy\#L613}
\end{thm}

The conversion of $p$ of type $[S_1, S_2] \prfun$ to a $\ureal$-valued function and then back to a {real-valued} function is still $p$ if $p$ is probabilistic. 
\begin{thm}
    \label{thm:prrvfun_inverse}
    $\prrvfun$ is the inverse of $\rvprfun$ if $p$ is a probabilistic, that is, $\isprob(p) \implies \prrvfunsym{\left(\rvprfunsym{p}\right)} = p$. 
\isalink{https://github.com/RandallYe/probabilistic_programming_utp/blob/6a4419b8674b84988065a58696f15093d176594c/probability/probabilistic_relations/utp_prob_rel_lattice_laws.thy\#L600}
\end{thm}

A corollary of this theorem, given below, states that the conversion of an Iverson bracket expression to $\prfun$ and then back to $\rvfun$ gives the expression itself because Iverson bracket expressions are probabilistic (Theorem~\ref{thm:isprob_ibracket}).
\begin{thm}
    \label{thm:prrvfun_inverse_ibracket} 
    $\prrvfunsym{\left(\rvprfunsym{\ibracket{p}}\right)} = \ibracket{p}$.
\isalink{https://github.com/RandallYe/probabilistic_programming_utp/blob/6a4419b8674b84988065a58696f15093d176594c/probability/probabilistic_relations/utp_prob_rel_lattice_laws.thy\#L621}
\end{thm}

\subsection{Syntax and semantics}
\label{ssec:syntax_semantics}

Our probabilistic programming language includes six constructs (except probabilistic recursions), and their semantics is given as follows.
\begin{definition}[Probabilistic programs]
    \label{def:prob_programs}
    {We define probabilistic programs, interpreted as $[S_1, S_2] \prfun$ (that is, $\ureal$-valued functions), constructed from the syntax below} where we fix $P$ and $Q$ as probabilistic programs of type $[S_1, S_2] \prfun$, $r$ as an expression of type $[S_1, S_2] \prfun$, $b$ as a relation of type $[S_1, S_2] \urel$, and $R$ and $T$ as relational real-valued expressions of type $[S_1, S_2]\rvfun$.
 \isalink{https://github.com/RandallYe/probabilistic_programming_utp/blob/6a4419b8674b84988065a58696f15093d176594c/probability/probabilistic_relations/utp_prob_rel_lattice.thy\#L225}
    \begin{align*}
        \pskip & \defs \rvprfunsym{\ibracket{\II}} \tag*{(skip)} \label{def:prog_skip} \\
        \left(\passign{x}{e}\right) & \defs \rvprfunsym{\ibracket{x := e}} \tag*{(assignment)} \label{def:prog_assign} \\
        \left(\ppchoice{r}{P}{Q}\right) & \defs \rvprfunsym{\usexpr{\prrvfunsym{r} * \prrvfunsym{P} + \left(\rfone - \prrvfunsym{r}\right) * \prrvfunsym{Q}}}\tag*{(probabilistic choice)} \label{def:prog_pchoice} \\ 
        \left(\pcchoice{b}{P}{Q}\right)& \defs \rvprfunsym{\usexpr{\IF b \THEN \prrvfunsym{P} \ELSE \prrvfunsym{Q}}}\tag*{(conditional choice)} \label{def:prog_condchoice} \\ 
        \pseq{P}{Q} & \defs \rvprfunsym{\fseq{\prrvfunsym{P}}{\prrvfunsym{Q}}} \mbox{ where }
        \fseq{R}{T} \defs {\usexpr{\infsum v_0 @ {R}[v_0/\vv'] * {T}[v_0/\vv]}} \tag*{(sequential composition)} \label{def:prog_seq} \\ 
        \pparallel{R}{T} & \defs \rvprfunsym{\normf \usexpr{R * T}}\tag*{(parallel composition)} \label{def:prog_parallel_ff}
    \end{align*}
\end{definition}
We note that the semantics of all these constructs are converting the corresponding real-valued expressions to the $\ureal$-valued expressions by $\rvprfun$. 

The probability skip $\pskip$ is the $\ureal$ version of the Iverson bracket of the relational skip $\II$. It changes no variables and terminates immediately. On termination, the final state equals the initial state with probability 1; all other assignments to the final state have probability 0.
Similarly, the probability assignment ($\passign{x}{e}$) is the $\ureal$ version of the Iverson bracket of the relational assignment ($x := e$). An assignment is a one-point distribution of the final state. 

The probabilistic choice $\left(\ppchoice{r}{P}{Q}\right)$, also denoted $\ppifchoice{r}{P}{Q}$, is the weighted sum of $\prrvfunsym{P}$ (the conversion of $P$ to the real-valued expression) and $\prrvfunsym{Q}$ based on their weights $\prrvfunsym{r}$ and $(\rfone-\prrvfunsym{r})$. Because of the type $[S_1, S_2] \prfun$ of $r$, both $\prrvfunsym{r}$ and $\rfone-\prrvfunsym{r}$ are probabilistic (or between 0 and 1) by Theorem~\ref{thm:prrvfun_prob}.

In the conditional choice $\left(\pcchoice{b}{P}{Q}\right)$, $b$ is a relation. If $b$ is evaluated to true, the choice is $\prrvfunsym{P}$. Otherwise, it is $\prrvfunsym{Q}$.

Sequential composition $(\pseq{P}{Q})$ is the serial composition of $\prrvfunsym{P}$ with $\prrvfunsym{Q}$ by $\fseq{}{}$ where both $P$ and $Q$ have the same type of $[S]\prhfun$. The $\fseq{R}{T}$, where both $R$ and $T$ are type of $[S]\prfun$, is the conditional probability of $T$ given $R$. Indeed, it is the summation of the product of $R[v_0/\vv']$, the substitution of $v_0$ for $\vv'$ in $R$, and $T[v_0/\vv]$, the substitution of $v_0$ for $\vv$ in $T$, over their intermediate states $v_0$. Its semantics can be interpreted as starting from an initial state $s$, the probability of $(\pseq{P}{Q})$ reaching a final state $s'$ is equal to the summation of the probability of $\prrvfunsym{Q}$ reaching state $s'$ from $v_0$, given the probability of $\prrvfunsym{P}$ reaching $v_0$ from $s$, over all intermediate states $v_0$. 

The $\pparallel{R}{T}$ is the parallel composition of ${R}$ with ${T}$, semantically as normalisation of the product of ${R}$ and ${T}$. It is the joint probability of $R$ and $T$. In its most general form, neither $R$ nor $T$ need to be proper probabilistic programs, but the result will be a probabilistic program.

In Bayesian inference, the posterior probability of $A$ given $B$ is computed based on a prior probability, which is estimated before $B$ is observed, a likelihood function over $B$ given fixed $A$, and model evidence $B$ according to Bayes' theorem given below.
\begin{align*}
    {\text{posterior} = \frac{\text{prior}*\text{likelihood}}{\text{evidence}}} \qquad \text{or}\qquad {\displaystyle P(A  \mid B)={\frac {P(A)P(B\mid A)}{P(B)}}} 
\end{align*}
where $B$ is a new observed data or evidence.
In our programming language, the update of the posterior probability is modelled using parallel composition for learning new facts and sequential composition for making actions, such as the movement of robots. This is illustrated in the forgetful Monty, the robot localisation, and the COVID diagnosis examples in Sects.~\ref{ssec:cases_forgetful_monty}, \ref{ssec:cases_robot_localisation}, and \ref{ssec:cancer_diagnosis}.

{In other work~\cite{Gordon2014,Kaminski2019a
}, this learning or conditioning is encoded using an \textbf{observe} statement such as $\textbf{observe}(\phi)$ where $\phi$ is a boolean expression or a predicate defined over program variables. This statement normalises all valid executions that satisfy $\phi$ with respect to the probability of total valid executions and blocks invalid executions (with probability 0). Comparatively, our parallel composition or Hehner's is not restricted to predicates. Indeed, $R$ and $T$ in $\pparallel{R}{T}$ can be any real-valued or $\ureal$-valued expressions, which are more general likelihood functions. To encode predicates similar to $\phi$ in the \textbf{observe} statement, we need to use $\ibracket{\phi}$ to convert it. As illustrated in the examples in Sects.~\ref{ssec:cases_forgetful_monty}, \ref{ssec:cases_robot_localisation}, and \ref{ssec:cancer_diagnosis}, three likelihood functions are ${{\ibracket{m' \neq p'}}}$, $\usexpr{3 * \ibracket{door(bel')} + 1}$, and $\ibracket{ct' = Pos}$.}

\subsection{Top and bottom}
\label{ssec:prog_top_bot}
According to Theorem~\ref{thm:ureal_func_complete}, the set of probabilistic programs is a complete lattice under $\leq$. The top and bottom elements of the lattice satisfy the properties below. 
\begin{thm}
    \label{thm:top_bot}
Provided that $P$ is a probabilistic program and $p$ is a real-valued function.
\isalink{https://github.com/RandallYe/probabilistic_programming_utp/blob/6a4419b8674b84988065a58696f15093d176594c/probability/probabilistic_relations/utp_prob_rel_lattice_laws.thy\#L906}
    \begin{align*}
        {\arraycolsep=10pt\def\arraystretch{1.0}
        \begin{array}[]{cccccccc}
        \top = \ufone & \bot = \ufzero & P \geq \ufzero & P \leq \ufone & \rvprfunsym{\rfone} = \ufone  & \rvprfunsym{\rfzero} = \ufzero & \prrvfunsym{\ufone} = \rfone  & \prrvfunsym{\ufzero} = \rfzero  \\
        p * \rfzero = \rfzero & p * \rfone = p & P * \ufzero = \ufzero & P * \ufone = P & \multicolumn{2}{c}{P - \ufzero = P} & \multicolumn{2}{c}{P + \ufzero = P}
        \end{array}
        }
    \end{align*}
\end{thm}
The top element $\top$ is $\ufone$ and the bottom $\bot$ is $\ufzero$. Any probabilistic program $P$ is between $\bot$ and $\top$. 
The constant $\ufone$ and $\rfone$ mutually correspond in $\ureal$-valued and real-valued functions, and they can be converted to each other. This is similar for $\ufzero$ and $\rfzero$.
Real-valued functions and $\ureal$-valued functions satisfy the right-zero and right-one laws. 

\subsection{Skip}
\label{ssec:prog_skip}
The $\pskip$ is a special case of assignment $\passign{x}{x}$ and is also a distribution as shown below.
\begin{thm}
    \label{thm:pskip}
\isalink{https://github.com/RandallYe/probabilistic_programming_utp/blob/6a4419b8674b84988065a58696f15093d176594c/probability/probabilistic_relations/utp_prob_rel_lattice_laws.thy\#L950}
    \begin{align}
        & \pskip = \left(\passign{x}{x}\right) \label{thm:pskip_id}\\
        & \isfinaldist\left(\prrvfunsym{\pskip}\right) \label{thm:pskip_final_dist} \\
        & \prrvfunsym{\left(\rvprfunsym{\ibracket{\II}}\right)} = \ibracket{\II} \label{thm:pskip_inverse}
    \end{align}
\end{thm}
Law~\ref{thm:pskip_inverse} shows that the conversion of $\ibracket{\II}$ to the $\ureal$-valued function, and then back to the real-valued function is still $\ibracket{\II}$.

\subsection{Assignments}
\label{ssec:prog_assigns}
A probabilistic assignment is a distribution.
\begin{thm}
    \label{thm:prob_assign_finaldist}
    $\isfinaldist\left(\prrvfunsym{\passign{x}{e}}\right)$
\isalink{https://github.com/RandallYe/probabilistic_programming_utp/blob/6a4419b8674b84988065a58696f15093d176594c/probability/probabilistic_relations/utp_prob_rel_lattice_laws.thy\#L991}
\end{thm}

\subsection{Probabilistic choice}
\label{ssec:prog_prob_choice}

Probabilistic choice preserves various properties below. 
\begin{thm}
    \label{thm:prob_prob_choice}
\isalink{https://github.com/RandallYe/probabilistic_programming_utp/blob/6a4419b8674b84988065a58696f15093d176594c/probability/probabilistic_relations/utp_prob_rel_lattice_laws.thy\#L998}
    \begin{align}
        & \isfinaldist\left(\prrvfunsym{P}\right) \land \isfinaldist\left(\prrvfunsym{Q}\right) \implies \isfinaldist\left(\prrvfunsym{\ppchoice{r}{P}{Q}}\right) \label{thm:pchoice_final_dist}\\
        & \left(\ppchoice{r}{P}{Q}\right) = \left(\ppchoice{\ufone-r}{Q}{P}\right) \label{thm:pchoice_commute}\\
        & \left(\ppchoice{\ufzero}{P}{Q}\right) = Q \label{thm:pchoice_zero}\\
        & \left(\ppchoice{\ufone}{P}{Q}\right) = P \label{thm:pchoice_one} \\
        & \left(\ppchoice{r}{P}{Q}\right) = \rvprfunsym{\prrvfunsym{r}*\prrvfunsym{P} + \left(\rfone - \prrvfunsym{r}\right) * \prrvfunsym{Q}} \label{thm:pchoice_altdef}
    \end{align}
\end{thm}
Law~\ref{thm:pchoice_final_dist} shows if $P$ and $Q$ are distributions, then the probabilistic choice is also a distribution. Laws~\ref{thm:pchoice_commute} to \ref{thm:pchoice_one} state the probabilistic choice is quasi-commutative, a zero, and a unit.


The probabilistic choice is also quasi-associative.
\begin{thm}[Quasi-associativity]
    \label{thm:pchoice_quasi_assoc}
    We fix $w_1, w_2, r_1, r_2: [S]\urexpr$, 
    \isalink{https://github.com/RandallYe/probabilistic_programming_utp/blob/6a4419b8674b84988065a58696f15093d176594c/probability/probabilistic_relations/utp_prob_rel_lattice_laws.thy\#L1199}
    \begin{align*}
        & \tautology{\left(\ufone - w_1\right)*\left(\ufone - w_2\right) = \left(\ufone - r_2\right)} \land \tautology{w_1 = r_1 * r_2}
        \implies \left(\ppchoice{{w_1}}{P}{\left(\ppchoice{{w_2}}{Q}{R}\right)}\right) 
        = \left(\ppchoice{{r_2}}{\left(\ppchoice{{r_1}}{P}{Q}\right)}{R}\right)
    \end{align*}
\end{thm}
The probabilistic choice is quasi-associative under the two assumptions involving $w_1$, $w_2$, $r_1$, and $r_2$. 

\subsection{Conditional choice}
Conditional choice satisfies various properties below. 
\label{ssec:prog_cond_choice}
\begin{thm}
    \label{thm:prog_cond_choice}
    \isalink{https://github.com/RandallYe/probabilistic_programming_utp/blob/6a4419b8674b84988065a58696f15093d176594c/probability/probabilistic_relations/utp_prob_rel_lattice_laws.thy\#L1294}
    \begin{align}
        & \isfinaldist\left(\prrvfunsym{P}\right) \land \isfinaldist\left(\prrvfunsym{Q}\right) \implies \isfinaldist\left(\prrvfunsym{\pcchoice{b}{P}{Q}}\right) \label{thm:cchoice_final_dist}\\
        & \left(\pcchoice{b}{P}{P}\right) = P \label{thm:cchoice_id}\\
        & \left(\pcchoice{b}{P}{Q}\right) = \left(\ppchoice{\rvprfunsym{\ibracket{b}}\ }{P}{Q}\right) \label{thm:cchoice_pchoice}\\
        & \left(P_1 \leq P_2 \land Q_1 \leq Q_2\right) \implies \left(\pcchoice{b}{P_1}{Q_1}\right) \leq \left(\pcchoice{b}{P_2}{Q_2}\right) \label{thm:cchoice_mono}
    \end{align}
\end{thm}
Law~\ref{thm:cchoice_final_dist} shows if $P$ and $Q$ are distributions, then the conditional choice is also a distribution. Law~\ref{thm:cchoice_id} shows the conditional choice between $P$ and $P$ is $P$ itself. A conditional choice is a special form of probabilistic choice, as given in Law~\ref{thm:cchoice_pchoice}, with its weight $\ibracket{b}$ being the Iverson bracket of $b$: either 0 (if $b$ is evaluated to false) or 1 (if $b$ is evaluated to true). The conditional choice is also monotonic, shown in Law~\ref{thm:cchoice_mono}.

\subsection{Sequential Composition}
\label{ssec:prog_seq_comp}
A variety of properties are held for sequential composition.
{We note that $r$ and $t$ below are predicates of type $[S]\upred$.}
\begin{thm}
    \label{thm:prog_seq_comp}
    \isalink{https://github.com/RandallYe/probabilistic_programming_utp/blob/6a4419b8674b84988065a58696f15093d176594c/probability/probabilistic_relations/utp_prob_rel_lattice_laws.thy\#L1378}
    \begin{align}
        & \isfinaldist\left(\prrvfunsym{P}\right) \land \isfinaldist\left(\prrvfunsym{Q}\right) \implies \isfinaldist\left(\prrvfunsym{\pseq{P}{Q}}\right) \label{thm:pseq_final_dist}\\
        & \pseq{\ufzero}{P} = \ufzero \label{thm:pseq_left_zero}\\
        & \pseq{P}{\ufzero} = \ufzero \label{thm:pseq_right_zero}\\
        & \pseq{\pskip}{P} = P \label{thm:pseq_left_unit}\\
        & \pseq{P}{\pskip} = P \label{thm:pseq_right_unit}\\
        & \isfinaldist(\prrvfunsym{P}) \implies \pseq{P}{\ufone} = \ufone \label{thm:pseq_one}\\
        & \left(P_1 \leq P_2 \land Q_1 \leq Q_2\right) \implies \left(\pseq{P_1}{Q_1}\right) \leq \left(\pseq{P_2}{Q_2}\right) \label{thm:pseq_mono}\\
        & \isfinaldist\left(\prrvfunsym{P}\right) \land \isfinaldist\left(\prrvfunsym{Q}\right) \land \isfinaldist\left(\prrvfunsym{R}\right) \implies \left(\pseq{P}{(\pseq{Q}{R})} = \pseq{(\pseq{P}{Q})}{R}\right) \label{thm:pseq_assoc}\\
        & \isfinalsubdist\left(\prrvfunsym{P}\right) \land \isfinalsubdist\left(\prrvfunsym{Q}\right) \land \isfinalsubdist\left(\prrvfunsym{R}\right) \implies \left(\pseq{P}{(\pseq{Q}{R})} = \pseq{(\pseq{P}{Q})}{R}\right) \label{thm:pseq_assoc_subdist}\\
        & \isfinalsubdist\left(\prrvfunsym{P}\right) \implies \left(\pseq{P}{\left(\pcchoice{b}{Q}{R}\right)} = \rvprfunsym{\usexpr{\prrvfunsym{\left(\pseq{P}{\left(\ibracket{b}*Q\right)}\right)} + \prrvfunsym{\left(\pseq{P}{\left(\ibracket{\lnot b}*R\right)}\right)} }}\right) \label{thm:pseq_dist_cchoice} \\
        & \pseq{\rvprfunsym{\ibracket{r}}}{\rvprfunsym{\ibracket{t}}} = \rvprfunsym{\usexpr{\infsum v_0 @ \ibracket{r[v_0/\vv'] \land t[v_0/\vv]}}} \label{thm:pseq_ibracket} 
    \end{align}
\end{thm}
If $P$ and $Q$ are distributions, then sequential composition of $P$ and $Q$ is also a distribution (Law~\ref{thm:pseq_final_dist}). Sequential composition is left zero (Law~\ref{thm:pseq_left_zero}), right zero (Law~\ref{thm:pseq_right_zero}), left unit (Law~\ref{thm:pseq_left_unit}), right unit (Law~\ref{thm:pseq_right_unit}), monotonic (Law~\ref{thm:pseq_mono}), and associative (Laws~\ref{thm:pseq_assoc} and~\ref{thm:pseq_assoc_subdist} if $P$, $Q$, and $R$ are distributions or subdistributions). If $P$ is a subdistribution, then Law~\ref{thm:pseq_dist_cchoice} shows $P$ is distributive through conditional choice.

An interesting Law~\ref{thm:pseq_one} (\isaref{https://github.com/RandallYe/probabilistic_programming_utp/blob/6a4419b8674b84988065a58696f15093d176594c/probability/probabilistic_relations/utp_prob_rel_lattice_laws.thy\#L1845}) states if $P$ is a distribution over the final state, then sequence composition of $P$ with $\ufone$ is $\ufone$. Its proof is given below.
\begin{proof}
\begin{align*}
  & \pseq{P}{\ufone} \\
= & \cmt{Definition~\ref{def:prog_seq}} \\
  & \rvprfunsym{\usexpr{\infsum v_0 @ \prrvfunsym{P}[v_0/\vv'] * \prrvfunsym{\ufone}[v_0/\vv]}} \\
= & \cmt{Theorem~\ref{thm:top_bot} and ${\rfone}[v_0/\vv] = \rfone$} \\
  & \rvprfunsym{\usexpr{\infsum v_0 @ \prrvfunsym{P}[v_0/\vv'] * {\rfone}}} \\
= & \cmt{Pointwise multiplication and multiplication unit law: $x * 1 = x$} \\
  & \rvprfunsym{\usexpr{\infsum v_0 @ \prrvfunsym{P}[v_0/\vv']}} \\
= & \cmt{ Assumption: $\isfinaldist\left(\prrvfunsym{P}\right)$ and Theorem~\ref{thm:final_distribtion} Law~\ref{thm:final_distribtion_1}: $\left(\forall s @ \infsum s' @ \prrvfunsym{P}(s, s') = 1\right)$} \\
  & \rvprfunsym{\rfone} \\
= & \cmt{Theorem~\ref{thm:top_bot}} \\
  & \ufone 
\end{align*}
\end{proof}

The sequential composition of two Iverson bracket expressions can be simplified to the summation of the Iverson bracket of conjunction, as shown in Law~\ref{thm:pseq_ibracket} (\isaref{https://github.com/RandallYe/probabilistic_programming_utp/blob/6a4419b8674b84988065a58696f15093d176594c/probability/probabilistic_relations/utp_prob_rel_lattice_laws.thy\#L1418}). This law is proved below.
\begin{proof}
\begin{align*}
 & \pseq{\rvprfunsym{\ibracket{r}}}{\rvprfunsym{\ibracket{t}}}  \\
= & \cmt{Definition~\ref{def:prog_seq}} \\
& \rvprfunsym{\usexpr{\infsum v_0 @ \prrvfunsym{\left(\rvprfunsym{\ibracket{r}}\right)}[v_0/\vv'] * \prrvfunsym{\left(\rvprfunsym{\ibracket{t}}\right)}[v_0/\vv]}} \\
= & \cmt{Theorem~\ref{thm:prrvfun_inverse_ibracket}} \\
& \rvprfunsym{\usexpr{\infsum v_0 @ \ibracket{r}[v_0/\vv'] * \ibracket{t}[v_0/\vv]}} \\
= & \cmt{Substitution distributive through Iverson bracket} \\
& \rvprfunsym{\usexpr{\infsum v_0 @ \ibracket{r[v_0/\vv']} * \ibracket{t[v_0/\vv]}}} \\
= & \cmt{Theorem~\ref{thm:ib} Law~\ref{thm:ib_conj}} \\
 & \rvprfunsym{\usexpr{\infsum v_0 @ \ibracket{r[v_0/\vv'] \land t[v_0/\vv]}}} 
\end{align*}
\end{proof}

A corollary of this law, given below, states that if the two expressions cannot agree on an intermediate state $v_0$, the sequence is just $\ufzero$. 
\begin{thm}
    \label{thm:pseq_ibracket_contradictory}
   $c_1 \neq c_2 \implies \pseq{\rvprfunsym{\ibracket{x'=c_1}}}{\rvprfunsym{\ibracket{x=c_2}}} = \ufzero$ 
   \isalink{https://github.com/RandallYe/probabilistic_programming_utp/blob/6a4419b8674b84988065a58696f15093d176594c/probability/probabilistic_relations/utp_prob_rel_lattice_laws.thy\#L1428}
\end{thm}
The intermediate state can be ignored if the two expressions agree on one intermediate state $c_1$.
\begin{thm}
    \label{thm:pseq_ibracket_agree_1_final_unspecified}
   $\pseq{\rvprfunsym{\ibracket{x=c_0 \land x:=c_1}}}{\rvprfunsym{\ibracket{x=c_1}}} = {\rvprfunsym{\ibracket{x=c_0}}}$ 
   \isalink{https://github.com/RandallYe/probabilistic_programming_utp/blob/6a4419b8674b84988065a58696f15093d176594c/probability/probabilistic_relations/utp_prob_rel_lattice_laws.thy\#L1443}
\end{thm}
The intermediate state can be ignored if the second expression is also a point distribution, but its final state is still specified.
\begin{thm}
    \label{thm:pseq_ibracket_agree_1_point}
   $\pseq{\rvprfunsym{\ibracket{x=c_0 \land x:=c_1}}}{\rvprfunsym{\ibracket{x=c_1 \land x: = c_2}}} = {\rvprfunsym{\ibracket{x=c_0 \land x'=c_2}}}$ 
   \isalink{https://github.com/RandallYe/probabilistic_programming_utp/blob/6a4419b8674b84988065a58696f15093d176594c/probability/probabilistic_relations/utp_prob_rel_lattice_laws.thy\#L1464}
\end{thm}

\subsection{Normalisation}
\label{ssec:prog_normalisation}

%


The $\normf$ in Definition~\ref{def:normf} gives a distribution of the final state, that is, over all variables in the state space. We also want the distribution of just one particular variable instead of all, for example, to define a uniform distribution. For this purpose, we define the alphabetised normalisation below.

\begin{definition}[Alphabetised normalisation]
    \label{def:norm_alpha}
    We fix $p$ of type $[S_1,S_2]\rvfun$ and a program variable $x$ of type $T_x$,
    \isalink{https://github.com/RandallYe/probabilistic_programming_utp/blob/6a4419b8674b84988065a58696f15093d176594c/probability/probabilistic_relations/utp_distribution.thy\#L63}
\begin{align*}
    \normal(x, p) \defs \usexpr{p / \left(\infsum x_0: T_x @ p[x_0/x']\right)}
\end{align*}
\end{definition}

Uniform distributions are defined using $\normal$.
\begin{definition}[Uniform distributions]
    \label{def:uniform_dist}
    We fix a program variable $x$ of type $T_x$ and a finite set $A$ of type $\power~T_x$,
    \isalink{https://github.com/RandallYe/probabilistic_programming_utp/blob/6a4419b8674b84988065a58696f15093d176594c/probability/probabilistic_relations/utp_distribution.thy\#L69}
\begin{align*}
    \uniformdist{x}{A} \defs \normal \left(x, \livbr \Inf v \in A @ x := v \rivbr\right)
\end{align*}
\end{definition}
A uniform distribution of $x$ from a finite set $A$ is an alphabetised normalisation of a program $\ibracket{\Inf v \in A @ x := v }$, nondeterministic choice $\Inf$ of the value of $x$ from $A$, over $x'$. Here, $\Inf$ is inside the Iverson bracket and is a UTP relation operator infimum, simply disjunction $\Union$.

The uniform distribution operator satisfies the properties below.
\begin{thm}
    \label{thm:uniform_dist}
    We fix $P$ of type $[S_1,S_2]\prfun$,
    \isalink{https://github.com/RandallYe/probabilistic_programming_utp/blob/6a4419b8674b84988065a58696f15093d176594c/probability/probabilistic_relations/utp_prob_rel_lattice_laws.thy\#L2222}
    \begin{align}
        & \uniformdist{x}{\emptyset} = \rfzero \label{thm:uniform_emptyset} \\
        & \finite(A) \implies \isprob\left(\uniformdist{x}{A}\right) \label{thm:uniform_prob} \\
        & \finite(A) \land A \neq \emptyset \implies \isfinaldist\left(\uniformdist{x}{A}\right)\label{thm:uniform_finaldist} \\
        & \finite(A) \land A \neq \emptyset \implies \left(\forall v \in A @ \pseq{\uniformdist{x}{A}}{\ibracket{x = v}} = \usexpr{1/\card (A)}\right) \label{thm:uniform_uniform} \\
        & \finite(A) \land A \neq \emptyset \implies \left(\uniformdist{x}{A} = \livbr \Union v \in A @ x := v \rivbr / \card (A)\right)\label{thm:uniform_form2} \\
        & \finite(A) \land A \neq \emptyset \implies \left(\pseq{\rvprfunsym{\uniformdist{x}{A}}}{P} = \rvprfunsym{\left(\infsum v \in A @ \prrvfunsym{P}[v/x]\right) / \card (A)}\right)\label{thm:uniform_pseq} 
    \end{align}
\end{thm}
Law~\ref{thm:uniform_emptyset} shows the distribution over an empty set $\emptyset$ is just the constant function $\rfzero$. Provided $A$ is finite, then $\uniformdist{x}{A}$ is probabilistic (Law~\ref{thm:uniform_prob}). If $A$ is also not empty ($A \neq \emptyset$), $\uniformdist{x}{A}$ is also a distribution (Law~\ref{thm:uniform_finaldist}). Under the same assumptions about $A$, $\uniformdist{x}{A}$ is truly a uniform distribution, that is, $x$ being any value from $A$ is equally likely ($1/\card(A)$ where $\card(A)$ is the cardinality of $A$), as shown in Law~\ref{thm:uniform_uniform}, where we use sequential composition $\pseq{\uniformdist{x}{A}}{\ibracket{x = v}}$ to express the probability of $x$ being a particular value $v$. The distribution $\uniformdist{x}{A}$ can be simplified to another form, shown in Law~\ref{thm:uniform_form2}. The sequence of a uniform distribution and $P$ is a left one-point, shown in Law~\ref{thm:uniform_pseq}.

\subsection{Parallel composition}
The \ref{def:prog_parallel_ff} is defined over real-valued functions and specifies a $\ureal$-valued function. It satisfies the properties below.
\label{ssec:prog_parallel_comp}
\begin{thm}
    Fix $p$, $q$, and $r$ of type $[S_1, S_2] \rvfun$, and $P$, $Q$, and $R$ of type $[S_1, S_2] \prfun$.
    \isalink{https://github.com/RandallYe/probabilistic_programming_utp/blob/6a4419b8674b84988065a58696f15093d176594c/probability/probabilistic_relations/utp_prob_rel_lattice_laws.thy\#L2559}
    \begin{align}
        & \isnonneg(p*q)\implies \isprob\left(\normf \usexpr{p * q}\right) \label{thm:pparallel_norm_prob} \\
        & \left(
        \begin{array}[]{l}
            \isfinalprob(p) \land \isfinalprob(q) \land \\
            \left(\summableonfinal(p) \lor \summableonfinal(q)\right) \land \finalreachables(p,q) 
        \end{array}\right)
        \implies \isfinaldist\left(\pparallel{p}{q}\right)\label{thm:pparallel_dist} \\
        & \left(
        \begin{array}[]{l}
            \isnonneg(p) \land \isnonneg(q) \land \lnot\finalreachables(p,q) 
        \end{array}\right)
        \implies \pparallel{p}{q} = \ufzero \label{thm:pparallel_contradiction_zero} \\
        & \pparallel{\rfzero}{p} = \ufzero \label{thm:pparallel_left_zero} \\
        & \pparallel{p}{\rfzero} = \ufzero \label{thm:pparallel_right_zero} \\
        & c \neq 0 \land \isfinaldist(p) \implies \pparallel{(\lambda s @ c)}{p} = \rvprfunsym{p} \label{thm:pparallel_left_unit} \\
        & c \neq 0 \land \isfinaldist(p) \implies \pparallel{p}{(\lambda s @ c)} = \rvprfunsym{p} \label{thm:pparallel_right_unit} \\
        & \pparallel{p}{q} = \pparallel{q}{p} \label{thm:pparallel_commute} \\
        & \left(
        \begin{array}[]{l}
            \isnonneg(p) \land \isnonneg(q) \land \isnonneg(r)\land \\
            \summableonfinals(p, q) \land 
            \summableonfinals(q, r) \land  \\
            \finalreachables(p, q) \land 
            \finalreachables(q, r) 
        \end{array}\right)
        \implies \pparallel{\left(\pparallel{p}{q}\right)}{r} = \pparallel{p}{\left(\pparallel{q}{r}\right)} \label{thm:pparallel_assoc} \\
        & 
        \summableonfinal(\prrvfunsym{Q})
        \implies \pparallel{\left(\pparallel{\prrvfunsym{P}}{\prrvfunsym{Q}}\right)}{\prrvfunsym{R}} = \pparallel{\prrvfunsym{P}}{\left(\pparallel{\prrvfunsym{Q}}{\prrvfunsym{R}}\right)} \label{thm:pparallel_assoc2} \\
        & \finite(A) \land A \neq \emptyset \implies \left(\pparallel{\uniformdist{x}{A}}{p} = \rvprfunsym{\usexpr{\left(\infsum v \in A @ \ibracket{x := v} * p[v/x']\right) / \left(\infsum v \in A @ p[v/x']\right)}}\right) \label{thm:pparallel_uniform} 
    \end{align}
\end{thm}
Law~\ref{thm:pparallel_norm_prob} shows the normalisation of the product $p*q$ of $p$, and $q$ is probabilistic if $p*q$ is nonnegative. 
If both $p$ and $q$ are probabilistic and summable on their final states, reachable on at least one same final state at the same time, then $\pparallel{p}{q}$ is also a distribution of the final state (Law~\ref{thm:pparallel_dist}). If both $p$ and $q$ are nonnegative and not reachable on at least one same final state at the same time (or a contradiction between $p$ and $q$), then $\pparallel{p}{q}$ is a zero (Law~\ref{thm:pparallel_contradiction_zero}). 

Parallel composition is a left zero (Law~\ref{thm:pparallel_left_zero}) and a right zero (Law~\ref{thm:pparallel_right_zero}), and a left unit (Law~\ref{thm:pparallel_left_unit}) and a right unit (Law~\ref{thm:pparallel_right_unit}) if a constant $c$ is not 0 and $p$ is a distribution. It is also commutative (Law~\ref{thm:pparallel_commute}).

Law~\ref{thm:pparallel_assoc} shows if $p$, $q$, and $r$ are nonnegative, both $p$ and $q$ are summable and reachable on their product, and both $q$ and $r$ are summable and reachable on their product, then the parallel composition is associative. If, however, $p$, $q$, and $r$ are converted from probabilistic programs $P$, $Q$, and $R$, and also $Q$ is summable on its final state, then the parallel composition is also associative (Law~\ref{thm:pparallel_assoc2}).

Law~\ref{thm:pparallel_uniform} shows if $A$ is finite and not empty, then the parallel composition of a uniform distribution over $x$ from $A$ and $p$ can be simplified to a division whose numerator denotes the value of $p$ reaching a final state with $x$ being a particular value $v$ and whose denominator represents the summation of the values of $p$ reaching final states with $x$ being any value $v$ from $A$.


\section{Recursion}
\label{sec:rec}
This section presents the syntax and semantics of probabilistic recursions. We use the Kleene fixed-point theorem to construct fixed points iteratively. 

\subsection{Probabilistic loops}
We introduce the syntax for the least and greatest fixed points based on the Knaster–Tarski fixed-point theorem~\ref{thm:tarski_fixed_point}. 

\begin{definition}[Least and greatest fixed points]
    \isalink{https://github.com/RandallYe/probabilistic_programming_utp/blob/6a4419b8674b84988065a58696f15093d176594c/probability/probabilistic_relations/utp_prob_rel_lattice.thy\#L384}
    \begin{align*}
        \lfp~X @ P &\defs \thlfp{}~\left(\lambda X @ P\right)\tag*{(least fixed point)} \label{def:lfp}\\ 
        \gfp~X @ P &\defs \thgfp{}~\left(\lambda X @ P\right)\tag*{(great fixed point)} \label{def:gfp}
    \end{align*}
\end{definition}

We define a loop function $\lfun$ below and use it to define probabilistic while loops later.
\begin{definition}[Loop function]
    We fix a homogeneous relation $b:[S]\hrel$, and homogeneous probabilistic programs $P$ and $X$ of type $[S]\prhfun$, then
    \isalink{https://github.com/RandallYe/probabilistic_programming_utp/blob/6a4419b8674b84988065a58696f15093d176594c/probability/probabilistic_relations/utp_prob_rel_lattice.thy\#L408}
    \begin{align*}
        &\lfunp{b}{P}{X} \defs~\pcchoice{b}{\left(\pseq{P}{X}\right)}{\pskip} \tag*{(loop function)} \label{def:lfun} 
    \end{align*}
\end{definition}
The function $\lfun$ is a conditional choice between the sequence $\left(\pseq{P}{X}\right)$ and the skip $\pskip$, depending on the relation $b$. We use $\lfunbp$ as a shorthand for $\lambda X @ \lfun(b, P, X)$. {Then $\lfunbp(X)$ can be expressed below.} 
\begin{align*}
   &\lfunbp(X)\\
 = & \cmt{Definition~\ref{def:lfun} } \\
   & \pcchoice{b}{\left(\pseq{P}{X}\right)}{\pskip} \\
 = & \cmt{Theorem~\ref{thm:prog_cond_choice} Law~\ref{thm:cchoice_pchoice} and Definition~\ref{def:prog_skip}} \\
 & \left(\ppchoice{\rvprfunsym{\ibracket{b}}}{\left(\pseq{P}{X}\right)}{\rvprfunsym{\ibracket{\II}}}\right) \\
 = & \cmt{Theorem~\ref{thm:prob_prob_choice} Law~\ref{thm:pchoice_altdef}} \\
 & \rvprfunsym{\prrvfunsym{\left(\rvprfunsym{\ibracket{b}}\right)}*\prrvfunsym{\left(\pseq{P}{X}\right)} + \left(\rfone - \prrvfunsym{\left(\rvprfunsym{\ibracket{b}}\right)}\right) * \prrvfunsym{\left(\rvprfunsym{\ibracket{\II}}\right)}} \\
 = & \cmt{Theorem~\ref{thm:prrvfun_inverse}, Theorem~\ref{thm:pskip} Law~\ref{thm:pskip_inverse}, Theorem~\ref{thm:prob_assign_finaldist}, and Theorem~\ref{thm:final_distribtion} Law~\ref{thm:final_distribtion_1}} \\
 & \rvprfunsym{{\ibracket{b}}*\prrvfunsym{\left(\pseq{P}{X}\right)} + \left(\rfone - {\ibracket{b}}\right) * {\ibracket{\II}}} \\
 = & \cmt{Theorem~\ref{thm:ib} Law~\ref{thm:ib_neg}} \\
 & \rvprfunsym{{\ibracket{b}}*\prrvfunsym{\left(\pseq{P}{X}\right)} + \ibracket{\lnot b} * {\ibracket{\II}}}\tag*{($\lfunbp$ altdef)} \label{thm:lfun_altdef}
\end{align*}

If a probabilistic program $P$ is a distribution, then $\lfunbp$ is monotonic.
\begin{thm}
    $\isfinaldist\left(\prrvfunsym{P}\right) \implies \mono\left(\lfunbp\right)$
    \isalink{https://github.com/RandallYe/probabilistic_programming_utp/blob/6a4419b8674b84988065a58696f15093d176594c/probability/probabilistic_relations/utp_prob_rel_lattice_laws.thy\#L4253}
\end{thm}

With $\lfun$, we define two probabilistic loops using the least and the greatest fixed points. The reason we define two probabilistic loops is to establish the uniqueness theorem of fixed points. Based on Knaster–Tarski fixed-point Theorem~\ref{thm:tarski_fixed_point}, the uniqueness is equivalent to the equality of the least and greatest fixed points.

\begin{definition}[Probabilistic loops]
    \isalink{https://github.com/RandallYe/probabilistic_programming_utp/blob/6a4419b8674b84988065a58696f15093d176594c/probability/probabilistic_relations/utp_prob_rel_lattice.thy\#L411}
    {
    \begin{align*}
        \pwhile{b}{P} &\defs \lfp~X @ \lfunbp(X) \tag*{(while loop by least fixed point)} \label{def:pwhile}\\
        \pwhiletop{b}{P} &\defs \gfp~X @ \lfunbp(X) \tag*{(while loop by greatest fixed point)} \label{def:pwhile_top}
    \end{align*}
}
\end{definition}
{The denotational semantics of recursions in programming is usually defined on the lfp~\cite{GUNTER1990} or the weakest fixed point in UTP~\cite{Hoare1998}, as we do here for the probabilistic loop ($\pwhile{b}{P}$). But why the lfp is commonly used to give the denotation semantics for recursive, instead of the gfp? Gunter and Scott~\cite{GUNTER1990} stated that it is intuitively reasonable and lfp yields a canonical solution. Hoare and He~\cite{Hoare1998} argued that the weakest fixed point (wfp) is more implementable (but might be \changed[\C{3}]{non-terminating}) and the strongest fixed point (sfp) might be not implementable such as the miracle. The sfp is useful to prove the correctness of recursion programs, but subject to two problems: invalidate the implication in UTP to model correctness of designs, and difficult to implement nondeterminism. However, sfp can be used to establish the uniqueness of fixed points, and so wfp and sfp are the same and will not be subject to the problems for each of them. Our definition of ($\pwhiletop{b}{P}$) is for the same reason, merely for the proof of the unique fixed point theorem.
}

The $\pwhile{b}{P}$ satisfy several laws below.
\begin{thm}
    \label{thm:pwhile_bot}
    \isalink{https://github.com/RandallYe/probabilistic_programming_utp/blob/6a4419b8674b84988065a58696f15093d176594c/probability/probabilistic_relations/utp_prob_rel_lattice_laws.thy\#L4386}
\begin{align*}
    & \isfinaldist\left(\prrvfunsym{P}\right) \implies \pwhile{b}{P} = \lfunbp\left(\pwhile{b}{P}\right) \tag*{(unfold)} \label{thm:pwhile_unfold}\\
    & \pwhile{\ufalse}{P} = \pskip \tag*{(false)} \label{thm:pwhile_false}\\
    & \pwhile{\utrue}{P} = \ufzero \tag*{(true)} \label{thm:pwhile_true}
\end{align*}
\end{thm}
If $P$ is a distribution, $\pwhile{b}{P}$ can be unfolded without its semantics changed. If $b$ is $\ufalse$, the loop is just $\pskip$. It is $\ufzero$ {if $b$ is $\utrue$}.

The $\pwhiletop{b}{P}$ satisfy similar laws below.
\begin{thm}
    \label{thm:pwhile_top}
    \isalink{https://github.com/RandallYe/probabilistic_programming_utp/blob/6a4419b8674b84988065a58696f15093d176594c/probability/probabilistic_relations/utp_prob_rel_lattice_laws.thy\#L4423}
\begin{align*}
    \isfinaldist\left(\prrvfunsym{P}\right) \implies & \pwhiletop{b}{P} = \lfunbp\left(\pwhiletop{b}{P}\right) \tag*{(unfold)} \label{thm:pwhiletop_unfold}\\
    & \pwhiletop{\ufalse}{P} = \pskip \tag*{(false)} \label{thm:pwhiletop_false}\\
    \isfinaldist\left(\prrvfunsym{P}\right) \implies & \pwhiletop{\utrue}{P} = \ufone \tag*{(true)} \label{thm:pwhiletop_true}
\end{align*}
\end{thm}
We note that if $P$ is a distribution and $b$ is $\utrue$, the loop is just $\ufone$, instead of $\ufzero$ for $\pwhileop$. 
{In Theorems~\ref{thm:pwhile_bot} and \ref{thm:pwhile_top}, both loops are well defined if $P$ is a distribution. To reason about nested loops, we need to give the semantics (that is, the fixed point $fp$) to the innermost loop first, prove $fp$ is a distribution, and then move to the next loop closed to that loop.} \changed[\C{2}]{If, however, the inner loop does not almost surely terminate, currently we cannot give the semantics to the outer loop because this loop is not well defined. One line of our future work is to investigate the weakened condition of $P$ from distributions to subdistributions to allow us to give the semantics to nested loops if the inner loop does not terminate almost surely.}

As the semantics for probabilistic loops in our programming language is the \ref{def:tarski_lfp} and the \ref{def:tarski_gfp} in Knaster–Tarski fixed-point Theorem~\ref{thm:tarski_fixed_point}, this does not give information about how iterations can compute fixed points. We resort to Kleene fixed-point Theorem~\ref{thm:kleene_fixed_point_theorem} for iterations. Our use of iterations to compute fixed points is motivated by a simple probabilistic program: flip a coin until the outcome is heads, defined as follows.

\subsection{Motivation example}
This example~\cite{Kozen1985,Morgan1996a,Dahlqvist2020} is about flipping a coin till the outcome is heads. 
\begin{definition}[Flip a coin]
    \label{def:coin_flip}
    \isalink{https://github.com/RandallYe/probabilistic_programming_utp/blob/6a4419b8674b84988065a58696f15093d176594c/probability/probabilistic_relations/Examples/utp_prob_rel_lattice_coin.thy\#L15}
    \begin{align*}
        & Tcoin ::= hd | tl \\ 
        & \isakwmaj{alphabet}\ cstate = c::Tcoin \\
        & cflip \defs \ppchoice{1/2}{\passign{c}{hd}}{\passign{c}{tl}} \\
        & flip \defs \pwhile{c = tl}{cflip}
    \end{align*}
\end{definition}
$Tcoin$ is a free type in the Z notation or an enumerable data type in Isabelle, and it contains two constants $hd$ and $tl$ of type $Tcoin$. The keyword $\isakwmaj{alphabet}$ is used to declare the state space of a program, and it is $cstate$ in this case. This state space is composed of only one variable $c$ of type $Tcoin$, denoting the outcome of the coin flip experiment. The $cflip$ is a probabilistic choice denoting both $hd$ and $tl$ are equally likely, or have a uniform distribution over two outcomes. Finally, this program $flip$ is a while loop whose condition is $c = tl$, specifying that the outcome is $tl$ and whose body is $cflip$.
The $cflip$ is simplified as follows.
\begin{align*}
   &cflip \\
 = & \cmt{Definition of $cflip$ } \\
   & \ppchoice{1/2}{\passign{c}{hd}}{\passign{c}{tl}}  \\
 = & \cmt{Law~\ref{thm:pchoice_altdef} in Theorem~\ref{thm:prob_prob_choice} and Definition~\ref{def:prog_assign}} \\
 & \rvprfunsym{\prrvfunsym{1/2}*\prrvfunsym{\left(\rvprfunsym{\ibracket{\passign{c}{hd}}}\right)} + \left(\rfone - \prrvfunsym{1/2}\right) * \prrvfunsym{\left(\rvprfunsym{\ibracket{\passign{c}{tl}}}\right)}} \\
 = & \cmt{Conversion Definitions~\ref{def:rvfun2prfun} and~\ref{def:u2r_r2u}, and Theorem~\ref{thm:prrvfun_inverse}} \\
 & \rvprfunsym{{1/2}*{\ibracket{\passign{c}{hd}}} + {1/2} * {\ibracket{\passign{c}{tl}}}} \\
 = & \cmt{Definition~\ref{def:uassign}} \\
 & \rvprfunsym{{1/2}*{\ibracket{c' = hd}} + {1/2} * {\ibracket{c' = tl}}} \tag*{($cflip$ altdef)} \label{thm:cflip_altdef}
\end{align*}
The $cflip$ is also a distribution.
\begin{lem}
    \label{thm:cflip_distribution}
    $\isfinaldist(\prrvfunsym{cflip})$ 
    \isalink{https://github.com/RandallYe/probabilistic_programming_utp/blob/6a4419b8674b84988065a58696f15093d176594c/probability/probabilistic_relations/Examples/utp_prob_rel_lattice_coin.thy\#L40}
\end{lem}
\begin{proof}
   \begin{align*}
    & \cmt{Theorem~\ref{thm:prob_assign_finaldist}} \\
    & \isfinaldist(\prrvfunsym{\passign{c}{hd}})  \land  \isfinaldist(\prrvfunsym{\passign{c}{tl}}) \\
    \implies & \cmt{Theorem~\ref{thm:prob_prob_choice} Law~\ref{thm:pchoice_final_dist}} \\
    & \isfinaldist(\prrvfunsym{cflip}) 
   \end{align*}
\end{proof}

Then the loop function $\lfun_{cflip}^{c = tl}(X)$, denoted as $\lfun_c$, is further simplified.
\begin{align*}
   & \lfun_c(X) \\
 = & \cmt{Law~\ref{thm:lfun_altdef}} \\
 & \rvprfunsym{{\ibracket{c=tl}}*\prrvfunsym{\left(\pseq{cflip}{X}\right)} + {\ibracket{\lnot c=tl}} * {\ibracket{\II}}} \\ 
 = & \cmt{Law~\ref{thm:cflip_altdef}, Theorem~\ref{def:utp_relation} Law~\ref{def:uskip}, and Thoerem~\ref{thm:ib} Law~\ref{thm:ib_conj}} \\
 & \rvprfunsym{{\ibracket{c=tl}}*\prrvfunsym{\left(\pseq{\left(\rvprfunsym{{1/2}*{\ibracket{c' = hd}} + {1/2} * {\ibracket{c' = tl}}} \right)}{X}\right)} + {\ibracket{\lnot c=tl}*\ibracket{c'=c}}} \\ 
 = & \cmt{Definition~\ref{def:coin_flip} of $Tcoin$: $(\lnot c = tl) = (c = hd)$} \\
 & \rvprfunsym{{\ibracket{c=tl}}*\prrvfunsym{\left(\pseq{\left(\rvprfunsym{{1/2}*{\ibracket{c' = hd}} + {1/2} * {\ibracket{c' = tl}}} \right)}{X}\right)} + {\ibracket{c=hd}*\ibracket{c'=c}}} \tag*{(loop body of $flip$)} \label{thm:F_cflip_tl_altdef} 
\end{align*}
Now consider $F^n(\bot)$ in the Kleene fixed-point theorem, and here $F$ is $\lfun_c$ and $\bot$ is $\ufzero$.
\begin{align*}
\left(\lfun_c\right)^0(\ufzero) =& \ufzero \\
\left(\lfun_c\right)^1(\ufzero) 
 =&\cmt{Law~\ref{thm:F_cflip_tl_altdef}} \\
& \rvprfunsym{{\ibracket{c=tl}}*\prrvfunsym{\left(\pseq{\left(\rvprfunsym{{1/2}*{\ibracket{c' = hd}} + {1/2} * {\ibracket{c' = tl}}} \right)}{\ufzero}\right)} + {\ibracket{c=hd}*\ibracket{c'=c}}}\\
 =&\cmt{Right Zero Theorem~\ref{thm:prog_seq_comp} Law~\ref{thm:pseq_right_zero}} \\
 & \rvprfunsym{{\ibracket{c=hd}*\ibracket{c'=c}}} \\
\left(\lfun_c\right)^2 (\ufzero) 
 = &\cmt{$F^2(\ufzero) = F(F^1(\ufzero))$ and Law~\ref{thm:lfun_altdef}} \\
 & \rvprfunsym{{\ibracket{c=tl}}*\prrvfunsym{\left(\pseq{cflip}{\left(F_{cflip}^{c=tl}\right)^1(\ufzero) }\right)} + {\ibracket{\lnot c=tl}} * {\ibracket{\II}}} \\ 
 = & \cmt{Law~\ref{thm:cflip_altdef}, ${\left(F_{cflip}^{c=tl}\right)^1(\ufzero)}$, Theorem~\ref{def:utp_relation} Law~\ref{def:uskip}, and Thoerem~\ref{thm:ib} Law~\ref{thm:ib_conj}} \\
 & \rvprfunsym{{\ibracket{c=tl}}*\prrvfunsym{\left(\pseq{\left(\rvprfunsym{{1/2}*{\ibracket{c' = hd}} + {1/2} * {\ibracket{c' = tl}}} \right)}{\rvprfunsym{{\ibracket{c=hd}*\ibracket{c'=c}}}}\right)} + {\ibracket{c=hd}*\ibracket{c'=c}}} \\ 
 = & \cmt{Definition~\ref{def:prog_seq}, Thoerem~\ref{thm:ib}, Theorem~\ref{thm:pseq_ibracket_contradictory}, $\cdots$} \\
 & \rvprfunsym{{\ibracket{c=tl}}*{\ibracket{c' = hd}}/2 + {\ibracket{c=hd}*\ibracket{c'=c}}} \\ 
\left(\lfun_c\right)^3 (\ufzero)
= & \cmt{{Same as previous proof}} \\
 & \rvprfunsym{{\ibracket{c=tl}}*{\ibracket{c' = hd}}/2 + {\ibracket{c=tl}}*{\ibracket{c' = hd}}/2^2 + {\ibracket{c=hd}*\ibracket{c'=c}}} \\ 
= & \cmt{{Same as previous proof}} \\
 & \rvprfunsym{{\ibracket{c=tl}}*{\ibracket{c' = hd}}*\left(1/2 + 1/2^2\right) + {\ibracket{c=hd}*\ibracket{c'=c}}} \\ 
 & \cdots \\
\left(\lfun_c\right)^n (\ufzero) 
 = & \cmt{Induction} \\
 & \rvprfunsym{{\ibracket{c=tl}}*{\ibracket{c' = hd}}*\left(1/2 + 1/2^2 + \cdots + 1/2^{(n-1)}\right) + {\ibracket{c=hd}*\ibracket{c'=c}}} \\ 
 = & \cmt{Summation} \\
 & \rvprfunsym{{\ibracket{c=tl}}*{\ibracket{c' = hd}}* \sum_{i=1}^{n-1} 1/2^{i} + {\ibracket{c=hd}*\ibracket{c'=c}}} \tag*{(iteration from bot)} \label{thm:F_b_P_clip_bot} 
\end{align*}
The $\left(\lfun_c\right)^n (\ufzero)$ corresponds to the termination probability after up to $n-1$ iterations of the loop body of $flip$ in Definition~\ref{def:coin_flip}. For example, $\left(\lfun_c\right)^1 (\ufzero)$ corresponds to zero iterations or immediate termination. Its semantics is $\rvprfunsym{{\ibracket{c=hd}*\ibracket{c'=c}}}$ which means if the initial value of $c$ is $hd$, then the final value is also $hd$ with probability 1. 
For example, $\left(\lfun_c\right)^3 (\ufzero)$ corresponds to the termination after up to two iterations, including the immediate termination, the termination after exact one iteration (${\ibracket{c=tl}}*{\ibracket{c' = hd}}/2$, meaning the initial value of $c$ is $tl$, and the outcome of the flip is $hd$ with probability $1/2$), and exact two iterations (${\ibracket{c=tl}}*{\ibracket{c' = hd}}/2^2$).  

Now consider $F^n(\top)$ in the Kleene fixed-point theorem, and here $\top$ is $\ufone$.
\begin{align*}
\left(\lfun_c\right)^0(\ufone) =& \ufone\\
\left(\lfun_c\right)^1(\ufone) 
 =&\cmt{Law~\ref{thm:F_cflip_tl_altdef}} \\
& \rvprfunsym{{\ibracket{c=tl}}*\prrvfunsym{\left(\pseq{\left(\rvprfunsym{{1/2}*{\ibracket{c' = hd}} + {1/2} * {\ibracket{c' = tl}}} \right)}{\ufone}\right)} + {\ibracket{c=hd}*\ibracket{c'=c}}}\\
 =&\cmt{Lemma~\ref{thm:cflip_distribution} and Theorem~\ref{thm:prog_seq_comp} Law~\ref{thm:pseq_one}} \\
 & \rvprfunsym{{\ibracket{c=tl}} + {\ibracket{c=hd}*\ibracket{c'=c}}} \\
\left(\lfun_c\right)^2 (\ufone) 
 = &\cmt{$F^2(\ufone) = F(F^1(\ufone))$ and Law~\ref{thm:lfun_altdef}} \\
 & \rvprfunsym{{\ibracket{c=tl}}*\prrvfunsym{\left(\pseq{cflip}{\left(F_{cflip}^{c=tl}\right)^1(\ufone) }\right)} + {\ibracket{\lnot c=tl}} * {\ibracket{\II}}} \\ 
 = & \cmt{Law~\ref{thm:cflip_altdef}, ${\left(F_{cflip}^{c=tl}\right)^1(\ufone)}$, Theorem~\ref{def:utp_relation} Law~\ref{def:uskip}, and Theorem~\ref{thm:ib} Law~\ref{thm:ib_conj}} \\
 & \rvprfunsym{{\ibracket{c=tl}}*\prrvfunsym{\left(\pseq{\left(\rvprfunsym{
     \begin{array}[]{l}
     {1/2}*{\ibracket{c' = hd}} + \\
     {1/2} * {\ibracket{c' = tl}}
     \end{array}
 } \right)}{\left(\rvprfunsym{
     \begin{array}[]{l}
        {\ibracket{c=tl}} + \\
        {\ibracket{c=hd}*\ibracket{c'=c}}
     \end{array}
 }\right)}\right)} + {\ibracket{c=hd}*\ibracket{c'=c}}} \\ 
 = & \cmt{Definition~\ref{def:prog_seq}, Theorems~\ref{thm:ib}, Theorem~\ref{thm:pseq_ibracket_contradictory}, $\cdots$} \\
 & \rvprfunsym{{\ibracket{c=tl}}*\prrvfunsym{\left({\rvprfunsym{
     \begin{array}[]{l}
     {1/2} * {\ibracket{c' = hd}} + {1/2}
     \end{array}
 } }\right)} + {\ibracket{c=hd}*\ibracket{c'=c}}} \\ 
 = & \cmt{ Theorem~\ref{thm:prrvfun_inverse} } \\
 & \rvprfunsym{{\ibracket{c=tl}/2}+ {\ibracket{c=tl}} * {\ibracket{c' = hd}}/2 + {\ibracket{c=hd}*\ibracket{c'=c}}} \\ 
\left(\lfun_c\right)^3 (\ufone)
= & \cmt{{Same as previous proof}} \\
 & \rvprfunsym{{\ibracket{c=tl}/2^2} + {\ibracket{c=tl}}*{\ibracket{c' = hd}}*\left(1/2 + 1/2^2\right) + {\ibracket{c=hd}*\ibracket{c'=c}}} \\ 
 & \cdots \\
\left(\lfun_c\right)^n (\ufone) 
 = & \cmt{Induction} \\
 & \rvprfunsym{{\ibracket{c=tl}/2^{n-1}} + {\ibracket{c=tl}}*{\ibracket{c' = hd}}* \sum_{i=1}^{n-1} 1/2^{i} + {\ibracket{c=hd}*\ibracket{c'=c}}} \tag*{(iteration from top)} \label{thm:F_b_P_clip_top} 
\end{align*}
We notice that $\left(\lfun_c\right)^n (\ufone)$ above is an addition of three operands of which the last two are the same as $\left(\lfun_c\right)^n (\ufzero)$ in \ref{thm:F_b_P_clip_bot}. The extra operand ${\ibracket{c=tl}/2^{n-1}}$ converges to 0 when $n$ approaches $\infty$, and so eventually $\left(\lfun_c\right)^n (\ufzero)$ and $\left(\lfun_c\right)^n (\ufone)$ coincide at $\infty$. This is illustrated in Fig.~\ref{fig:coin_t_prob_iteration} where the common part ${\ibracket{c=hd}*\ibracket{c'=c}}$ in $\left(\lfun_c\right)^n (\ufzero)$ and $\left(\lfun_c\right)^n (\ufone)$ is omitted.
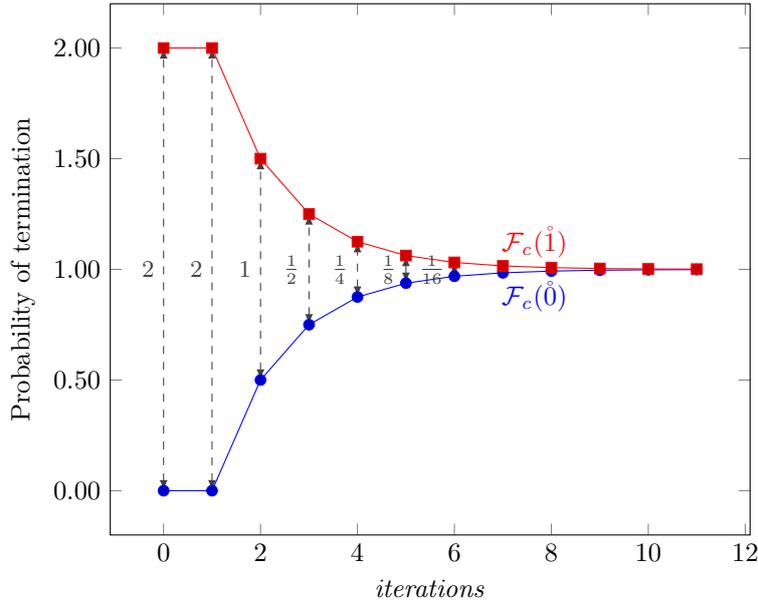
\begin{figure}[!ht]
    \begin{center}
\begin{tikzpicture}
\begin{axis}[
    xlabel={$iterations$},
    ylabel={Probability of termination},
    y tick label style={/pgf/number format/.cd,fixed,fixed zerofill,precision=2},
]
    \addplot table[header=false,col sep=&,row sep=\\,y expr={\thisrowno{1}}] {
        0 & 0\\
        1 & 0\\
        2 & 1 - 1/2^1\\
        3 & 1 - 1/2^2\\
        4 & 1 - 1/2^3\\
        5 & 1 - 1/2^4\\
        6 & 1 - 1/2^5\\
        7 & 1 - 1/2^6\\
        8 & 1 - 1/2^7\\
        9 & 1 - 1/2^8\\
        10& 1 - 1/2^9\\
        11& 1 - 1/2^10\\
    } node[below,pos=0.7] {$\lfun_c(\ufzero)$};

    \addplot table[header=false,col sep=&,row sep=\\,y expr={\thisrowno{1}}] {
        0 & 2\\
        1 & 2\\
        2 & 2/2^1 + 1 - 1/2^1 \\
        3 & 2/2^2  + 1 - 1/2^2 \\
        4 & 2/2^3  + 1 - 1/2^3 \\
        5 & 2/2^4  + 1 - 1/2^4 \\
        6 & 2/2^5  + 1 - 1/2^5 \\
        7 & 2/2^6  + 1 - 1/2^6 \\
        8 & 2/2^7  + 1 - 1/2^7 \\
        9 & 2/2^8  + 1 - 1/2^8 \\
        10& 2/2^9  + 1 - 1/2^9 \\
        11& 2/2^10 + 1 - 1/2^10\\
    } node[above,pos=0.7] {$\lfun_c(\ufone)$};

    \draw[latex-latex, darkgray, dashed] (axis cs:0,0) -- node[left]{\small 2} (axis cs:0,2);
    \draw[latex-latex, darkgray, dashed] (axis cs:1,0) -- node[left]{\small 2} (axis cs:1,2);
    \draw[latex-latex, darkgray, dashed] (axis cs:2, 1 - 1/2^1) -- node[left]{\small 1} (axis cs:2, 2/2^1 + 1 - 1/2^1);
    \draw[latex-latex, darkgray, dashed] (axis cs:3, 1 - 1/2^2) -- node[left]{\small $\frac{1}{2}$} (axis cs:3, 2/2^2 + 1 - 1/2^2);
    \draw[latex-latex, darkgray, dashed] (axis cs:4, 1 - 1/2^3) -- node[left]{\small $\frac{1}{4}$} (axis cs:4, 2/2^3 + 1 - 1/2^3);
    \draw[latex-latex, darkgray, dashed] (axis cs:5, 1 - 1/2^4) -- node[left]{\small $\frac{1}{8}$} (axis cs:5, 2/2^4 + 1 - 1/2^4);
    \draw[latex-latex, darkgray, dashed] (axis cs:6, 1 - 1/2^5) -- node[left]{\small $\frac{1}{16}$} (axis cs:6, 2/2^5 + 1 - 1/2^5);
  \end{axis}
\end{tikzpicture}
    \end{center}
    \caption{Termination probability over iterations from bottom and top for coin flip.}
    \label{fig:coin_t_prob_iteration}
\end{figure}
As shown in the diagram, along with the increasing iterations, $\left(\lfun_c\right)^n (\ufzero)$ increases towards 1 from the initial 0, while $\left(\lfun_c\right)^n (\ufone)$ decreases towards 1 from the initial 2. Their differences, marked with dashed lines, are becoming smaller and smaller.
From this example, we observe that there is one unique fixed point where the least fixed point and the greatest fixed point are the same. The precondition for this uniqueness is their iteration differences converging to 0. If this is the case, we must find a fixed point and prove it to reason about a probabilistic loop. Then the fixed point is the semantics of the loop. We do not need to compute it by iterations. 

This example motivates us to give semantics to probabilistic loops as follows.
First, we need to prove the loop function is continuous. Then if the differences of iterations from top and bottom converge to 0, there is a unique fixed point. Otherwise, we use the Kleene fixed-point Theorem~\ref{thm:kleene_fixed_point_theorem} to compute the least and greatest fixed points by iterations.

\subsection{Fixed point theorems}
We define the function $\iter$ below recursively for iterations $(\lfun_P^b)^n(X)$, and the function $\iterdiff$ for the differences between iterations from top and bottom.
\begin{definition}[Iteration and iteration difference]
    \isalink{https://github.com/RandallYe/probabilistic_programming_utp/blob/6a4419b8674b84988065a58696f15093d176594c/probability/probabilistic_relations/utp_prob_rel_lattice.thy\#L417}
    \begin{align*}
        &\iter\left(n, b, P, X\right) \defs \left(\IF n = 0 \THEN X \ELSE \lfunbp\left(\iter\left(n-1, b, P, X\right)\right)\right)\tag*{(iteration)} \label{def:iter}\\
        &\lfundiff(b,P,X) \defs~\pcchoice{b}{\left(\pseq{P}{X}\right)}{\ufzero} \\
        &\iterdiff\left(n, b, P\right) \defs \left(\IF n = 0 \THEN {\ufone} \ELSE \lfundiff\left(b, P, \iterdiff\left(n-1, b, P\right)\right)\right)\tag*{(iteration difference)} \label{def:iterdiff}
    \end{align*}
\end{definition}
The $\lfundiff$ is similar to $\lfun$ in Definition~\ref{def:lfun} except that if the condition $b$ does not hold, it is $\ufzero$ instead of $\pskip$ in $\lfun$.

{We show that $\iterdiff(n, b, P)$ indeed captures the difference between iterations from top and bottom for any $n$.
\begin{thm}
    \label{thm:iterdiff}
    Provided $P$ is a distribution, that is, $\isfinaldist(P)$. 
    \begin{align*}
    \forall n:\nat \bullet {\lfunbp}^n(\ufone) - {\lfunbp}^n(\ufzero) = \iterdiff(n, b, P)
    \end{align*}
\end{thm}
The proof of this theorem is shown in \ref{appendix:proof_iterdiff}.
}

We show that the iteration from the bottom is an ascending chain and the iteration from the top is a descending chain if $P$ is a distribution.
\begin{thm}
    \label{thm:iter_bot_ascending}
   $\isfinaldist\left(\prrvfunsym{P}\right) \implies \incseq\left(\lambda n @ \iter\left(n, b, P, \ufzero\right)\right)$ 
   \isalink{https://github.com/RandallYe/probabilistic_programming_utp/blob/6a4419b8674b84988065a58696f15093d176594c/probability/probabilistic_relations/utp_prob_rel_lattice_laws.thy\#L4520}
\end{thm}

\begin{thm}
    \label{thm:iter_top_descending}
   $\isfinaldist\left(\prrvfunsym{P}\right) \implies \decseq\left(\lambda n @ \iter\left(n, b, P, \ufone\right)\right)$ 
   \isalink{https://github.com/RandallYe/probabilistic_programming_utp/blob/6a4419b8674b84988065a58696f15093d176594c/probability/probabilistic_relations/utp_prob_rel_lattice_laws.thy\#L4564}
\end{thm}

For an ascending or descending chain $f$, we define $\finstatesasc$ and $\finstatesdes$ to denote there are only finite states to have their suprema or infima different from their initial values $f(0)$.
\begin{definition}[Finite states]
    \label{def:finstates}
    We fix $f: \nat \fun [S]\prhfun$, then define
    \isalink{https://github.com/RandallYe/probabilistic_programming_utp/blob/6a4419b8674b84988065a58696f15093d176594c/probability/probabilistic_relations/utp_prob_rel_lattice.thy\#L401}
    \begin{align*}
        \finstatesasc(f) & \defs \finite\left(\left\{s : S \mid \left(\left(\thnsup n @ f(n,s)\right) > f(0,s) \right)\right\}\right) \\
        \finstatesdes(f) & \defs \finite\left(\left\{s : S \mid \left(\left(\thninf n @ f(n,s)\right) < f(0,s) \right)\right\}\right)
    \end{align*}
\end{definition}
The intuition behind the definitions of $\finstatesasc$ and $\finstatesdes$ is that if an ascending or descending chain $f$ has its supremum or infimum not equal to $f(0)$ for a particular state $s$, then {for any $\varepsilon:\real > 0$, there exists a $m:\nat$ such that} $\left(\thnsup n @ f(n,s) - f(m,s) < \varepsilon\right)$ or $\left(f(m,s) - \thninf n @ f(n,s)  < \varepsilon\right)$. 
%

We show that if $f$ is an ascending chain and {$\finstatesasc(f)$ also holds}, then there exists a $N:\nat$ such that for any $n\geq N$, $f(n,s)$ is close to its supremum in a given distance $\varepsilon : \real > 0$ for any $s$. 
\begin{thm}
    \label{thm:incseq_limit_is_lub_all}
    We fix $f:\nat \fun [S_1,S_2]\prfun$, then 
    \isalink{https://github.com/RandallYe/probabilistic_programming_utp/blob/6a4419b8674b84988065a58696f15093d176594c/probability/probabilistic_relations/utp_prob_rel_lattice_laws.thy\#L3717}
    \begin{align*}
        & \incseq(f) \land \finstatesasc(f) \implies \forall \varepsilon:\real > 0 \bullet \exists N:\nat \bullet \forall n \geq N \bullet \forall s \bullet \left(\thnsup m @ f(m, s)\right) - f(n,s) < \varepsilon 
    \end{align*}
\end{thm}

\begin{figure}[!ht]
    \begin{center}
\begin{tikzpicture}
\begin{axis}[
    xlabel={$n$},
    xmin=0,
    xmax=12,
    x=1.0cm,
    ymax=0.6,
    ylabel={$f(n,s)$},
    y tick label style={/pgf/number format/.cd,fixed,fixed zerofill,precision=2},
    legend pos=outer north east,
]
    \addplot table[header=false,col sep=&,row sep=\\,y expr={\thisrowno{1}}] {
        0 & 0\\
        1 & 0\\
        2 & 1/2 - (2/3)^2\\
        3 & 1/2 - (2/3)^3\\
        4 & 1/2 - (2/3)^4\\
        5 & 1/2 - (2/3)^5\\
        6 & 1/2 - (2/3)^6\\
        7 & 1/2 - (2/3)^7\\
        8 & 1/2 - (2/3)^8\\
        9 & 1/2 - (2/3)^9\\
        10& 1/2 - (2/3)^10\\
        11& 1/2 - (2/3)^11\\
    } node[below right,pos=0.7] {};
    \addlegendentry{$f(n,s_1)$};

    \addplot table[header=false,col sep=&,row sep=\\,y expr={\thisrowno{1}}] {
        0 & 0\\
        1 & 0\\
        2 & 1/3 - (3/5)^2 + 0.1\\
        3 & 1/3 - (3/5)^3\\
        4 & 1/3 - (3/5)^4\\
        5 & 1/3 - (3/5)^5\\
        6 & 1/3 - (3/5)^6\\
        7 & 1/3 - (3/5)^7\\
        8 & 1/3 - (3/5)^8\\
        9 & 1/3 - (3/5)^9\\
        10& 1/3 - (3/5)^10\\
        11& 1/3 - (3/5)^11\\
    } node[below right,pos=0.7] {};
    \addlegendentry{$f(n,s_2)$};

    \addplot table[header=false,col sep=&,row sep=\\,y expr={\thisrowno{1}}] {
        0 & 0\\
        1 & 0\\
        2 & 0.15\\
        3 & 0.15\\
        4 & 0.15\\
        5 & 0.15\\
        6 & 0.15\\
        7 & 0.15\\
        8 & 0.15\\
        9 & 0.15\\
        10& 0.15\\
        11& 0.15\\
    } node[below right,pos=0.7] {};
    \addlegendentry{$f(n,s_3)$};

    \addplot table[header=false,col sep=&,row sep=\\,y expr={\thisrowno{1}}] {
        0 & 0\\
        1 & 0\\
        2 & 0\\
        3 & 0\\
        4 & 0\\
        5 & 0\\
        6 & 0\\
        7 & 0\\
        8 & 0\\
        9 & 0\\
        10& 0\\
        11& 0\\
    } node[above right,pos=0.7] {};
    \addlegendentry{$f(n,s_4)$};

    \draw[blue,densely dashed] (axis cs:-1,0.5) -- node[above left]{\small $\thnsup m @ f(m,s_1)$} (axis cs:12,0.5);
    \draw[darkgray,loosely dashed] (axis cs:-1,0.45) -- node[above left]{\small $\left(\thnsup m @ f(m,s_1)\right) - \varepsilon$} (axis cs:12,0.45);

    \draw[red,densely dashed] (axis cs:-1,1/3) -- node[above left]{\small $\thnsup m @ f(m,s_2)$} (axis cs:12,1/3);
    \draw[darkgray,loosely dashed] (axis cs:-1,1/3-0.05) -- node[above left]{\small $\left(\thnsup m @ f(m,s_2)\right) - \varepsilon$} (axis cs:12,1/3-0.05);

    \draw[brown,densely dashed] (axis cs:-1,0.15) -- node[above right]{\small $\thnsup m @ f(m,s_3)$} (axis cs:12,0.15);
    \draw[darkgray,loosely dashed] (axis cs:-1,0.1) -- node[below right]{\small $\left(\thnsup m @ f(m,s_3)\right) - \varepsilon$} (axis cs:12,0.1);

    \draw[black,dashdotted] (axis cs:8,-1) -- node[above left]{\small $N$} (axis cs:8,0.6);
  \end{axis}
\end{tikzpicture}
    \end{center}
    \caption{Illustration of limits of an increasing chain for various states.{}}
    \label{fig:finitestateasc}
\end{figure}
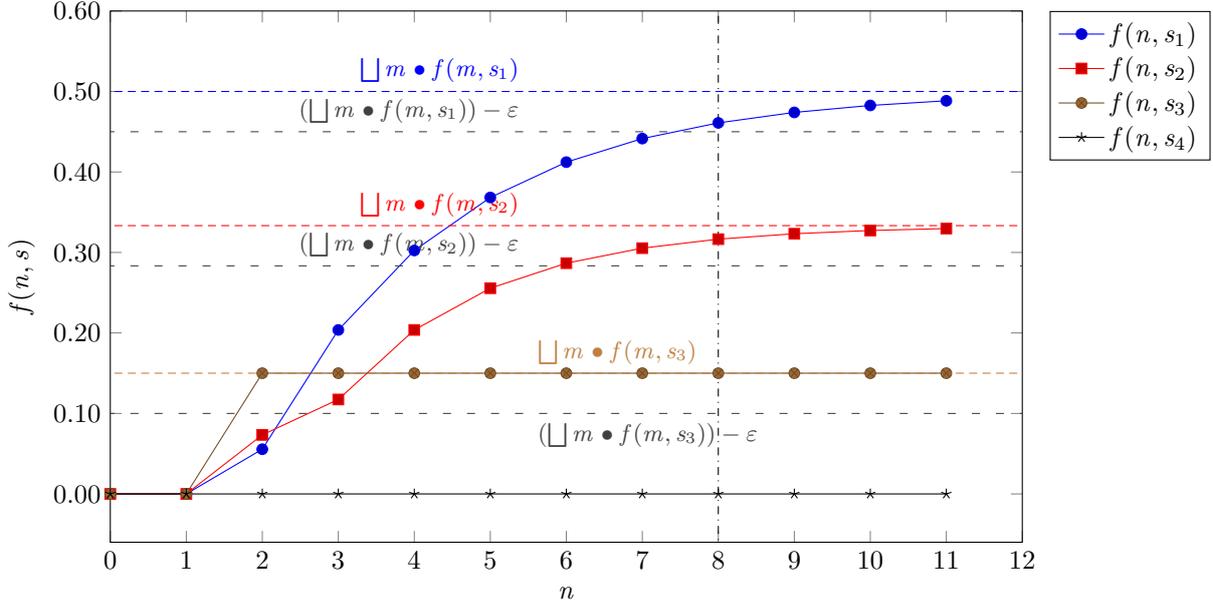

This is explained and illustrated in Fig.~\ref{fig:finitestateasc} where $f$ is a monotonic function, such as $\left(\lambda n @ \iter\left(n, b, P, \ufzero\right)\right)$, whose domain is a complete lattice, and so its limit exists and is the supremum ($\thnsup n@ f(n,s_i)$) of the increasing chain of this function for a particular state $s_i$. In this diagram, we show there are four states (four combinations of the product states $(s,s')$, indeed) in the observation space, denoted as $s_1$,$s_2$,$s_3$, and $s_4$. We draw the function $f$ of the four states for $n$ up to 11, as shown in the diagram as $f(n,s_1)$, $f(n, s_2)$, $f(n, s_3)$, and $f(n, s_4)$. The function for each state has a corresponding supremum, such as $\thnsup n @ f(n,s_1)$ for $s_1$, and a densely dashed line denotes the supremum. We also consider a real number $\varepsilon > 0$, and so {$\varepsilon$ regions (that is, areas between two parallel lines whose distance is $\varepsilon$)} are formed between the densely dashed lines (the supremum) and the loosely dashed line ($\left(\thnsup n @ f(n,s_i)\right) - \varepsilon$). The theorem above shows there is always a $N$ for all $m \geq N$ such that the closeness of $f(m,s_i)$ to its supremum ($\thnsup n@ f(n,s_i)$) is less than $\varepsilon$ for any $s_i$. Because the limit of $f(n,s_i)$ is $\left(\thnsup n @ f(n,s_i)\right)$, there always exists a $N_i$ to satisfy this closeness. For constant zero functions, such as $f(n,s_4)$, $N_i$ is just 0. According to the assumption of finiteness, there are finite states to have their $N_i$ larger than 0, the states $\{s_1, s_2, s_3\}$, in this example, whose $N_i$ are 8, 6, and 2 respectively. We can choose $N$ as the maximum number from this $N_i$ set, and it is 8 (illustrated as a dotted dash line at $x=8$) here. Now for all $n \geq N$ and any state $s$, the function $f(n,s)$ is close to its supremum within the given $\varepsilon$. This theorem is necessary to prove Theorem~\ref{thm:lfun_limit_as_supremum} below and eventually the continuity theorem~\ref{thm:continuity_lfun_bot}.

A descending chain $f$ satisfies the similar theorem below.
\begin{thm}
    \label{thm:decseq_limit_is_glb_all}
    We fix $f:\nat \fun [S_1,S_2]\prfun$, then 
    \isalink{https://github.com/RandallYe/probabilistic_programming_utp/blob/6a4419b8674b84988065a58696f15093d176594c/probability/probabilistic_relations/utp_prob_rel_lattice_laws.thy\#L4081}
    \begin{align*}
        & \decseq(f) \land \finstatesdes(f) \implies \forall \varepsilon:\real > 0 \bullet \exists N:\nat \bullet \forall n \geq N \bullet \forall s \bullet f(n,s) - \left(\thnsup m @ f(m, s)\right) < \varepsilon
    \end{align*}
\end{thm}


From Theorems~\ref{thm:iter_bot_ascending} and~\ref{thm:incseq_limit_is_lub_all}, we prove the following theorem stating that for any state $s$ the limit of the application of $\lfun$ to $\iter$ is the application of $\lfun$ to the supremum of iterations.
\begin{thm}[Limit as supremum]
    \label{thm:lfun_limit_as_supremum}
    \isalink{https://github.com/RandallYe/probabilistic_programming_utp/blob/6a4419b8674b84988065a58696f15093d176594c/probability/probabilistic_relations/utp_prob_rel_lattice_laws.thy\#L4599}
    \begin{align*}
        \begin{array}[]{l}
        \left(
            \isfinaldist\left(\prrvfunsym{P}\right) \land 
            \finstatesasc\left(\lambda n @ \iter\left(n, b, P, \ufzero\right) \right) 
        \right) \implies \\
        \forall s @ \left(\lambda n:\nat @ \lfunbp\left({\iter\left(n, b, P, \ufzero\right)}\right)(s)\right) \tendsto \left(\lfunbp\left({\thnsup n @ \iter\left(n, b, P, \ufzero\right)}\right)(s)\right)
        \end{array}
    \end{align*}
\end{thm}

From Theorems~\ref{thm:iter_top_descending} and~\ref{thm:decseq_limit_is_glb_all}, we prove the following similar theorem stating that for any state $s$ the limit of the application of $\lfun$ to $\iter$ is the application of $\lfun$ to the infimum of iterations.
\begin{thm}[Limit as infimum]
    \label{thm:lfun_limit_as_infimum}
    \isalink{https://github.com/RandallYe/probabilistic_programming_utp/blob/6a4419b8674b84988065a58696f15093d176594c/probability/probabilistic_relations/utp_prob_rel_lattice_laws.thy\#L5308}
    \begin{align*}
        \begin{array}[]{l}
        \left(
            \isfinaldist\left(\prrvfunsym{P}\right) \land 
            \finstatesdes\left(\lambda n @ \iter\left(n, b, P, \ufone\right) \right) 
        \right) \implies \\ 
        \forall s @ \left(\lambda n:\nat @ \lfunbp\left({\iter\left(n, b, P, \ufone\right)}\right)(s)\right) \tendsto \left(\lfunbp\left({\thninf n @ \iter\left(n, b, P, \ufone\right)}\right)(s)\right)
        \end{array}
    \end{align*}
\end{thm}

Continuity of $\lfun$ for iterations from bottom and top is derived from Theorems~\ref{thm:lfun_limit_as_supremum} and~\ref{thm:lfun_limit_as_infimum} and presented below.
\begin{thm}[Continuity - iteration from bottom]
    \label{thm:continuity_lfun_bot}
    \isalink{https://github.com/RandallYe/probabilistic_programming_utp/blob/6a4419b8674b84988065a58696f15093d176594c/probability/probabilistic_relations/utp_prob_rel_lattice_laws.thy\#L4848}
    \begin{align*}
        \begin{array}[]{l}
        \left(
            \isfinaldist\left(\prrvfunsym{P}\right) \land 
            \finstatesasc\left(\lambda n @ \iter\left(n, b, P, \ufzero\right) \right) 
        \right) 
        \implies \lfunbp{\left(\thnsup n @ \iter\left(n, b, P, \ufzero\right)\right)} = \left(\thnsup n @ \iter\left(n, b, P, \ufzero\right)\right)
        \end{array}
    \end{align*}
\end{thm}

\begin{thm}[Continuity - iteration from top]
    \label{thm:continuity_lfun_top}
    \isalink{https://github.com/RandallYe/probabilistic_programming_utp/blob/6a4419b8674b84988065a58696f15093d176594c/probability/probabilistic_relations/utp_prob_rel_lattice_laws.thy\#L5537}
    \begin{align*}
        \begin{array}[]{l}
        \left(
            \isfinaldist\left(\prrvfunsym{P}\right) \land 
            \finstatesdes\left(\lambda n @ \iter\left(n, b, P, \ufone\right) \right) 
        \right) 
        \implies \lfunbp{\left(\thninf n @ \iter\left(n, b, P, \ufone\right)\right)} = \left(\thninf n @ \iter\left(n, b, P, \ufone\right)\right)
        \end{array}
    \end{align*}
\end{thm}

{We show that in \ref{appendix:necessity_finstatesasc} to establish the continuity above, $\finstatesasc\left(\lambda n @ \iter\left(n, b, P, \ufzero\right) \right)$, indeed, is required. The standard non-probabilistic continuity theorem~\cite[Section~5.3]{Nielson2007} does not have the similar requirement because the semantics of sequential composition in the language is \emph{the functional composition over one (deterministic) intermediate state}. The semantics of sequential composition in our language, however, is \emph{the infinite summation over all possible intermediate states.} We discuss this required premise in detail in \ref{appendix:necessity_finstatesasc}.}
The Kleene fixed-point theorem~\ref{thm:kleene_fixed_point_theorem} states the least (or greatest) fixed point of a continuous function $\lfun$ is the supremum (or infimum) of the ascending (or descending) chain of the function. This is just the semantics of while loops. 
\begin{thm}[Least fixed point by construction]
    \label{thm:rec_least_fixed_point}
    \isalink{https://github.com/RandallYe/probabilistic_programming_utp/blob/6a4419b8674b84988065a58696f15093d176594c/probability/probabilistic_relations/utp_prob_rel_lattice_laws.thy\#L5883}
    \begin{align*}
        \begin{array}[]{l}
        \left(
            \isfinaldist\left(\prrvfunsym{P}\right) \land 
            \finstatesasc\left(\lambda n @ \iter\left(n, b, P, \ufzero\right) \right) 
        \right)  
        \implies \pwhile{b}{P} = \left(\thnsup n @ \iter\left(n, b, P, \ufzero\right)\right)
        \end{array}
    \end{align*}
\end{thm}

\begin{thm}[Greatest fixed point by construction]
    \label{thm:rec_great_fixed_point}
    \isalink{https://github.com/RandallYe/probabilistic_programming_utp/blob/6a4419b8674b84988065a58696f15093d176594c/probability/probabilistic_relations/utp_prob_rel_lattice_laws.thy\#L5917}
    \begin{align*}
        \begin{array}[]{l}
        \left(
            \isfinaldist\left(\prrvfunsym{P}\right) \land
            \finstatesdes\left(\lambda n @ \iter\left(n, b, P, \ufone\right) \right) 
        \right)  
        \implies \pwhiletop{b}{P} =\left(\thninf n @ \iter\left(n, b, P, \ufone\right)\right)
        \end{array}
    \end{align*}
\end{thm}
There are several benefits in having the semantics of probabilistic loops constructed by iterations, as shown in Theorems~\ref{thm:rec_least_fixed_point} and~\ref{thm:rec_great_fixed_point}. Essentially, the theorems give the semantics to probabilistic loops theoretically. In practice, they also enable us to compute the semantics by approximation or iterations, for example, in model checking.
Another benefit is facilitating the proof of the unique fixed point theorem.

\begin{thm}[Unique fixed point - finite final states] \label{thm:rec_unique}
    \isalink{https://github.com/RandallYe/probabilistic_programming_utp/blob/6a4419b8674b84988065a58696f15093d176594c/probability/probabilistic_relations/utp_prob_rel_lattice_laws.thy\#L6321}
    \begin{align*}
        & \left(
        \begin{array}[]{ll}
            \isfinaldist\left(\prrvfunsym{P}\right) &\land 
            \finstatesasc\left(\lambda n @ \iter\left(n, b, P, \ufzero\right) \right) \land \\ 
            \left(\forall s @ \left(\lambda n @ \prrvfunsym{\iterdiff\left(n, b, P\right)}(s)\right) \tendsto 0\right) & \land 
            \lfunbp\left(fp\right) = fp
        \end{array} \right) 
        \\
        & \implies \left(\pwhile{b}{P} = fp\right) \land \left(\pwhiletop{b}{P} = fp\right) 
    \end{align*}
\end{thm}
There are four assumptions in the theorem. The third one corresponds to the differences between iterations from top and bottom, illustrated as dashed lines in Fig.~\ref{fig:coin_t_prob_iteration}. If for any state $s$, the difference tends to 0, then the least fixed point by Theorem~\ref{thm:rec_least_fixed_point} and the greatest fixed point by Theorem~\ref{thm:rec_great_fixed_point} coincide. The fourth assumption states $fp$ is a fixed point of $\lfun$. The conclusion states that both the least fixed point and the greatest fixed point are just $fp$.  
This theorem largely simplifies the proof obligation for reasoning about loops to establish these assumptions.

The second assumption $\finstatesasc\left(\lambda n @ \iter\left(n, b, P, \ufzero\right) \right)$ restricts the application of the theorem above and previous theorems to a limited subset of probabilistic programs. For example, the program with a time variable $t$ to model the dice example, described in Sect.~\ref{sec:relwork}, does not satisfy the assumption because {the set of final states with positive probabilities is infinite}. We have proved more general theorems to support wider probabilistic programs, including those with a time variable. First, we define $\finitefinal$ to characterise a program that always produces {finitely many final states}.
\begin{align*}
    \finitefinal(P) \defs \forall s @ \finite \left\{s':S | P(s, s') > 0\right\}
\end{align*}

$P$ is $\finitefinal$ if for any initial state $s$, $P$ has {finitely many reachable states.} The assumption $\finstatesasc$ or $\finstatesdes$ in Theorems~\ref{thm:lfun_limit_as_supremum} to~\ref{thm:rec_unique} is now replaced by $\finitefinal(P)$. The conclusions of these theorems are still valid. We present the updated unique fixed theorem {below} and omit others here for {brevity}.

\begin{thm}[Unique fixed point - finite final states for each iteration] \label{thm:rec_unique_fin}
    \isalink{https://github.com/RandallYe/probabilistic_programming_utp/blob/6a4419b8674b84988065a58696f15093d176594c/probability/probabilistic_relations/utp_prob_rel_lattice_laws.thy\#L6417}
    \begin{align*}
        & \left(
        \begin{array}[]{l}
            \isfinaldist\left(\prrvfunsym{P}\right) \land 
            \finitefinal(P) \land 
            \left(\forall s @ \left(\lambda n @ \prrvfunsym{\iterdiff\left(n, b, P\right)}(s)\right) \tendsto 0\right) \land 
            \lfunbp\left(fp\right) = fp
        \end{array} \right) 
        \\
        & \implies \left(\pwhile{b}{P} = fp\right) \land \left(\pwhiletop{b}{P} = fp\right) 
    \end{align*}
\end{thm}

\section{Examples and case studies}
\label{sec:ex_cases}

\subsection{Doctor Who's Tardis Attack}

Two robots, the Cyberman C and the Dalek D, attack Doctor Who's Tardis once a day between them. C has a probability of 1/2 of a successful attack, while D has a probability of 3/10 of a successful attack. C attacks more often than D, with a probability of 3/5 on a particular day (and so D attacks with a probability of 2/5 on that day). What is the probability that there will be a successful attack today?

We model the problem in our probabilistic programming in the definition below.
\begin{definition}[Doctor Who's Tardis Attack]
  \label{def:dwta}
  \isalink{https://github.com/RandallYe/probabilistic_programming_utp/blob/6a4419b8674b84988065a58696f15093d176594c/probability/probabilistic_relations/Examples/utp_prob_rel_lattice_dwta.thy\#L22}
  \begin{align*}
    & Attacker ::= C | D \\ 
    & Status :: = S | F \\ 
    & \isakwmaj{alphabet}\ dwtastate = r::Attacker \qquad a::Status \\
    & dwta \defs 
      \ppchoice{3/5} 
      {\left(\pseq{\left(\passign{r}{C}\right)}{\left(\ppchoice{1/2}{\passign{a}{S}}{\passign{a}{F}}\right)}\right)}
      {\left(\pseq{\left(\passign{r}{D}\right)}{\left(\ppchoice{3/10}{\passign{a}{S}}{\passign{a}{F}}\right)} \right) }
  \end{align*}
\end{definition}
We define the attackers $C$ and $D$ of type $Attacker$, and $S$ and $F$ of type $Status$ for a successful or failed attacker. The observation space of this program is $dwtastate$ containing two variables $r$ and $a$ to record the attacker and the attack status. The problem is modelled as a program $dwta$, composed of probabilistic choice, assignment, and sequential composition.

Using the reasoning framework and the algebraic laws mechanised in Isabelle, $dwta$ is simplified and proved semantically equal to a probabilistic program shown below.
\begin{thm}[Simplified program]
  \isalink{https://github.com/RandallYe/probabilistic_programming_utp/blob/6a4419b8674b84988065a58696f15093d176594c/probability/probabilistic_relations/Examples/utp_prob_rel_lattice_dwta.thy\#L129}
  \begin{align*}
    & dwta = \rvprfunsym{\usexpr{
      \begin{array}[]{l}
        3/10 * \ibracket{r' = C \land a' = S} + 
        3/10 * \ibracket{r' = C \land a' = F} + \\
        6/50 * \ibracket{r' = D \land a' = S} + 
        14/50 * \ibracket{r' = D \land a' = F} 
      \end{array}
    }} 
  \end{align*}
\end{thm}
This law shows $C$ has a probability of 3/10 of a successful or failed attack, while $D$ has a probability of 6/50 of a successful attack and 14/50 of a failed attack. We note that this simplified program is a distribution of the final state because the sum of the probabilities of these combinations equals 1: $3/10+3/10+6/50+14/50=1$.

With this simplified program, we can use it to answer interesting quantitative queries using sequential composition. The answer to the question ``What is the probability of a successful attack?'', for example, is 21/50.

\begin{thm}
  \label{thm:dwta_success}
  $ \pseq{\prrvfunsym{dwta}}{{\ibracket{a = S}}} = {\usexpr{21/50}}$
  \isalink{https://github.com/RandallYe/probabilistic_programming_utp/blob/6a4419b8674b84988065a58696f15093d176594c/probability/probabilistic_relations/Examples/utp_prob_rel_lattice_dwta.thy\#L167}
\end{thm}

\subsection{The Monty Hall problem}


We model the problem in the program below, where the doors are numbered as natural numbers: 0, 1, and 2.

\begin{definition}[Monty hall]
  \isalink{https://github.com/RandallYe/probabilistic_programming_utp/blob/6a4419b8674b84988065a58696f15093d176594c/probability/probabilistic_relations/Examples/utp_prob_rel_lattice_monty_hall.thy\#L14}
  \begin{align*}
    \isakwmaj{alphabet}\ &mhstate = p::\nat \qquad c::\nat \qquad m::\nat\\
    init & \defs \pseq{\uniformdist{p}{\{0 \upto 2\}}}{\uniformdist{c}{\{0 \upto 2\}}} \\
    mc & \defs {\left(\ppchoice{1/2}{\left(\passign{m}{(c+1) \mod 3}\right)}{\left(\passign{m}{(c+2) \mod 3}\right)}\right)} \\
    mha & \defs \pcchoice{c = p}{mc}{\passign{m}{3 - c - p}} \\
    mha\_nc & \defs \pseq{\pseq{init}{mha}}{\pskip}\tag*{(no change strategy)} \label{thm:monty_mha_nc}\\
    mha\_c & \defs \pseq{init}{\pseq{mha}{\passign{c}{3 - c -m}}}\tag*{(change strategy)} \label{thm:monty_mha_c}
  \end{align*}
\end{definition}
The observation space $mhstate$ contains three variables of type $\nat$: $p$ for the number of the prize door, $c$ for the contestant's choice, and $m$ for the door Monty opens. The $init$ is the initial configuration of the problem where the values of $p$ and $c$ follow a uniform distribution from an interval between 0 and 2 inclusive, so the prize and the contestant's choice are random.  The $mha$ models if the contestant's choice is the prize ($c=p$), the Monty randomly (with probability 1/2) chooses $m$ from the other two doors, denoted as $(c+1) \mod 3$ and $(c+2) \mod 3$. Otherwise, the prize is not revealed, and then Monty chooses the one that is not $p$ (he knows the value of $p$, the prize door) where $3-c-p$ guarantees that $m$ is different from both $c$ and $p$.  The $mha\_nc$ models a no-change strategy, and the $mha\_c$ models a change strategy after Monty reveals one.

We simplify $init$ according to the law below.
\begin{thm}[Initial]
  $init = \rvprfunsym{ \usexpr{ \ibracket{p' \in \{0 \upto 2\}} * \ibracket{c' \in \{0 \upto 2\}} * \ibracket{m' = m} / 9 }}$
  \isalink{https://github.com/RandallYe/probabilistic_programming_utp/blob/6a4419b8674b84988065a58696f15093d176594c/probability/probabilistic_relations/Examples/utp_prob_rel_lattice_monty_hall.thy\#L143}
\end{thm}
In the initial configuration, the combinations of $p$ and $c$ have an equal probability $1/9$ with $m$ unchanged. The $init$ is also a distribution: the summation over its final states equals 1.

The $mha\_nc$ is proved to be equal to the program below.

\begin{thm}[No change strategy]
  \isalink{https://github.com/RandallYe/probabilistic_programming_utp/blob/6a4419b8674b84988065a58696f15093d176594c/probability/probabilistic_relations/Examples/utp_prob_rel_lattice_monty_hall.thy\#L550}
  \begin{align*}
    & mha\_nc =  \rvprfunsym{
      \usexpr{
      \begin{array}[]{l}
        \ibracket{c' = p'} * \ibracket{p' \in \{0 \upto 2\}} * \ibracket{c' \in \{0 \upto 2\}} * \ibracket{m' = (c'+1)\mod 3} / 18\ + \\
        \ibracket{c' = p'} * \ibracket{p' \in \{0 \upto 2\}} * \ibracket{c' \in \{0 \upto 2\}} * \ibracket{m' = (c'+2)\mod 3} / 18\ + \\
        \ibracket{c' \neq p'} * \ibracket{p' \in \{0 \upto 2\}} * \ibracket{c' \in \{0 \upto 2\}} * \ibracket{m' = 3 - c' - p'} / 9 \\
      \end{array}
    }}
  \end{align*}
\end{thm}

The $mha\_c$ is also proved to be equal to the program below.

\begin{thm}[Change strategy]
  \isalink{https://github.com/RandallYe/probabilistic_programming_utp/blob/6a4419b8674b84988065a58696f15093d176594c/probability/probabilistic_relations/Examples/utp_prob_rel_lattice_monty_hall.thy\#L1273}
  \begin{align*}
    & mha\_c =  \rvprfunsym{
      \usexpr{
      \begin{array}[]{l}
        \ibracket{p' \in \{0 \upto 2\}} * \ibracket{c' = 3 - p' - m'} * \ibracket{m' = (p'+1)\mod 3} / 18\ + \\
        \ibracket{p' \in \{0 \upto 2\}} * \ibracket{c' = 3 - p' - m'} * \ibracket{m' = (p'+2)\mod 3} / 18\ + \\
        \ibracket{c' = p'} * \ibracket{p' \in \{0 \upto 2\}} * \ibracket{3 - p' - m' \neq p'} * \ibracket{3 - p' - m' \leq 2} *  \ibracket{3 - p' - m' \geq 0} / 9 \\
      \end{array}
    }}
  \end{align*}
\end{thm}

With these laws, we can answer questions like the probability of winning for each strategy and whether you change the choice. 
\begin{thm}[Winning probability]
  \hfill \isaref{https://github.com/RandallYe/probabilistic_programming_utp/blob/6a4419b8674b84988065a58696f15093d176594c/probability/probabilistic_relations/Examples/utp_prob_rel_lattice_monty_hall.thy\#L928}
  and 
  \isaref{https://github.com/RandallYe/probabilistic_programming_utp/blob/6a4419b8674b84988065a58696f15093d176594c/probability/probabilistic_relations/Examples/utp_prob_rel_lattice_monty_hall.thy\#L1633}
  \begin{align*}
    & \pseq{\prrvfunsym{mha\_nc}}{{\ibracket{c = p}}} = {\usexpr{1/3}} \tag*{(winning probability of no-change strategy)} \label{thm:monty_nc_winning} \\
    & \pseq{\prrvfunsym{mha\_c}}{{\ibracket{c = p}}} = {\usexpr{2/3}}  \tag*{(winning probability of change strategy)} \label{thm:monty_c_winning} 
  \end{align*}
\end{thm}
The above law shows that the winning probabilities are $1/3$ for the no-change strategy and $2/3$ for the change strategy, so you should change the choice because you have a {higher} probability of winning.

\subsection{The forgetful Monty}
\label{ssec:cases_forgetful_monty}
The new problem is modelled below.
\begin{definition}[Forgetful Monty]
  \label{def:forgetful_monty}
  \isalink{https://github.com/RandallYe/probabilistic_programming_utp/blob/6a4419b8674b84988065a58696f15093d176594c/probability/probabilistic_relations/Examples/utp_prob_rel_lattice_monty_hall.thy\#L2000}
  \begin{align*}
    & forgetful\_monty \defs \pseq{init}{mc} \tag*{(the forgetful Monty)} \label{def:monty_forgetful}\\
    & learn\_fact \defs \pparallel{{forgetful\_monty}}{\rvprfunsym{\ibracket{m' \neq p'}}} \tag*{(learn new fact that the prize is not revealed)} \label{def:monty_forgetful_learn}
  \end{align*}
\end{definition}
After initialisation, the forgetful Monty randomly chooses one from the other two doors. The learned new fact that the prize is not revealed ($m'\neq p'$) is fed into the program by parallel composition as shown in the definition of $learn\_fact$. The program equals the one below, and the winning probability is queried.
\begin{thm}
  \hfill \isaref{https://github.com/RandallYe/probabilistic_programming_utp/blob/6a4419b8674b84988065a58696f15093d176594c/probability/probabilistic_relations/Examples/utp_prob_rel_lattice_monty_hall.thy\#L2615} and \isaref{https://github.com/RandallYe/probabilistic_programming_utp/blob/6a4419b8674b84988065a58696f15093d176594c/probability/probabilistic_relations/Examples/utp_prob_rel_lattice_monty_hall.thy\#L2424} 
  \begin{align*}
    & learn\_fact = \rvprfunsym{
      \usexpr{
      \begin{array}[]{l}
        \ibracket{p' \in \{0 \upto 2\}} * \ibracket{c' \in \{0 \upto 2\}} * \ibracket{m' = (c'+1)\%3} * \ibracket{m' \neq p'} / 12\ + \\
        \ibracket{p' \in \{0 \upto 2\}} * \ibracket{c' \in \{0 \upto 2\}} * \ibracket{m' = (c'+2)\%3} * \ibracket{m' \neq p'} / 12
      \end{array}
    }}\\
    & \pseq{\prrvfunsym{learn\_fact}}{{\ibracket{c = p}}} = {\usexpr{1/2}} \tag*{(winning probability of learning new fact)} \label{thm:monty_learn_winning}
  \end{align*}
\end{thm}
The probability of winning is now $1/2$ and so if Monty is forgetful, and the contestant happens to choose a door with no prize, it does not matter whether the contestant sticks or switches because they have the equal probability of $1/2$.

\subsection{Robot localisation}
\label{ssec:cases_robot_localisation}
The likelihood functions are defined below.

\begin{definition}[Likelihood functions]
    \isalink{https://github.com/RandallYe/probabilistic_programming_utp/blob/6a4419b8674b84988065a58696f15093d176594c/probability/probabilistic_relations/Examples/utp_prob_rel_lattice_robot_localisation.thy\#L24}
\begin{align*}
    & scale\_door \defs  \usexpr{3 * \ibracket{door(bel')} + 1} \\
    & scale\_wall \defs  \usexpr{3 * \ibracket{\lnot door(bel')} + 1} 
\end{align*}
\end{definition}

We are interested in questions like how many measurements and moves are necessary to estimate the robot's location accurately.

\subsubsection{Initialisation}
Initially, the robot is randomly placed, and so a uniform distribution. This is defined below by the program $init$.

\begin{definition}[Initialisation]
    $ init \defs \uniformdist{bel}{\{0 \upto 2\}} $ 
    \isalink{https://github.com/RandallYe/probabilistic_programming_utp/blob/6a4419b8674b84988065a58696f15093d176594c/probability/probabilistic_relations/Examples/utp_prob_rel_lattice_robot_localisation.thy\#L20}
\end{definition}

\subsubsection{First sensor reading}

The sensor detects a door. We learn new knowledge and update our beliefs accordingly using parallel composition.  The prior probability distribution (prior) is $init$, and the likelihood function is $scale\_door$. The posterior probability distribution is given in the theorem below.  \begin{thm}[First posterior]
  $ \pparallel{init}{scale\_door} = \rvprfunsym{4/9 * \ibracket{bel' = 0} + 1/9 * \ibracket{bel' = 1} + 4/9 *  \ibracket{bel' = 2}}$ 
  \isalink{https://github.com/RandallYe/probabilistic_programming_utp/blob/6a4419b8674b84988065a58696f15093d176594c/probability/probabilistic_relations/Examples/utp_prob_rel_lattice_robot_localisation.thy\#L91}
\end{thm}
We have a high probability of $4/9$ to believe the robot is in front of a door at position 0 or 2. 

\subsubsection{Move one space to the right}
Now the robot takes an action to move one space to the right, and the belief position is shifted one to the right. This is defined as $move\_right$ below.
\begin{definition}[Move to the right]
  $move\_right \defs \left(\passign{bel}{(bel + 1) \mod 3}\right) $
  \isalink{https://github.com/RandallYe/probabilistic_programming_utp/blob/6a4419b8674b84988065a58696f15093d176594c/probability/probabilistic_relations/Examples/utp_prob_rel_lattice_robot_localisation.thy\#L30}
\end{definition}

An action updates the belief using sequential composition. The posterior probability distribution after the move is given as follows.  \begin{thm}[Posterior after the first move]
  \isalink{https://github.com/RandallYe/probabilistic_programming_utp/blob/6a4419b8674b84988065a58696f15093d176594c/probability/probabilistic_relations/Examples/utp_prob_rel_lattice_robot_localisation.thy\#L116}
  \begin{align*}
    & \pseq{\left(\pparallel{init}{scale\_door}\right)}{move\_right}  = \rvprfunsym{4/9 * \ibracket{bel' = 0} + 4/9 * \ibracket{bel' = 1} + 1/9 *  \ibracket{bel' = 2}} 
  \end{align*}
\end{thm}
We observe that the probability values are not changing, but the positions are shifted in the distribution.

\subsubsection{Second sensor reading}
The sensor detects a door again. The posterior probability distribution is updated accordingly. 
\begin{thm}[Second posterior]
  \isalink{https://github.com/RandallYe/probabilistic_programming_utp/blob/6a4419b8674b84988065a58696f15093d176594c/probability/probabilistic_relations/Examples/utp_prob_rel_lattice_robot_localisation.thy\#L514}
  \begin{align*}
    & \pparallel{\left(\pseq{\left(\pparallel{init}{scale\_door}\right)}{move\_right}\right)}{scale\_door} = \rvprfunsym{2/3 * \ibracket{bel' = 0} + 1/6 * \ibracket{bel' = 1} + 1/6 *  \ibracket{bel' = 2}} 
  \end{align*}
\end{thm}
We have a high probability of $2/3$ to believe the robot is in front of a door at position 0 and a low probability of $1/6$ in the other two positions.

\subsubsection{Move one space to the right}

Another action is to move the robot to its right, and the posterior probability distribution is shifted accordingly.

\begin{thm}[Posterior after the second move]
  \isalink{https://github.com/RandallYe/probabilistic_programming_utp/blob/6a4419b8674b84988065a58696f15093d176594c/probability/probabilistic_relations/Examples/utp_prob_rel_lattice_robot_localisation.thy\#L603}
  \begin{align*}
    & \pseq{\left(\pparallel{\left(\pseq{\left(\pparallel{init}{scale\_door}\right)}{move\_right}\right)}{scale\_door}\right)}{move\_right} \\
    = \,& \rvprfunsym{1/6 * \ibracket{bel' = 0} + 2/3 * \ibracket{bel' = 1} + 1/6 *  \ibracket{bel' = 2}} 
  \end{align*}
\end{thm}

\subsubsection{Third sensor reading}

The learn sensor detects a wall. The posterior probability distribution is updated accordingly. 

\begin{thm}[Third posterior]
  \isalink{https://github.com/RandallYe/probabilistic_programming_utp/blob/6a4419b8674b84988065a58696f15093d176594c/probability/probabilistic_relations/Examples/utp_prob_rel_lattice_robot_localisation.thy\#L994}
  \begin{align*}
    & \pparallel{\left(\pseq{\left(\pparallel{\left(\pseq{\left(\pparallel{init}{scale\_door}\right)}{move\_right}\right)}{scale\_door}\right)}{move\_right}\right)}{scale\_wall} \\
    = \,& \rvprfunsym{1/18 * \ibracket{bel' = 0} + 8/9 * \ibracket{bel' = 1} + 1/18 *  \ibracket{bel' = 2}} 
  \end{align*}
\end{thm}
After three sensor readings and two moves, our beliefs about the robot's position are $8/9$ at position 1 and $1/18$ at position 0 or 2. We plot the beliefs in Fig.~\ref{fig:robot_localisation_belief} where each position has six updates corresponding to the initial prior (I1), the three sensor readings (door - D2 and D4, and wall - W6), and the two moves to the right (M3 and M5).  The diagram shows that the difference between the highest and lowest probability for each update becomes big or stays the same: from 0 for I1 to 15/18 ($=8/9-1/18$) for W6. So the robot gains more knowledge in each update.  From the diagram, we are confident of (probability $8/9$) the robot's localisation after three measurements and two moves.
\makeatletter
\pgfplotsset{
  calculate offset/.code={
    \pgfkeys{/pgf/fpu=true,/pgf/fpu/output format=fixed}
    \pgfmathsetmacro\testmacro{(\pgfplotspointmeta*10^\pgfplots@data@scale@trafo@EXPONENT@y)*\pgfplots@y@veclength)}
    \pgfkeys{/pgf/fpu=false}
  },
  every node near coord/.style={
    /pgfplots/calculate offset,
    yshift=-\testmacro,
  },
  name node/.style={
    every node near coord/.append style={
      name=#1-\coordindex
    }}
}    
\pgfplotstableread{
  0 0.3333 0.4444	0.4444	0.6667	0.1667	0.0556	
  1 0.3333 0.1111	0.4444	0.1667	0.6667	0.8889	
  2 0.3333 0.4444 0.1111  0.1667  0.1667  0.0556
}\dataset
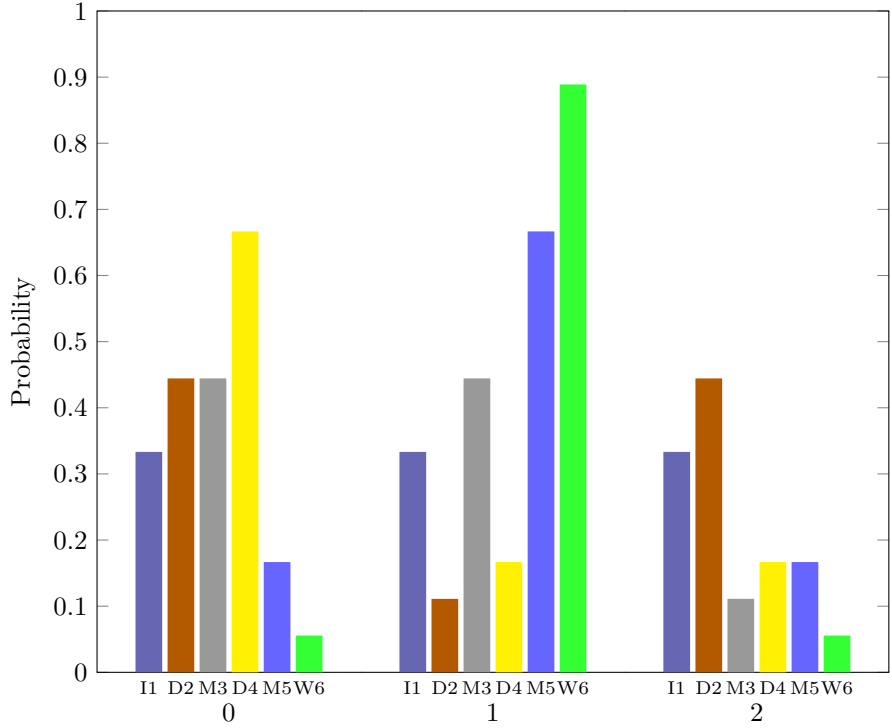
\begin{figure}[!ht]
  \begin{center}
    \begin{tikzpicture}
      \begin{axis}[ybar,
        width=12cm,
        ymin=0,
        ymax=1,
        xmin=-0.5,
        xmax=2.5,
        ylabel={Probability},
        xtick=data,
        xticklabels = {
          0,
          1,
          2,
        },
        xticklabel style={yshift=-1ex},
        major x tick style = {opacity=0},
        minor x tick num = 1,
        minor tick length=0ex,
        every node near coord/.append style={ anchor=north,font=\scriptsize }
        ]
        \addplot[draw=none,fill=blue!50!black!60, name node=1, nodes near coords=I1] table[x index=0,y index=1] \dataset; 
        \addplot[draw=none,fill=orange!70!black, name node=2, nodes near coords=D2] table[x index=0,y index=2] \dataset;
        \addplot[draw=none,fill=gray!80, name node=3, nodes near coords=M3] table[x index=0,y index=3] \dataset;
        \addplot[draw=none,fill=yellow, name node=4, nodes near coords=D4] table[x index=0,y index=4] \dataset;
        \addplot[draw=none,fill=blue!60, name node=5, nodes near coords=M5] table[x index=0,y index=5] \dataset;
        \addplot[draw=none,fill=green!80, name node=6, nodes near coords=W6] table[x index=0,y index=6] \dataset;
      \end{axis}
    \end{tikzpicture}
  \end{center}
  \caption{The update of the robot's belief at different positions after three measurements and two moves with a prior.}
  \label{fig:robot_localisation_belief}
\end{figure}

\subsection{Classification - COVID-19 diagnosis}
\label{ssec:cancer_diagnosis}


We define the state space $cdstate$ of this example as follows.
\begin{definition}[State space]
  \isalink{https://github.com/RandallYe/probabilistic_programming_utp/blob/6a4419b8674b84988065a58696f15093d176594c/probability/probabilistic_relations/Examples/machine_learning_examples/utp_prob_rel_cancer_diagnosis.thy\#L24}
  \begin{align*}
    & CovidTest ::= Pos | Neg \\ 
    &\isakwmaj{alphabet}\ cdstate = c::\bool \qquad ct::CovidTest
  \end{align*}
\end{definition}
Whether a person has COVID or not is recorded in a boolean variable $c$: true for COVID and false for no COVID. The test result is recorded in a variable $ct$ of type $CovidTest$ whose value could be $Pos$itive or $Neg$ative.

The prior probability of a randomly selected person having COVID is $p_1$ of type $\ureal$. The prior probability distribution, therefore, is a probabilistic choice, defined below.
\begin{definition}[Prior]
  $Init \defs \ppifchoice{p_1}{\passign{c}{True}}{\passign{c}{False}}$
\end{definition}
So the probability of a person having COVID is $p_1$ and having no COVID is $(1-p_1)$.

The test is imperfect. Its sensitivity (true positive) is $p_2$, and specificity (true negative) is $1-p_3$. It means if a person with COVID is tested, the probability of a positive result is $p_2$, and if a person without COVID is tested, the probability of a negative result is $1-p_3$. We, therefore, define the action of a test as below.

\begin{definition}[Test]
  $TestAction \defs \pcchoice{c}
  {(\ppchoice{p_2}{\passign{ct}{Pos}}{\passign{ct}{Neg}})}
  {(\ppchoice{p_3}{\passign{ct}{Pos}}{\passign{ct}{Neg}})}$ 
  \isalink{https://github.com/RandallYe/probabilistic_programming_utp/blob/6a4419b8674b84988065a58696f15093d176594c/probability/probabilistic_relations/Examples/machine_learning_examples/utp_prob_rel_cancer_diagnosis.thy\#L58}
\end{definition}
It is a conditional choice between two probabilistic choices, defining the probabilities of true positive, false negative, false positive, and true negative. A positive test result is a new learned evidence, defined as follows.
\begin{definition}
  $TestResPos \defs \ibracket{ct' = Pos}$ 
  \isalink{https://github.com/RandallYe/probabilistic_programming_utp/blob/6a4419b8674b84988065a58696f15093d176594c/probability/probabilistic_relations/Examples/machine_learning_examples/utp_prob_rel_cancer_diagnosis.thy\#L67}
\end{definition}
It, essentially, is an Iverson bracket expression of a relation $ct'=Pos$ stating that the test result $ct$ is positive.

We conduct the first test and learn its positive result, modelled as the program below.
\begin{definition}[First test]
  $FirstTestPos \defs \pparallel{\left(\pseq{Init}{TestAction}\right)}{TestResPos}$
  \isalink{https://github.com/RandallYe/probabilistic_programming_utp/blob/6a4419b8674b84988065a58696f15093d176594c/probability/probabilistic_relations/Examples/machine_learning_examples/utp_prob_rel_cancer_diagnosis.thy\#L110}
\end{definition}
As usual, the action $TestAction$ is sequentially composed, and the evidence is learned in parallel.  We show the posterior in the theorem below.
\begin{thm}[Posterior after the first test]
  \isalink{https://github.com/RandallYe/probabilistic_programming_utp/blob/6a4419b8674b84988065a58696f15093d176594c/probability/probabilistic_relations/Examples/machine_learning_examples/utp_prob_rel_cancer_diagnosis.thy\#L390}
  \begin{align*}
    FirstTestPos = \usexpr{\left(
    \begin{array}[]{l}
      \ibracket{c'}*\ibracket{ct'=Pos}*p_1*p_2 + \\
      \ibracket{\lnot c'}*\ibracket{ct'=Pos}*(1-p_1)*p_3
    \end{array}
    \right) / \left(p_1*p_2+(1-p_1)*p_3\right)}
  \end{align*}
\end{thm}
From the theorem, we know that the probability that the person has COVID, given a positive test, is 
\begin{align*}
  &\left(p_1*p_2\right)/\left(p_1*p_2+(1-p_1)*p_3\right)
\end{align*}
Provided $p_1=0.002$, $p_2=0.89$, and $p_3=0.05$, the probability of the person having COVID is 0.0344 (much less than we may think?), and so the probability without COVID is 0.9656.

If a second test is conducted, the result is still positive, and so new evidence is learned again.
\begin{definition}[Second test]
  $SecondTestPos \defs \pparallel{\left(\pseq{FirstTestPos}{TestAction}\right)}{TestResPos}$
  \isalink{https://github.com/RandallYe/probabilistic_programming_utp/blob/6a4419b8674b84988065a58696f15093d176594c/probability/probabilistic_relations/Examples/machine_learning_examples/utp_prob_rel_cancer_diagnosis.thy\#L132}
\end{definition}
We show the posterior after the second test with the learned fact in the theorem below. 
\begin{thm}[Posterior after the second test]
  \isalink{https://github.com/RandallYe/probabilistic_programming_utp/blob/6a4419b8674b84988065a58696f15093d176594c/probability/probabilistic_relations/Examples/machine_learning_examples/utp_prob_rel_cancer_diagnosis.thy\#L661}
  \begin{align*}
    SecondTestPos = \usexpr{\left(
    \begin{array}[]{l}
      \ibracket{c'}*\ibracket{ct'=Pos}*p_1*p_2^2 + \\
      \ibracket{\lnot c'}*\ibracket{ct'=Pos}*(1-p_1)*p_3^2
    \end{array}
    \right) / \left(p_1*p_2^2+(1-p_1)*p_3^2\right)}
  \end{align*}
\end{thm}
From the theorem, we know that the probability that the person has COVID, given a positive test, is 
\begin{align*}
  &\left(p_1*p_2^2\right)/\left(p_1*p_2^2+(1-p_1)*p_3^2\right)
\end{align*}
Provided $p_1=0.002$, $p_2=0.89$, and $p_3=0.05$, the probability of the person having COVID is 0.3884, so the probability without COVID is 0.6116.  With the second test, it is more likely (38.84\% vs. 3.44\%) that the person may have COVID. In this case, a second test should be conducted, given the first test is positive.

\subsection{(Parametrised) coin flip}
\label{ssec:cases_coin}

The program $flip$ in Definition~\ref{def:coin_flip} is for an unbiased coin (a Bernoulli distribution with $p=1/2$), and $pflip$ below defines a parametrised program where the parameter $p$ denotes the probability to have its outcome as heads (a Bernoulli distribution with probability $p$). So it could be a biased coin.

\begin{definition}[Parametrised coin]
  $ pflip (p)\defs \pwhile{c = tl}{\ppchoice{p}{\passign{c}{hd}}{\passign{c}{tl}}} $
  \isalink{https://github.com/RandallYe/probabilistic_programming_utp/blob/6a4419b8674b84988065a58696f15093d176594c/probability/probabilistic_relations/Examples/utp_prob_rel_lattice_coin.thy\#L485}
\end{definition}

Both $flip$ and $pflip$ contain probabilistic loops. We use the unique fixed point theorem~\ref{thm:rec_unique} to give semantics to them where $P$ in the theorem is $cflip$ here for $flip$. Previously, we have shown that $cflip$ is a distribution. And obviously, the observation space $cstate$ is finite (two elements $hd$ and $tl$), so the product $cstate \cross cstate$ is also finite.  We also show that the differences in iterations from top and bottom tend to 0, which is illustrated in Fig.~\ref{fig:coin_t_prob_iteration}.
\begin{thm}
  \label{thm:cflip_iterdiff_0}
  $\forall s:cstate\cross cstate @ \left(\lambda n @ \prrvfunsym{\iterdiff\left(n, c=tl, cflip\right)}(s)\right) \tendsto 0$ 
  \isalink{https://github.com/RandallYe/probabilistic_programming_utp/blob/6a4419b8674b84988065a58696f15093d176594c/probability/probabilistic_relations/Examples/utp_prob_rel_lattice_coin.thy\#L389}
\end{thm}
Additionally, $\rvprfunsym{\ibracket{c'=hd}}$ is a fixed point of the loop function.
\begin{thm}
  \label{thm:cflip_fp}
  $\lfun^{c = tl}_{cflip}\left(\rvprfunsym{\ibracket{c'=hd}}\right) = \rvprfunsym{\ibracket{c'=hd}}$
  \isalink{https://github.com/RandallYe/probabilistic_programming_utp/blob/6a4419b8674b84988065a58696f15093d176594c/probability/probabilistic_relations/Examples/utp_prob_rel_lattice_coin.thy\#L407}
\end{thm}
All four assumptions of Theorem~\ref{thm:rec_unique} are now established. The $flip$, therefore, is semantically (surprisingly) just the fixed point $\rvprfunsym{\ibracket{c'=hd}}$.
\begin{thm}
  \label{thm:flip_sem}
  $flip = \rvprfunsym{\ibracket{c' = hd}}$
  \isalink{https://github.com/RandallYe/probabilistic_programming_utp/blob/6a4419b8674b84988065a58696f15093d176594c/probability/probabilistic_relations/Examples/utp_prob_rel_lattice_coin.thy\#L432}
\end{thm}
The $flip$ terminates almost surely and is almost impossible for non-termination.
\begin{thm}
  \label{thm:flip_sem_termination}
  \isalink{https://github.com/RandallYe/probabilistic_programming_utp/blob/6a4419b8674b84988065a58696f15093d176594c/probability/probabilistic_relations/Examples/utp_prob_rel_lattice_coin.thy\#L447}
  \begin{align*}
    & \pseq{\prrvfunsym{flip}}{{\ibracket{c = hd}}} = {\usexpr{1}} \tag*{(termination probability)} \label{thm:coin_flip_term}\\
    & \pseq{\prrvfunsym{flip}}{{\ibracket{\lnot c = hd}}} = {\usexpr{0}} \tag*{(non-termination probability)} \label{thm:coin_flip_nonterm}
  \end{align*}
\end{thm}
This theorem shows the probability of the final state $c'$ being $hd$ is 1 and is not $hd$ is 0. This is equivalent to the termination probability.

We also show $pflip(p)$ is semantically equal to $flip$ if $p$ is not 0.
\begin{thm}
  \label{thm:pflip_sem}
  $ p \neq 0 \implies pflip(p) = \rvprfunsym{\ibracket{c' = hd}}$ 
  \isalink{https://github.com/RandallYe/probabilistic_programming_utp/blob/6a4419b8674b84988065a58696f15093d176594c/probability/probabilistic_relations/Examples/utp_prob_rel_lattice_coin.thy\#L655}
\end{thm}
If $p$ is 0, we know this program $pflip$ is non-terminating because the probabilistic choice in the loop body always chooses $tl$, and so not possible to terminate.

Though $flip$ and $pflip(p)$ are semantically equal, we expect there are differences between the two programs in other aspects, such as average termination time. Consider a biased coin with probability $p=0.75$ to see heads. Then we know, on average, it needs fewer flips than an unbiased coin to see heads. In other words, $pflip(p)$ has a smaller average termination time than $flip$. This is modelled in Hehner's work~\cite{Hehner2011} by a time variable $t$ of type natural numbers to count iterations in a loop. In our language, it is defined below.

\begin{definition}[Coin flip with time]
  \isalink{https://github.com/RandallYe/probabilistic_programming_utp/blob/6a4419b8674b84988065a58696f15093d176594c/probability/probabilistic_relations/Examples/utp_prob_rel_lattice_coin.thy\#L668}
  \begin{align*}
    & \isakwmaj{alphabet}\ cstate\_t = t::\nat \qquad c::Tcoin \\
    & flip\_t \defs \pwhile{c = tl}{\pseq{\left(\ppchoice{1/2}{\passign{c}{hd}}{\passign{c}{tl}}\right)}{\passign{t}{t + 1}}} 
  \end{align*}
\end{definition}

The new state space is $cstate\_t$ with an additional variable $t$ of $\nat$, and the loop $flip\_t$ will increase $t$ by 1 in each iteration.  After the introduction of $t$, we use Theorem~\ref{thm:rec_unique_fin} to prove the semantics of $flip\_t$ is

\begin{thm}
  \label{thm:flip_t_sem}
  $ flip\_t = \rvprfunsym{\usexpr{
      \begin{array}[]{l}
        \ibracket{c=hd}*\ibracket{c'=hd}*\ibracket{t'=t}+\\
        \ibracket{c=tl}*\ibracket{c'=hd}*\ibracket{t'\geq t+1}*(1/2)^{t'-t}
      \end{array}
    }}$ 
  \isalink{https://github.com/RandallYe/probabilistic_programming_utp/blob/6a4419b8674b84988065a58696f15093d176594c/probability/probabilistic_relations/Examples/utp_prob_rel_lattice_coin.thy\#L1008}
\end{thm}

If the initial value of $c$ is $hd$ (that is, $c=hd$), $flip\_t$ terminates immediately ($t'=t$) and its final value of $c$ is $hd$ ($c'=hd$).  If the initial value of $c$ is $tl$ ($c=tl$), $flip\_t$ terminates ($c'=hd$) only when $t'$ is larger than or equal to $t+1$, that is, at least one flip of the coin. The probability that $flip\_t$ terminates at time $t'$ is given by $(1/2)^{t'-t}$ which can be regarded as $t'-t-1$ tails followed by heads: The termination probability of $flip\_t$ is the sum of $(1/2)^{t'-t}$ over natural numbers starting from 1: $\Sigma_{t'=t+1}^{\infty}(1/2)^{t'-t}=\Sigma_{n=1}^{\infty}(1/2)^n$. It is a geometric series with a common ratio $1/2$, and so its sum is equal to $(1/(1-(1/2))) - 1 = 1$. This is shown in the theorem below.

\begin{thm}
  \label{thm:flipt_sem_termination}
  $ \pseq{\prrvfunsym{flip\_t}}{{\ibracket{c = hd}}} = {\usexpr{1}}$
  \isalink{https://github.com/RandallYe/probabilistic_programming_utp/blob/6a4419b8674b84988065a58696f15093d176594c/probability/probabilistic_relations/Examples/utp_prob_rel_lattice_coin.thy\#L1024}
\end{thm}

With $t$, we can quantify the expected value of $t$ (the number of flips on average to get the loop terminated) by sequential composition.
\begin{thm}
  \label{thm:flipt_sem_expected_t}
  $ \pseq{\prrvfunsym{flip\_t}}{t} = {\usexpr{\ibracket{c=hd}*t+\ibracket{c=tl}*{(t+2)}}}$
  \isalink{https://github.com/RandallYe/probabilistic_programming_utp/blob/6a4419b8674b84988065a58696f15093d176594c/probability/probabilistic_relations/Examples/utp_prob_rel_lattice_coin.thy\#L1081}
\end{thm}
The expectation of $t$ given the distribution by $flip\_t$ is $t$ itself (terminate immediately) if the initial value of $c$ is $hd$, and $t+2$ (2 flips on average) otherwise.

We consider the parametrised version with $t$ where $p$ is a \ureal\ number.
\begin{definition}[Parametrised coin flip with time]
  \isalink{https://github.com/RandallYe/probabilistic_programming_utp/blob/6a4419b8674b84988065a58696f15093d176594c/probability/probabilistic_relations/Examples/utp_prob_rel_lattice_coin.thy\#L1214}
  \begin{align*}
    & pflip\_t(p)\defs \pwhile{c = tl}{\pseq{\left(\ppchoice{p}{\passign{c}{hd}}{\passign{c}{tl}}\right)}{\passign{t}{t + 1}}} 
  \end{align*}
\end{definition}

We show its semantics below.
\begin{thm}
  \label{thm:pflip_t_sem}
  $ p \neq 0 \implies pflip\_t(p) = \rvprfunsym{\usexpr{
      \begin{array}[]{l}
        \ibracket{c=hd}*\ibracket{c'=hd}*\ibracket{t'=t}+\\
        \ibracket{c=tl}*\ibracket{c'=hd}*\ibracket{t'\geq t+1}*(1-\ur{p})^{t'-t-1}*\ur{p}
      \end{array}
    }}$ 
  \isalink{https://github.com/RandallYe/probabilistic_programming_utp/blob/6a4419b8674b84988065a58696f15093d176594c/probability/probabilistic_relations/Examples/utp_prob_rel_lattice_coin.thy\#L1518}
\end{thm}

If $p$ is not 0, then the probability that $pflip\_t(p)$ terminates at $t'$ now is $(t'-t-1)$ tails (probability $(1-\ur{p})^{t'-t-1}$, $\ur{p}$ is the conversion of $p$ to $\real$) and followed by heads (probability $\ur{p}$) when the initial $c$ is $tl$.  The following theorem shows the program $pflip\_t(p)$ terminates almost surely.

\begin{thm}
    \label{thm:flipt_p_sem_termination}
    $ p \neq 0 \implies \pseq{\prrvfunsym{pflip\_t(p)}}{{\ibracket{c = hd}}} = {\usexpr{1}}$
    \isalink{https://github.com/RandallYe/probabilistic_programming_utp/blob/6a4419b8674b84988065a58696f15093d176594c/probability/probabilistic_relations/Examples/utp_prob_rel_lattice_coin.thy\#L1530}
\end{thm}

Its expected termination time is $1/\ur{p}$ flips, shown below.
\begin{thm}
  \label{thm:flipt_p_sem_expected_t}
  $ p \neq 0 \implies \pseq{\prrvfunsym{pflip\_t(p)}}{t} = {\usexpr{\ibracket{c=hd}*t+\ibracket{c=tl}*{(t+1/\ur{p})}}}$
  \isalink{https://github.com/RandallYe/probabilistic_programming_utp/blob/6a4419b8674b84988065a58696f15093d176594c/probability/probabilistic_relations/Examples/utp_prob_rel_lattice_coin.thy\#L1610}
\end{thm}

In essence, the proof of this theorem is the calculation of the following summation.
\begin{align*}
  \left(\infsum v_0 | c_v(v_0) = hd \land Suc(t) \leq t_v(v_0) @ (1-\ur{p})^{(t_v(v_0) - Suc(t))} * \ur{p}*t_v(v_0) \right)
\end{align*}
where $c_v(v_0)$ and $t_v(v_0)$ extract the values of the variables $c$ and $t$ from the state $v_0$. The calculation involves several important steps:
\begin{enumerate}[labelindent=\parindent,leftmargin=*,label=(\arabic*).]
\item find an injective function to reindex the summation over $v_0$ into the summation over $n$ of natural numbers: $\left(\infsum n:\nat @ f(n)\right)$ where $f(n) \defs ((1-\ur{p})^n * \ur{p}* (Suc(t)+n)$;
\item prove $f(n)$ is summable using the ratio test for convergence by supplying a constant ratio $c$ that is less than 1 and a natural number $N$ so that for all numbers larger than $N$, the ratio $f(n+1)/f(n)$ is less than $c$; 
\item because $f(n)$ is summable, we can assume $f(n+1)$ sums to $l$, then $f(n)$ must sum to $l+f(0)=l+\ur{p}*Suc(t)$; 
\item alternatively, $f(n+1)=(1-\ur{p})^{n+1} * \ur{p}* (Suc(t)+n+1)=f(n)*(1-\ur{p})+(1-\ur{p})^{n} * \ur{p} * (1-\ur{p})$, and so $\left(\infsum n:\nat @ f(n+1)\right)=\left(\infsum n:\nat @ f(n)* (1-\ur{p})\right)+ \left(\infsum n:\nat @ (1-\ur{p})^{n} * \ur{p} * (1-\ur{p})\right)$;
  \begin{itemize}
  \item $\left(\infsum n:\nat @ (1-\ur{p})^{n} * \ur{p} * (1-\ur{p})\right)$ is a geometric series and equal to $\ur{p} * (1-\ur{p}) * \left(1/(1-(1-\ur{p}))\right)=1-\ur{p}$
  \end{itemize}
\item get an equation $l = (l+\ur{p}*Suc(t))*(1-\ur{p}) + (1 - \ur{p})$, solve this equation and we get the result of $l$, then we know $f(n)$ sums to $\left(t + 1/\ur{p}\right)$.
\end{enumerate}

We also note that the parameter $p$ is not present in the semantics (see Theorem~\ref{thm:pflip_sem}) of $pflip$ as long as $p$ is larger than 0. We have seen that the average termination time $1/\ur{p}$ is a function of the parameter $p$, which entitles us to reason about parametric probabilistic models intrinsically, not like approximation and limitations in probabilistic model checking~\cite{Daws2005,Hahn2011}.\footnote{\url{https://www.prismmodelchecker.org/manual/RunningPRISM/ParametricModelChecking}}


\subsection{Dice}
\label{sec:ex_cases:dice}
This example~\cite{Hehner2011} is about throwing a pair of dice till they have the same outcome. The $dice$ program is defined below.
\begin{definition}
  \label{def:dice}
  \isalink{https://github.com/RandallYe/probabilistic_programming_utp/blob/6a4419b8674b84988065a58696f15093d176594c/probability/probabilistic_relations/Examples/utp_prob_rel_lattice_dices.thy\#L41}
  \begin{align*} 
    Tdice & ::= \{1..6\}\\ 
    \isakwmaj{alphabet}\ & fdstate = d_1::Tdice \qquad d_2::Tdice \\
    throw & \defs {\pseq{\rvprfunsym{\uniformdist{d_1}{Tdice}}}{\rvprfunsym{\uniformdist{d_2}{Tdice}}}} \\
    dice & \defs \pwhile{d_1 \neq d_2}{throw}
  \end{align*}
\end{definition}
The outcome of a die is from 1 to 6 as given in $Tdice$. The observation space of this program is $fdstate$, containing two variables $d_1$ and $d_2$ of type $Tdice$, denoting the outcome of each dice in an experiment. The program $throw$ is the sequential composition of two uniform distributions to choose $d_1$ and $d_2$ independently, and $dice$ models the example: continue throwing till the outcomes of two dice are equal ($d_1=d_2$).

We use the unique fixed point theorem~\ref{thm:rec_unique} to give semantics. First, $throw$ is a distribution. 
\begin{thm}
  $\isfinaldist(\prrvfunsym{throw})$ 
  \isalink{https://github.com/RandallYe/probabilistic_programming_utp/blob/6a4419b8674b84988065a58696f15093d176594c/probability/probabilistic_relations/Examples/utp_prob_rel_lattice_dices.thy\#L409}
\end{thm}
Second, the observation space $Tdice$ is finite, and sct $fdstate \cross fdstate$ is also finite.
Third, the differences of iterations from top and bottom tend to 0.
\begin{thm}
  \label{thm:dice_iterdiff_0}
  $\forall s:fdstate\cross fdstate @ \left(\lambda n @ \prrvfunsym{\iterdiff\left(n, d_1 \neq d_2, throw\right)}(s)\right) \tendsto 0$ 
  \isalink{https://github.com/RandallYe/probabilistic_programming_utp/blob/6a4419b8674b84988065a58696f15093d176594c/probability/probabilistic_relations/Examples/utp_prob_rel_lattice_dices.thy\#L1221}
\end{thm}
Finally, we define $H$, 
\isalink{https://github.com/RandallYe/probabilistic_programming_utp/blob/6a4419b8674b84988065a58696f15093d176594c/probability/probabilistic_relations/Examples/utp_prob_rel_lattice_dices.thy\#L389}
\begin{align*}
  H & \defs \usexpr{\ibracket{d_1 = d_2} * \ibracket{d_1' = d_1 \land d_2' = d_2} + \ibracket{d_1 \neq d_2} * \ibracket{d_1' = d_2'} / 6}
\end{align*}
and prove it is a fixed point.
\begin{thm}
  \label{thm:dice_H}
  $\lfun^{d_1\neq d_2}_{throw}\left(\rvprfunsym{H}\right) = \rvprfunsym{H}$
  \isalink{https://github.com/RandallYe/probabilistic_programming_utp/blob/6a4419b8674b84988065a58696f15093d176594c/probability/probabilistic_relations/Examples/utp_prob_rel_lattice_dices.thy\#L1239}
\end{thm}

The $H$ gives the distribution on the final states. If initially, $d_1$ is equal to $d_2$; it has probability 1 to establish that both $d_1'$ and $d_2'$ are equal to their initial states, so they are identical too. This is the semantics of $\pskip$. However, if $d_1$ is not equal to $d_2$ initially, it has a probability $1/6$ to establish $d_1'=d_2'$. Because there are six combinations of the equal values of $d_1'$ and $d_2'$, the total probability is still 1 ($6*1/6$), so $H$ is a distribution.

All four assumptions of Theorem~\ref{thm:rec_unique} are now established. The $dice$, therefore, is semantically just the fixed point $\rvprfunsym{H}$.

\begin{thm}
  \label{thm:dice_sem}
  $dice = \rvprfunsym{H}$
  \isalink{https://github.com/RandallYe/probabilistic_programming_utp/blob/6a4419b8674b84988065a58696f15093d176594c/probability/probabilistic_relations/Examples/utp_prob_rel_lattice_dices.thy\#L1361}
\end{thm}

The $dice$ terminates almost surely and is almost impossible for non-termination.
\begin{thm}
  \isalink{https://github.com/RandallYe/probabilistic_programming_utp/blob/6a4419b8674b84988065a58696f15093d176594c/probability/probabilistic_relations/Examples/utp_prob_rel_lattice_dices.thy\#L1379}
  \begin{align*}
    & \pseq{\prrvfunsym{dice}}{{\ibracket{d_1 = d_2}}} = {\usexpr{1}} \tag*{(termination probability)} \label{thm:dice_term}\\
    & \pseq{\prrvfunsym{dice}}{{\ibracket{d_1 \neq d_2}}} = {\usexpr{0}} \tag*{(non-termination probability)} \label{thm:dice_nonterm}
  \end{align*}
\end{thm}

We now consider the dice program with a time variable $t$.
\begin{definition}[Dice with time]
  \isalink{https://github.com/RandallYe/probabilistic_programming_utp/blob/6a4419b8674b84988065a58696f15093d176594c/probability/probabilistic_relations/Examples/utp_prob_rel_lattice_dices.thy\#L1430}
  \begin{align*}
    & \isakwmaj{alphabet}\ dstate\_t = t::\nat \qquad d_1::Tdice \qquad d_2::Tdice \\
    & throw\_t \defs {\pseq{\pseq{\rvprfunsym{\uniformdist{d_1}{Tdice}}}{\rvprfunsym{\uniformdist{d_2}{Tdice}}}}{\passign{t}{t + 1}}} \\
    & dice\_t \defs \pwhile{d_1 \neq d_2}{throw\_t}
  \end{align*}
\end{definition}

{

We show that 
\begin{thm}
  \label{thm:dice_pt}
  $throw\_t = \rvprfunsym{\ibracket{t'=t+1} / 36}$
  \isalink{https://github.com/RandallYe/probabilistic_programming_utp/blob/6a4419b8674b84988065a58696f15093d176594c/probability/probabilistic_relations/Examples/utp_prob_rel_lattice_dices.thy\#L1883}
  \isalink{https://github.com/RandallYe/probabilistic_programming_utp/blob/6a4419b8674b84988065a58696f15093d176594c/probability/probabilistic_relations/Examples/utp_prob_rel_lattice_dices.thy\#L1903}
\end{thm}

The ${\ibracket{t'=t+1} / 36}$ is a distribution.
\begin{thm}
  \label{thm:dice_pt_distr}
  $\isfinaldist(\ibracket{t'=t+1} / 36)$
  \isalink{https://github.com/RandallYe/probabilistic_programming_utp/blob/6a4419b8674b84988065a58696f15093d176594c/probability/probabilistic_relations/Examples/utp_prob_rel_lattice_dices.thy\#L1941}
\end{thm}

So the conversion of $throw\_t$ to real-valued functions is just the distribution.
\begin{thm}
  \label{thm:dice_pt_inverse}
  ${\prrvfunsym{throw\_t}} = {\ibracket{t'=t+1} / 36}$
  \isalink{https://github.com/RandallYe/probabilistic_programming_utp/blob/6a4419b8674b84988065a58696f15093d176594c/probability/probabilistic_relations/Examples/utp_prob_rel_lattice_dices.thy\#L1952}
\end{thm}
\begin{proof}
   This can be proved using Theorems~\ref{thm:prrvfun_inverse}, \ref{thm:dice_pt_distr}, and \ref{thm:final_distribtion}.
\end{proof}
}

We define $Ht$.
\begin{definition}[$Ht$]
    \label{def:Ht}
\begin{align*}
  Ht & \defs \usexpr{
       \begin{array}[]{l}
         \ibracket{d_1 = d_2} * \ibracket{t'=t \land d_1' = d_1 \land d_2' = d_2} + \\
         \ibracket{d_1 \neq d_2} * \ibracket{d_1' = d_2'} * \ibracket{t'\geq t+1}*(5/6)^{t'-t-1}*(1/36)
       \end{array}}
  \tag*{\isalink{https://github.com/RandallYe/probabilistic_programming_utp/blob/6a4419b8674b84988065a58696f15093d176594c/probability/probabilistic_relations/Examples/utp_prob_rel_lattice_dices.thy\#L1686}}
\end{align*}
\end{definition}

{
    The $Ht$ is a distribution. 
\begin{thm}
  \label{thm:dice_ht_distr}
  $\isfinaldist(Ht)$
  \isalink{https://github.com/RandallYe/probabilistic_programming_utp/blob/6a4419b8674b84988065a58696f15093d176594c/probability/probabilistic_relations/Examples/utp_prob_rel_lattice_dices.thy\#L2018}
\end{thm}

\begin{thm}
  \label{thm:dice_ht_inverse}
  ${\prrvfunsym{\left(\rvprfunsym{Ht}\right)}} = Ht$
\end{thm}
\begin{proof}
   This can be proved using Theorems~\ref{thm:prrvfun_inverse}, \ref{thm:dice_ht_distr}, and \ref{thm:final_distribtion}.
\end{proof}
}

The $Ht$ is proved to be a fixed point of $dice\_t$.\footnote{We note that our $Ht$ is different from that of~\cite{Hehner2011} where the probability of having $d_1'=d_2'$ is $1/6$ (instead of $1/36$ in ours). After a careful comparison of our mechanised proof and the pencil-and-paper proof in~\cite{Hehner2011}, we figured out the mistake is introduced in a step of the proofs in~\cite{Hehner2011}.}
\begin{thm}
  \label{thm:dice_Ht}
  $\lfun^{d_1\neq d_2}_{throw\_t}\left(\rvprfunsym{Ht}\right) = \rvprfunsym{Ht}$
  \isalink{https://github.com/RandallYe/probabilistic_programming_utp/blob/6a4419b8674b84988065a58696f15093d176594c/probability/probabilistic_relations/Examples/utp_prob_rel_lattice_dices.thy\#L2053}
\end{thm}
{
\begin{proof}
\begin{align*}
    & \lfun^{d_1\neq d_2}_{throw\_t}\left(\rvprfunsym{Ht}\right) \\
= & \cmt{Law~\ref{thm:lfun_altdef}} \\
  & \rvprfunsym{{\ibracket{d_1\neq d_2}}*\prrvfunsym{\left(\pseq{throw\_t}{\rvprfunsym{Ht}}\right)} + \ibracket{\lnot d_1\neq d_2} * {\ibracket{\II}}}\\
  = & \cmt{Definition~\ref{def:prob_programs} for $\pseq{}{}$} \\
  & \rvprfunsym{{\ibracket{d_1\neq d_2}}*\prrvfunsym{\left(\rvprfunsym{\fseq{\prrvfunsym{throw\_t}}{\prrvfunsym{\left(\rvprfunsym{Ht}\right)}}}\right)} + \ibracket{\lnot d_1\neq d_2} * {\ibracket{\II}}}\\
  = & \cmt{Theorems~\ref{thm:dice_pt_inverse} and \ref{thm:dice_ht_inverse}} \\
  & \rvprfunsym{{\ibracket{d_1\neq d_2}}*\prrvfunsym{\left(\rvprfunsym{\fseq{{\ibracket{t'=t+1} / 36}}{Ht}}\right)} + \ibracket{\lnot d_1\neq d_2} * {\ibracket{\II}}}\\
  = & \cmt{Definition~\ref{def:prob_programs} for $\fseq{}{}$} \\
  & \rvprfunsym{{\ibracket{d_1\neq d_2}}*\prrvfunsym{\left(\rvprfunsym{{{\infsum \vv'' @ {\left(\ibracket{t'=t+1} / 36\right)}[\vv''/\vv'] * {Ht}[\vv''/\vv]}}}\right)} + \ibracket{\lnot d_1\neq d_2} * {\ibracket{\II}}}\\
  = & \cmt{Definitions~\ref{def:Ht}, substitution, and omit $\ibracket{\lnot d_1\neq d_2} * {\ibracket{\II}}$} \\
  & \rvprfunsym{{\ibracket{d_1\neq d_2}}*\prrvfunsym{\left(\rvprfunsym{{{\infsum \vv'' @ 
      \begin{array}[]{l}
      {\left(\ibracket{t''=t+1} / 36\right)} * \\
      {\left(
       \begin{array}[]{l}
         \ibracket{d_1'' = d_2''} * \ibracket{t'=t'' \land d_1' = d_1'' \land d_2' = d_2''} + \\
         \ibracket{d_1'' \neq d_2''} * \ibracket{d_1' = d_2'} * \ibracket{t'\geq t''+1}*(5/6)^{t'-t''-1}*(1/36)
       \end{array}\right)}
      \end{array}
  }}}\right)} + \cdots} \\
  = & \cmt{Multiplication distributive over addition} \\
  & \rvprfunsym{{\ibracket{d_1\neq d_2}}*\prrvfunsym{\left(\rvprfunsym{{{\infsum \vv'' @ 
      \begin{array}[]{l}
      {\left(\ibracket{t''=t+1} / 36\right)} * \ibracket{d_1'' = d_2''} * \ibracket{t'=t'' \land d_1' = d_1'' \land d_2' = d_2''} + \\
      {\left(\ibracket{t''=t+1} / 36\right)} * \\{
       \begin{array}[]{l}
         \ibracket{d_1'' \neq d_2''} * \ibracket{d_1' = d_2'} * \ibracket{t'\geq t''+1}*(5/6)^{t'-t''-1}*(1/36)
       \end{array}}
      \end{array}
  }}}\right)} + \cdots} \\
  = & \cmt{Law~\ref{thm:summation_add} and proofs of summable omitted } \\
  & \rvprfunsym{{\ibracket{d_1\neq d_2}}*\prrvfunsym{\left(\rvprfunsym{{{
      \begin{array}[]{l}
      \left(\infsum \vv'' @ {\left(\ibracket{t''=t+1} / 36\right)} * \ibracket{d_1'' = d_2''} * \ibracket{t'=t'' \land d_1' = d_1'' \land d_2' = d_2''}\right) + \\
      \infsum \vv'' @ {\left(\ibracket{t''=t+1} / 36\right)} * \\{
       \begin{array}[]{l}
         \ibracket{d_1'' \neq d_2''} * \ibracket{d_1' = d_2'} * \ibracket{t'\geq t''+1}*(5/6)^{t'-t''-1}*(1/36)
       \end{array}}
      \end{array}
  }}}\right)} + \cdots} \\
  = & \cmt{In the first summation, only one state $\vv''[t''=t+1,d_1''=d_1',d_2''=d_2']$ satisfies the predicates }\\
  & \rvprfunsym{{\ibracket{d_1\neq d_2}}*\prrvfunsym{\left(\rvprfunsym{{{
      \begin{array}[]{l}
        {\ibracket{d_1' = d_2'} * \ibracket{t'=t+1} / 36} + \\
      \infsum \vv'' @ {\left(\ibracket{t''=t+1} / 36\right)} * \\{
       \begin{array}[]{l}
         \ibracket{d_1'' \neq d_2''} * \ibracket{d_1' = d_2'} * \ibracket{t'\geq t''+1}*(5/6)^{t'-t''-1}*(1/36)
       \end{array}}
      \end{array}
  }}}\right)} + \cdots} \\
  = & \cmt{There are 30 states $\vv''[t+1/t'',x/d_1'',y/d_2'']$ where $x \neq y$ satisfies the predicates }\\
  & \rvprfunsym{{\ibracket{d_1\neq d_2}}*\prrvfunsym{\left(\rvprfunsym{{{
      \begin{array}[]{l}
        {\ibracket{d_1' = d_2'} * \ibracket{t'=t+1} / 36} + \\
        {
         \ibracket{d_1' = d_2'} * \ibracket{t'\geq t+1+1}*(5/6)^{t'-(t+1)-1}*30*(1/36)*(1/36)
        }
      \end{array}
  }}}\right)} + \cdots} \\
  = & \cmt{$30/36=5/6$}\\
  & \rvprfunsym{{\ibracket{d_1\neq d_2}}*\prrvfunsym{\left(\rvprfunsym{{{
      \begin{array}[]{l}
        {\ibracket{d_1' = d_2'} * \ibracket{t'=t+1} / 36} + \\
        {
         \ibracket{d_1' = d_2'} * \ibracket{t'\geq t+2}*(5/6)^{t'-t-1}*(1/36)
        }
      \end{array}
  }}}\right)} + \cdots} \\
  = & \cmt{ Merged }\\
  & \rvprfunsym{{\ibracket{d_1\neq d_2}}*\prrvfunsym{\left(\rvprfunsym{{{
      \begin{array}[]{l}
        {
         \ibracket{d_1' = d_2'} * \ibracket{t'\geq t+1}*(5/6)^{t'-t-1}*(1/36)
        }
      \end{array}
  }}}\right)} + \cdots} \\
  = & \cmt{Theorems~\ref{thm:prrvfun_inverse}, \ref{thm:dice_ht_distr}, and \ref{thm:final_distribtion}, and the omitted}\\
  & \rvprfunsym{{\ibracket{d_1\neq d_2}}*{\left({{{
      \begin{array}[]{l}
        {
         \ibracket{d_1' = d_2'} * \ibracket{t'\geq t+1}*(5/6)^{t'-t-1}*(1/36)
        }
      \end{array}
  }}}\right)} + \ibracket{\lnot d_1\neq d_2} * {\ibracket{\II}}} \\
  = & \cmt{ Definition~\ref{def:uskip} }\\
  & \rvprfunsym{Ht}
\end{align*}
\end{proof}
}


Using Theorem~\ref{thm:rec_unique_fin}, we prove the semantics of $dice\_t$ is just $Ht$.
\begin{thm}
  \label{thm:dice_t_sem}
  $dice\_t = \rvprfunsym{Ht}$
  \isalink{https://github.com/RandallYe/probabilistic_programming_utp/blob/6a4419b8674b84988065a58696f15093d176594c/probability/probabilistic_relations/Examples/utp_prob_rel_lattice_dices.thy\#L2380}
\end{thm}

\changed[\C{1}]{We note that the semantics of $pflip\_t$ in Theorem~\ref{thm:pflip_t_sem} has a pattern 
\begin{align*}
    \ibracket{c=tl}*\ibracket{c'=hd}*\ibracket{t'\geq t+1}*(5/6)^{t'-t-1}*(1/6) 
\end{align*}
if $p=1/6$, and the semantics of $dice\_t$ here has a pattern 
\begin{align*}
    \ibracket{d_1 \neq d_2} * \ibracket{d_1' = d_2'} * \ibracket{t'\geq t+1}*(5/6)^{t'-t-1}*(1/36) 
\end{align*}
The $(1/6)$ or $(1/36)$ above denotes the success probability of each experiment in terms of a particular valuation of the variables in the observation space. For example, $(1/6)$ denotes the probability of $\ibracket{c'=hd}$ for a particular $t'$ and $c'$ (where $c'=hd$ is the only value to establish $\ibracket{c'=hd}$), and $(1/36)$ denotes the probability of $\ibracket{d_1' = d_2'}$ for a particular $t'$, $d_1'$, and $d_2'$ (where for each $t'$, there are overall 6 values of $d_1'$ and $d_2'$ to establish $\ibracket{d_1' = d_2'}$, that is, both take the same value from 1 to 6).
} 

The $dice\_t$ terminates almost surely.
\begin{thm}
  \label{thm:dice_t_terminate}
  $ \pseq{\prrvfunsym{dice\_t}}{{\ibracket{d_1 = d_2}}} = {\usexpr{1}}$
  \isalink{https://github.com/RandallYe/probabilistic_programming_utp/blob/6a4419b8674b84988065a58696f15093d176594c/probability/probabilistic_relations/Examples/utp_prob_rel_lattice_dices.thy\#L2395}
\end{thm}

On average, the $dice\_t$ takes six dice throws to get an equal outcome.
\begin{thm}
  \label{thm:dice_t_expectation}
  $ \pseq{\prrvfunsym{dice\_t}}{t} = \usexpr{\ibracket{d1=d2}*t + \ibracket{d1 \neq d2}*(t+6)}$
  \isalink{https://github.com/RandallYe/probabilistic_programming_utp/blob/6a4419b8674b84988065a58696f15093d176594c/probability/probabilistic_relations/Examples/utp_prob_rel_lattice_dices.thy\#L2431}
\end{thm}

\section{Conclusion}
\label{sec:concl}

Previous work~\cite{Ye2022,Ye2021} has shown the modelling of aleatoric uncertainty in RoboChart based on the semantics of MDP and in \emph{pGCL} based on the theory of probabilistic designs and the automated verification of probabilistic behaviours using probabilistic model checking and theorem proving.
This work presents a new probabilistic semantic framework, \emph{ProbURel}, and probabilistic programming to cover modelling both aleatoric and epistemic uncertainties and the automated verification of probabilistic systems exhibiting both uncertainties using theorem proving. 
We discuss our probabilistic vision in Sect.~\ref{sec:intro}, and the new semantic framework is our first step in the big picture. With ProbURel, we can give semantics to deterministic, probabilistic sequential programs with the support of discrete distributions and time.

We have based our work on Hehner's predicative probabilistic programming and addressed obstacles to applying his work by formalising and mechanising its semantics in Isabelle/UTP. We have introduced an Iverson bracket notation to separate arithmetic semantics from relational semantics so that reasoning about a probabilistic program can reuse existing reasoning techniques for both arithmetic and relational semantics. 
We have used the UTP's alphabetised relational calculus to formalise its relational semantics, and so probabilistic programs benefit from automated reasoning in Isabelle/UTP. 
We have used the summations over the topological space of real numbers for arithmetic semantics, and so probabilistic programs also benefit from mechanised theories in Isabelle/HOL for reasoning. We have enriched the semantics domains from probabilistic distributions to subdistributions and superdistributions to use the constructive Kleene fixed point theorem to give semantics to probabilistic loops based on the least fixed point and derive a unique fixed point theorem to vastly simplify the reasoning of probabilistic loops. With formalisation and mechanisation, we have reasoned about six examples of probabilistic programs. 

{
\subsection{{Our probabilistic vision}}
\label{sec:concl:vision}
Recently, we presented a probabilistic extension~\cite{Ye2022} to RoboChart~\cite{Miyazawa2019}, a state machine-based DSL for robotics, to allow the modelling of probabilistic behaviour in robot control software. 
RoboChart~\cite{Miyazawa2019,Ye2022} is a core notation in the RoboStar%
\footnote{%
  \url{robostar.cs.york.ac.uk}.%
} framework~\cite{Cavalcanti2021} that brings modern modelling and verification technologies into software engineering for robotics. In this framework, three key elements are models, formal mathematical semantics for models, and automated modelling and verification tool support.  
RoboChart is a UML-like architectural and state machine modelling notation featuring discrete time and probabilistic modelling. It has formal semantics: state machines and architectural semantics~\cite{Miyazawa2019} based on the CSP process algebra~\cite{Hoare1985,Roscoe2011} and time semantics~\cite{Miyazawa2019} based on \emph{tock}-CSP~\cite{Baxter2021,Roscoe2011}. CSP is a formal notation to describe concurrent systems where processes interact using communication. In the framework, robot hardware and control software are captured in robotic platforms, and controllers of RoboChart. The environment is captured in RoboWorld~\cite{Cavalcanti2021a} whose semantics is based on \textit{\textsf{CyPhyCircus}}~\cite{Foster2020a}, a hybrid process algebra, because of the continuous nature of the environment.

While RoboChart and RoboWorld are high-level specification languages, RoboSim~\cite{Cavalcanti2019} is a cycle-based simulation-level notation in the RoboStar framework. RoboSim also has semantics in CSP and \emph{tock}-CSP. A RoboSim model can be automatically transformed from a RoboChart model directly, and its correctness is established by refinements~\cite{Roscoe2011} in CSP.

Probability is used to capture uncertainties from physical robots and the environment and randomisation in controllers. The probabilistic extension in the RoboStar framework requires its semantic extension to either base on process algebras and hybrid process algebras or have the richness to deal with probabilistic concurrent and reactive systems. The features of its probabilistic semantics that we consider in this big vision include discrete and continuous distributions, discrete time, nondeterminism, concurrency, and refinement. The type system and the comprehensive expression language of RoboChart~\cite{Miyazawa2019}, additionally, are based on those of the Z notation~\cite{Spivey1992,Woodcock1996} and include mathematical data types such as relations and functions, quantifications, and lambda expressions. Because of such richness in semantics and language features of RoboChart, the formal verification support of RoboChart requires theorem proving and model checking.

Our immediate thought is to consider existing probabilistic extensions to process algebras, including CSP-based~\cite{Morgan1996,Morgan2005,Nunez1995,Gomez1997,Kwiatkowska1998,Georgievska2012}, CCS-based~\cite{Hansson1990,Giacalone1990,Yi1992,Vanglabbeek1995}, and ACP-based~\cite{Andova2002}. The main difference between these extensions is how existing constructs or operators, particularly nondeterministic and external choice, interact with probabilistic choice. We also looked at extensions based on probabilistic transition systems~\cite{Larsen1991,Bloom1989,Jonsson2001} and automata~\cite{Wu1997,Hartmanns2015}. To preserve the distributivity of existing operators over probabilistic choice, some algebraic properties are lost, such as the congruence for hiding and asynchronous parallel composition~\cite{Kwiatkowska1998}, idempotence for nondeterministic choice~\cite{Morgan1996a}, or even no standard nondeterministic choice~\cite{Gomez1997,Seidel1995}. The critical problem, however, is the lack of tool support for these extensions. For example, FDR~\cite{T.GibsonRobinson2014}, a refinement model checker for CSP and tock-CSP, cannot verify the probabilistic extensions in CSP. For this reason, we explored other solutions. 

In~\cite{Woodcock2019,Ye2022}, we give probabilistic semantics of RoboChart on probabilistic designs~\cite{Ye2021} in Hoare and He's Unifying Theories of Programming (UTP)~\cite{Hoare1998} and then use the theorem prover Isabelle/UTP~\cite{Foster2020}, an implementation of UTP in Isabelle/HOL, to verify probabilistic models. Probabilistic designs are an embedding of standard non-probabilistic designs into the probabilistic world. The theory of probabilistic designs gives probabilistic semantics to the imperative nondeterministic probabilistic sequential programming language \emph{pGCL}~\cite{McIver2005}, but not reactive aspects of RoboChart. We have thought about lifting probabilistic designs into probabilistic reactive designs. Still, the main obstacle is the complexity of reasoning about probabilistic distributions in probabilistic designs because distributions are captured in a dedicated variable $prob$, representing a probability mass function. In particular, the definition~\cite{Ye2021} of sequential composition includes an existential quantification over intermediate distributions. The proof of sequential composition needs to supply a witness for the intermediate distributions, which is helpful but non-trivial.

We also gave RoboChart probabilistic semantics~\cite{Ye2022,Miyazawa2020} in the PRISM language~\cite{Kwiatkowska2011}. We developed plugins for RoboTool,\footnote{\url{www.cs.york.ac.uk/robostar/robotool/}} an accompanying tool for RoboChart, to support automated verification through probabilistic model checking using PRISM. PRISM, however, employs a closed-world assumption: systems are not subjected to environmental inputs. To verify a RoboChart model, such as a high voltage controller\footnote{\url{github.com/UoY-RoboStar/hvc-case-study/tree/prism_verification/sbmf}} for a painting robot~\cite{Murray2020} and an agricultural robot\footnote{\url{github.com/UoY-RoboStar/uvc-case-study}} for UV-light treatment using PRISM, we need to constrain the environmental input and verify its expected outputs through an additional PRISM module being in parallel with the corresponding PRISM model that is automatically transformed from the RoboChart model. Finally, the safety and reachability properties of the RoboChart model (checked by the trace refinement in FDR) become deadlock freedom problems in PRISM. However, this cannot verify other properties like liveness, which requires failures-divergences refinement in CSP and FDR.

The research question that we aim to answer is a probabilistic semantic framework that 
\begin{enumerate*}[label={(\arabic*)}]
    \item has rich semantics to capture our probabilistic vision, 
    \item is simple and flexible to allow further extensions, and
    \item supports theorem proving.
\end{enumerate*}
This question is comprehensive and needs a research programme, instead of a project, to address it. The work we present in this paper is our first step to answering this question. 
}

\subsection{Future work}
{We have not proved and mechanised the SRW example~\ref{ex:srw}. Our immediate future work is to verify SRW: its semantics, termination, and expected runtime in terms of the parameters $m$ and $p$. We are also interested in the mathematical (that is, the probability theory) way to calculate the termination distribution and comparing it with our programming way (that is, lfp) to establish the equivalence between them.}

Our fixed point theorems, such as Theorems~\ref{thm:rec_least_fixed_point}, \ref{thm:rec_great_fixed_point}, \ref{thm:rec_unique}, and \ref{thm:rec_unique_fin} for probabilistic loops, cannot deal with the programs (the loop body) whose final observation space contains infinite states with positive probabilities. 
The restriction is introduced in Theorems~\ref{thm:incseq_limit_is_lub_all} and \ref{thm:decseq_limit_is_glb_all}, which are used to prove continuity theorems~\ref{thm:continuity_lfun_bot} and \ref{thm:continuity_lfun_top}, and eventually for the least and greatest fixed point theorems \ref{thm:rec_least_fixed_point} and \ref{thm:rec_great_fixed_point}. Our immediate future work is to extend our fixed point theorem to support such countably infinite state space, enabling us to give semantics to loops containing such programs. Our approach is to use Cousot's constructive version of the Knaster–Tarski fixed point theorem~\cite{Cousot1979} to weaken continuity to monotonicity and treat the least fixed point as the stationary limit of transfinite iteration sequences. 
With this extension, our semantics can tackle more general probabilistic programs with countably infinite state space. 
{Hehner~\cite{Hehner2011} presented a simpler semantics for loops. His approach is to include a time variable of type extended integer or real numbers to count iterations, similar to the $t$ (but its type is natural numbers) in our examples. He argued that if a fixed point is proved for a loop, then it is the only fixed point, and so the semantics for the loop. This is very interesting to us. We could formalise his proof and mechanise it in Isabelle/UTP, which may benefit our approach to simplify reasoning of loops or address the limitation of finite states with positive probabilities.} 

The probabilistic programming we present in this paper only considers discrete probabilistic distributions. One of our future works is to support continuous distributions such as normal or Gaussian distributions, uniform distributions, and exponential distributions, which are naturally presented in many physical systems in our semantics. Each point has zero probability in (absolute) continuous distributions, so the probability mass functions for discrete distributions could not describe them. Instead, they are characterised by probability density functions, which require measure theory to deal with probabilities and integration~\cite{Dahlqvist2020} over intervals. We, therefore, will introduce measure theory to our semantics and mechanise it in Isabelle/UTP based on the measure theory in Isabelle. After these lines of future work are complete, our probabilistic programming can automatically model a wide range of probabilistic systems and reasoning about them. 

With UTP and ProbURel, we could bring different approaches to handling uncertainty, such as epistemic mu-calculus and probabilistic synthesis, together and unify these approaches. Our semantics for ProbURel are denotational, which could underpin the operational semantics for other approaches. By unifying these theories, we could link different tools. For example, one model could be analysed using our theorem prover, and it could also be transformed into another probabilistic programming language and analysed by the supported tools for it, such as PRISM. This will be beneficial for analysis by leveraging the advantages of different tools.



\section*{Acknowledgements}

This work is funded by the EPSRC projects RoboCalc (Grant EP/M025756/1) and RoboTest (Grant EP/R025479/1).

\appendix
{
\section{Proofs}
\subsection{Proof of Theorem~\ref{thm:iterdiff}}
\label{appendix:proof_iterdiff}

We present and prove three theorems first and then use them to prove Theorem~\ref{thm:iterdiff}.
\begin{thm}
    \label{thm:iterdiff_bot}
    Provided $P$ is a distribution, that is, $\isfinaldist(P)$. 
    \begin{align*}
        &\left(\lfun_P^b\right)^0(\ufzero) = \ufzero \\ 
        &\left(\lfun_P^b\right)^1(\ufzero) = \lambda (s, s'). ~\rvprfunsym{\ibracket{\lnot b(s)}*\ibracket{s'=s}}
    \end{align*}
    If $n>1$, then
    \begin{align*}
        & \left(\lfun_P^b\right)^n(\ufzero) = \\ 
        & \lambda (s, s'). \rvprfunsym{
            \left(
            \begin{array}[]{@{}l} 
                \sum\limits_{i=1}^{n-1} \left(
                \begin{array}{@{}l}
                    \infsum s_{i-1}.  
                    \ibracket{b(s)}*{\prrvfunsym{P}}(s,s_{i-1})* \\
                    \left(
                    \begin{array}{@{}l}
                        \infsum s_{i-2}. \ibracket{b(s_{i-1})}*{\prrvfunsym{P}}(s_{i-1}, s_{i-2})* \\\left(
                        \begin{array}{@{}l}
                            \vdots * \\
                            \left(
                            \infsum s_0. \ibracket{b(s_1)}* {\prrvfunsym{P}}(s_1, s_0) * \ibracket{\lnot b(s_0)}*\ibracket{s'=s_0} 
                            \right)
                        \end{array}\right)
                    \end{array}
                    \right)
                \end{array}\right) \\
                + \left({\ibracket{\lnot b(s)}}\right) * {\ibracket{s'=s}}
            \end{array} 
            \right) 
        } 
    \end{align*}
\end{thm}

\begin{proof}
    We show below that $\left(\lfun_P^b\right)^n(\ufzero)$ for $n=0$ to $3$ satisfies the theorem.
    \begin{align*}
        &\left(\lfun_P^b\right)^0(\ufzero) = \lambda (s, s'). 0 = \ufzero \\
        &\left(\lfun_P^b\right)^1(\ufzero) \\
        = & \cmt{Defintion~\ref{def:lfun}} \\
        & \pcchoice{b}{\left(\pseq{P}{\ufzero}\right)}{\pskip}\\ 
        = & \cmt{Law~\ref{thm:lfun_altdef}} \\
        & \rvprfunsym{{\ibracket{b}}*\prrvfunsym{\left(\pseq{P}{\ufzero}\right)} + \ibracket{\lnot b} * {\ibracket{\II}}} \\
        = & \cmt{Theorem~\ref{thm:prog_seq_comp} Law~\ref{thm:pseq_right_zero}} \\
        & \rvprfunsym{\ibracket{\lnot b} * {\ibracket{\II}}} \\
        = &\cmt{Expand as a function form, use $s$ and $s'$ for initial and final observation states} \\
        &\cmt{Substitution and Definition~\ref{def:uskip}} \\
         &\lambda (s, s'). ~\rvprfunsym{\ibracket{\lnot b(s)}*\ibracket{s'=s}}\\
        &\left(\lfun_P^b\right)^2 (\ufzero) \\
        = &\cmt{$\lfun^2(\ufzero) = \lfun(\lfun^1(\ufzero))$ and Law~\ref{thm:lfun_altdef}} \\
        & \pcchoice{b}{\left(\pseq{P}{\rvprfunsym{\ibracket{\lnot b} * {\ibracket{\II}}}}\right)}{\pskip}\\ 
        = &\cmt{Theorem~\ref{thm:prog_cond_choice} Law~\ref{thm:cchoice_pchoice} and Theorem~\ref{thm:prob_prob_choice} Law~\ref{thm:pchoice_altdef} } \\
        & \rvprfunsym{\prrvfunsym{\left({\rvprfunsym{\ibracket{b}}}\right)}*\prrvfunsym{\left(\pseq{P}{\rvprfunsym{\ibracket{\lnot b} * {\ibracket{\II}}}}\right)} + \left(\rfone - \prrvfunsym{\left({\rvprfunsym{\ibracket{b}}}\right)}\right) * \prrvfunsym{\pskip}}  \\
        = &\cmt{Theorem~\ref{thm:prrvfun_inverse_ibracket}, Definition~\ref{def:prog_skip}} \\
        & \rvprfunsym{{\ibracket{b}}*\prrvfunsym{\left(\pseq{P}{\rvprfunsym{\ibracket{\lnot b} * {\ibracket{\II}}}}\right)} + \left(\rfone - {\ibracket{b}}\right) * {\ibracket{\II}}} \\
        = &\cmt{Definition~\ref{def:prog_seq}, Theorem~\ref{thm:ib} Law~\ref{thm:ib_neg}} \\
        & \rvprfunsym{{\ibracket{b}}*\prrvfunsym{\left(\rvprfunsym{\left({\infsum v_0 @ {{\prrvfunsym{P}}}[v_0/\vv'] * {{\prrvfunsym{\left(\rvprfunsym{\ibracket{\lnot b} * {\ibracket{\II}}}\right)}}}[v_0/\vv]}\right)}\right)} + \left({\ibracket{\lnot b}}\right) * {\ibracket{\II}}} \\
        = &\cmt{Theorem~\ref{thm:prrvfun_inverse_ibracket}} \\
        & \rvprfunsym{{\ibracket{b}}*\prrvfunsym{\left(\rvprfunsym{\left({\infsum v_0 @ {{\prrvfunsym{P}}}[v_0/\vv'] * {{{\left({\ibracket{\lnot b} * {\ibracket{\II}}}\right)}}}[v_0/\vv]}\right)}\right)} + \left({\ibracket{\lnot b}}\right) * {\ibracket{\II}}} \\
        = &\cmt{Expand as a function form, use $s$ and $s'$ for initial and final observation states} \\
        &\cmt{Substitution and Definition~\ref{def:uskip}} \\
        & \lambda (s, s'). \rvprfunsym{{\ibracket{b(s)}}*\prrvfunsym{\left(\rvprfunsym{\left({\infsum s_0 @ {{\prrvfunsym{P}}}(s,s_0) * {{{\left({\ibracket{\lnot b(s_0)} * {\ibracket{s'=s_0}}}\right)}}}}\right)}\right)} + \left({\ibracket{\lnot b(s)}}\right) * {\ibracket{s'=s}}} \\
        = &\cmt{Theorem~\ref{thm:prrvfun_inverse} where the proof of $\isprob$ is omitted} \\
        & \lambda (s, s'). \rvprfunsym{{\ibracket{b(s)}}*{{\left({\infsum s_0 @ {{\prrvfunsym{P}}}(s,s_0) * {{{\left({\ibracket{\lnot b(s_0)} * {\ibracket{s'=s_0}}}\right)}}}}\right)}} + \left({\ibracket{\lnot b(s)}}\right) * {\ibracket{s'=s}}} \\
        = &\cmt{Law~\ref{thm:summation_cmult_left} and proof of summable is omitted } \\
        & \lambda (s, s'). \rvprfunsym{{{\left({\infsum s_0 @ {\ibracket{b(s)}}*{{\prrvfunsym{P}}}(s,s_0) * {{{\left({\ibracket{\lnot b(s_0)} * {\ibracket{s'=s_0}}}\right)}}}}\right)}} + \left({\ibracket{\lnot b(s)}}\right) * {\ibracket{s'=s}}} \\
    %
        &\left(\lfun_P^b\right)^3 (\ufzero) \\
        = &\cmt{$\lfun^3(\ufzero) = \lfun(\lfun^2(\ufzero))$ and same as previous proof} \\
        & \lambda (s, s'). \rvprfunsym{
            \left(
            \begin{array}[]{@{}l}
                {\ibracket{b(s)}}*\prrvfunsym{\left(\rvprfunsym{\left(
                    \begin{array}[]{@{}l}
                        \infsum s_1 @ {{\prrvfunsym{P}}}(s,s_1) * \\
                        \prrvfunsym{\left(
                            \rvprfunsym{
                                \begin{array}[]{@{}l}
                                    {{\left({\infsum s_0 @ {\ibracket{b(s_1)}}*{{\prrvfunsym{P}}}(s_1,s_0) * {{{\left({\ibracket{\lnot b(s_0)} * {\ibracket{s'=s_0}}}\right)}}}}\right)}} \\
                                    + {\ibracket{\lnot b(s_1)}} * {\ibracket{s'=s_1}}
                                \end{array} }\right)}
                            \end{array} \right)}\right)} \\
                            + \left({\ibracket{\lnot b(s)}}\right) * {\ibracket{s'=s}}
                        \end{array}
                        \right)
                    } \\
                    = &\cmt{Theorem~\ref{thm:prrvfun_inverse} where the proof of $\isprob$ is omitted} \\
                    & \lambda (s, s'). \rvprfunsym{
                        \left(
                        \begin{array}[]{@{}l}
                            {\ibracket{b(s)}}*\prrvfunsym{\left(\rvprfunsym{\left(
                                \begin{array}[]{@{}l}
                                    \infsum s_1 @ {{\prrvfunsym{P}}}(s,s_1) * \\
                                    {\left( {
                                        \begin{array}[]{@{}l}
                                            {{\left({\infsum s_0 @ {\ibracket{b(s_1)}}*{{\prrvfunsym{P}}}(s_1,s_0) * {{{\left({\ibracket{\lnot b(s_0)} * {\ibracket{s'=s_0}}}\right)}}}}\right)}} \\
                                            + {\ibracket{\lnot b(s_1)}} * {\ibracket{s'=s_1}}
                                        \end{array} }\right)}
                                    \end{array} \right)}\right)} \\
                                    + \left({\ibracket{\lnot b(s)}}\right) * {\ibracket{s'=s}}
                                \end{array}
                                \right)
                            } \\
                            = &\cmt{Law~\ref{thm:summation_add} and proofs of summable omitted } \\
                            & \lambda (s, s'). \rvprfunsym{
                                \left(
                                \begin{array}[]{@{}l}
                                    {\ibracket{b(s)}}*\prrvfunsym{\left(\rvprfunsym{\left(
                                        \begin{array}[]{@{}l}
                                            \infsum s_1 @ 
                                            {\prrvfunsym{P}}(s,s_1) * \\
                                            \;{{\left({\infsum s_0 @ {\ibracket{b(s_1)}}*{{\prrvfunsym{P}}}(s_1,s_0) * {{{\left({\ibracket{\lnot b(s_0)} * {\ibracket{s'=s_0}}}\right)}}}}\right)}} 
                                            \\
                                            + \infsum s_1 @ {\prrvfunsym{P}}(s,s_1) * {\ibracket{\lnot b(s_1)}} * {\ibracket{s'=s_1}}
                                        \end{array} \right)}\right)} \\
                                        + \left({\ibracket{\lnot b(s)}}\right) * {\ibracket{s'=s}}
                                    \end{array}
                                    \right)
                                } \\
                                = &\cmt{Theorem~\ref{thm:prrvfun_inverse} where the proof of $\isprob$ is omitted} \\
                                & \lambda (s, s'). \rvprfunsym{
                                    \left(
                                    \begin{array}[]{@{}l} 
                                        {\ibracket{b(s)}} *\infsum s_1 @ {\prrvfunsym{P}}(s,s_1) * \\
                                        \;{{\left({\infsum s_0 @ {\ibracket{b(s_1)}}*{{\prrvfunsym{P}}}(s_1,s_0) * {{{\left({\ibracket{\lnot b(s_0)} * {\ibracket{s'=s_0}}}\right)}}}}\right)}} \\
                                        {\ibracket{b(s)}} *\infsum s_1 @ {\prrvfunsym{P}}(s,s_1) * {\ibracket{\lnot b(s_1)}} * {\ibracket{s'=s_1}} \\ 
                                        + \left({\ibracket{\lnot b(s)}}\right) * {\ibracket{s'=s}}
                                    \end{array} 
                                    \right) 
                                } \\
                                = &\cmt{Law~\ref{thm:summation_cmult_left} and proof of summable is omitted } \\
                                & \lambda (s, s'). \rvprfunsym{
                                    \left(
                                    \begin{array}[]{@{}l} 
                                        \infsum s_1 @ {\ibracket{b(s)}} *{\prrvfunsym{P}}(s,s_1) * \\
                                        \;{{\left({\infsum s_0 @ {\ibracket{b(s_1)}}*{{\prrvfunsym{P}}}(s_1,s_0) * {{{\left({\ibracket{\lnot b(s_0)} * {\ibracket{s'=s_0}}}\right)}}}}\right)}} \\
                                        \infsum s_1 @ {\ibracket{b(s)}} *{\prrvfunsym{P}}(s,s_1) * {\ibracket{\lnot b(s_1)}} * {\ibracket{s'=s_1}} \\ 
                                        + \left({\ibracket{\lnot b(s)}}\right) * {\ibracket{s'=s}}
                                    \end{array} 
                                    \right) 
                                } \\
                                = &\cmt{Rewrite } \\
                                & \lambda (s, s'). \rvprfunsym{
                                    \left(
                                    \begin{array}[]{@{}l} 
                                        \sum\limits_{i=1}^{2} \left(
                                        \begin{array}{@{}l}
                                            \infsum s_{i-1}.  
                                            \ibracket{b(s)}*{\prrvfunsym{P}}(s,s_{i-1})* \\
                                            \left(
                                            \begin{array}{@{}l}
                                                \infsum s_{i-2}. \ibracket{b(s_{i-1})}*{\prrvfunsym{P}}(s_{i-1}, s_{i-2})* \\\left(
                                                \begin{array}{@{}l}
                                                    \vdots * \\
                                                    \left(
                                                    \infsum s_0. \ibracket{b(s_1)}* {\prrvfunsym{P}}(s_1, s_0) * \ibracket{\lnot b(s_0)}*\ibracket{s'=s_0} 
                                                    \right)
                                                \end{array}\right)
                                            \end{array}
                                            \right)
                                        \end{array}\right) \\
                                        + \left({\ibracket{\lnot b(s)}}\right) * {\ibracket{s'=s}}
                                    \end{array} 
                                    \right) 
                                } 
                            \end{align*}
Assume $(\lfun_P^b)^n(\ufzero)$ satisfies the theorem, we show below that $(\lfun_P^b)^{n+1}(\ufzero)$ also satisfies the theorem.
                            \begin{align*}
                                &\left(\lfun_P^b\right)^{n+1} (\ufzero) \\
                                = &\cmt{Assumption, $\lfun^{n+1}(\ufzero) = \lfun(\lfun^n(\ufzero))$, and same as previous proof} \\
                                & \lambda (s, s'). \rvprfunsym{
                                    \left(
                                    \begin{array}[]{@{}l}
                                        {\ibracket{b(s)}}*\\
                                        \prrvfunsym{\left(\rvprfunsym{\left(
                                            \begin{array}[]{@{}l}
                                                \infsum s_{n-1} @ {{\prrvfunsym{P}}}(s,s_{n-1}) * \\
                                                { \left(
                                                    \begin{array}[]{@{}l} 
                                                        \sum\limits_{i=1}^{n-1} \left(
                                                        \begin{array}{@{}l}
                                                            \infsum s_{i-1}. 
                                                            \ibracket{b(s_{n-1})}*{\prrvfunsym{P}}(s_{n-1},s_{i-1})* \\
                                                            \left(
                                                            \begin{array}{@{}l}
                                                                \infsum s_{i-2}. \ibracket{b(s_{i-1})}*{\prrvfunsym{P}}(s_{i-1}, s_{i-2})* \\\left(
                                                                \begin{array}{@{}l}
                                                                    \vdots * \\
                                                                    \left(
                                                                    \infsum s_0. \ibracket{b(s_1)}* {\prrvfunsym{P}}(s_1, s_0) * \ibracket{\lnot b(s_0)}*\ibracket{s'=s_0} 
                                                                    \right)
                                                                \end{array}\right)
                                                            \end{array}
                                                            \right)
                                                        \end{array}\right) \\
                                                        + \left({\ibracket{\lnot b(s_{n-1})}}\right) * {\ibracket{s'=s_{n-1}}}
                                                    \end{array} 
                                                    \right) 
                                                }
                                            \end{array} \right)}\right)} \\
                                            + \left({\ibracket{\lnot b(s)}}\right) * {\ibracket{s'=s}}
                                        \end{array}
                                        \right)
                                    } \\
                                    = &\cmt{ Expand finite summation } \\
                                    & \lambda (s, s'). \rvprfunsym{
                                        \left(
                                        \begin{array}[]{@{}l}
                                            {\ibracket{b(s)}}*\\
                                            \prrvfunsym{\left(\rvprfunsym{\left(
                                                \begin{array}[]{@{}l}
                                                    \infsum s_{n-1} @ {{\prrvfunsym{P}}}(s,s_{n-1}) * \\
                                                    { \left(
                                                        \begin{array}[]{@{}l} 
                                                            \left(
                                                            \begin{array}{@{}l}
                                                                \infsum s_{{n-1}-1}.  
                                                                \ibracket{b(s_{n-1})}*{\prrvfunsym{p}}(s_{n-1},s_{{n-1}-1})* \\
                                                                \left(
                                                                \begin{array}{@{}l}
                                                                    \infsum s_{{n-1}-2}. \ibracket{b(s_{{n-1}-1})}*{\prrvfunsym{p}}(s_{{n-1}-1}, s_{{n-1}-2})* \\\left(
                                                                    \begin{array}{@{}l}
                                                                        \vdots * \\
                                                                        \left(
                                                                        \infsum s_0. \ibracket{b(s_1)}* {\prrvfunsym{p}}(s_1, s_0) * \ibracket{\lnot b(s_0)}*\ibracket{s'=s_0} 
                                                                        \right)
                                                                    \end{array}\right)
                                                                \end{array}
                                                                \right)
                                                            \end{array}\right) \\
                                                            + \cdots + \\
                                                            \left(
                                                            \infsum s_0. \ibracket{b(s_{n-1})}* {\prrvfunsym{p}}(s_{n-1}, s_0) * \ibracket{\lnot b(s_0)}*\ibracket{s'=s_0} 
                                                            \right) \\
                                                            + \left({\ibracket{\lnot b(s_{n-1})}}\right) * {\ibracket{s'=s_{n-1}}}
                                                        \end{array} 
                                                        \right) 
                                                    }
                                                \end{array} \right)}\right)} \\
                                                + \left({\ibracket{\lnot b(s)}}\right) * {\ibracket{s'=s}}
                                            \end{array}
                                            \right)
                                        } \\
                                        = &\cmt{Law~\ref{thm:summation_add} and proofs of summable omitted } \\
                                        & \lambda (s, s'). \rvprfunsym{
                                            \left(
                                            \begin{array}[]{@{}l}
                                                {\ibracket{b(s)}}*\\
                                                \prrvfunsym{\left(\rvprfunsym{\left(
                                                    \begin{array}[]{@{}l}
                                                        { \left(
                                                            \begin{array}[]{@{}l} 
                                                                \left(
                                                                \begin{array}{@{}l}
                                                                    \infsum s_{n-1} @ {{\prrvfunsym{P}}}(s,s_{n-1}) * \\ 
                                                                    \infsum s_{{n-1}-1}.  
                                                                    \ibracket{b(s_{n-1})}*{\prrvfunsym{p}}(s_{n-1},s_{{n-1}-1})* \\
                                                                    \left(
                                                                    \begin{array}{@{}l}
                                                                        \infsum s_{{n-1}-2}. \ibracket{b(s_{{n-1}-1})}*{\prrvfunsym{p}}(s_{{n-1}-1}, s_{{n-1}-2})* \\\left(
                                                                        \begin{array}{@{}l}
                                                                            \vdots * \\
                                                                            \left(
                                                                            \infsum s_0. \ibracket{b(s_1)}* {\prrvfunsym{p}}(s_1, s_0) * \ibracket{\lnot b(s_0)}*\ibracket{s'=s_0} 
                                                                            \right)
                                                                        \end{array}\right)
                                                                    \end{array}
                                                                    \right)
                                                                \end{array}\right) \\
                                                                + \cdots + \\
                                                                \infsum s_{n-1} @ {{\prrvfunsym{P}}}(s,s_{n-1}) * \\
                                                                \infsum s_0. \ibracket{b(s_{n-1})}* {\prrvfunsym{p}}(s_{n-1}, s_0) * \ibracket{\lnot b(s_0)}*\ibracket{s'=s_0} \\ 
                                                                + 
                                                                \infsum s_{n-1} @ {{\prrvfunsym{P}}}(s,s_{n-1}) * 
                                                                \left({\ibracket{\lnot b(s_{n-1})}}\right) * {\ibracket{s'=s_{n-1}}}
                                                            \end{array} 
                                                            \right) 
                                                        }
                                                    \end{array} \right)}\right)} \\
                                                    + \left({\ibracket{\lnot b(s)}}\right) * {\ibracket{s'=s}}
                                                \end{array}
                                                \right)
                                            } \\
                                            = &\cmt{Law~\ref{thm:summation_cmult_left} and combine summation } \\
                                            & \lambda (s, s'). \rvprfunsym{
                                                \left(
                                                \begin{array}[]{@{}l} 
                                                    \sum\limits_{i=1}^{n} \left(
                                                    \begin{array}{@{}l}
                                                        \infsum s_{i-1}.  
                                                        \ibracket{b(s)}*{\prrvfunsym{P}}(s,s_{i-1})* \\
                                                        \left(
                                                        \begin{array}{@{}l}
                                                            \infsum s_{i-2}. \ibracket{b(s_{i-1})}*{\prrvfunsym{P}}(s_{i-1}, s_{i-2})* \\\left(
                                                            \begin{array}{@{}l}
                                                                \vdots * \\
                                                                \left(
                                                                \infsum s_0. \ibracket{b(s_1)}* {\prrvfunsym{P}}(s_1, s_0) * \ibracket{\lnot b(s_0)}*\ibracket{s'=s_0} 
                                                                \right)
                                                            \end{array}\right)
                                                        \end{array}
                                                        \right)
                                                    \end{array}\right) \\
                                                    + \left({\ibracket{\lnot b(s)}}\right) * {\ibracket{s'=s}}
                                                \end{array} 
                                                \right) 
                                            } \\
                                        \end{align*}
                                        This concludes the proof.
                                    \end{proof}

\begin{thm}
    \label{thm:iterdiff_top}
    Provided $P$ is a distribution, that is, $\isfinaldist(P)$. 
    \begin{align*}
        &\left(\lfun_P^b\right)^0(\ufone) = \ufone \\ 
        &\left(\lfun_P^b\right)^1(\ufone) = \lambda (s, s'). ~\rvprfunsym{\infsum s_0. \ibracket{b(s)}* {\prrvfunsym{P}}(s, s_0) + \ibracket{\lnot b(s)}*\ibracket{s'=s}}
    \end{align*}
    If $n>1$, then
    \begin{align*}
        & \left(\lfun_P^b\right)^n(\ufone) = \\ 
        & \lambda (s, s'). \rvprfunsym{
            \left(
            \begin{array}[]{@{}l} 
                \left(
                \begin{array}{@{}l}
                    \infsum s_{n-1}.  
                    \ibracket{b(s)}*{\prrvfunsym{P}}(s,s_{n-1})* \\
                    \left(
                    \begin{array}{@{}l}
                        \infsum s_{n-2}. \ibracket{b(s_{n-1})}*{\prrvfunsym{P}}(s_{n-1}, s_{n-2})* \\\left(
                        \begin{array}{@{}l}
                            \vdots * \\
                            \left(
                            \infsum s_0. \ibracket{b(s_1)}* {\prrvfunsym{P}}(s_1, s_0) 
                            \right)
                        \end{array}\right)
                    \end{array}
                    \right)
                \end{array}\right) + \hfill(\textnormal{\color{red} diff})\\
                \sum\limits_{i=1}^{n-1} \left(
                \begin{array}{@{}l}
                    \infsum s_{i-1}.  
                    \ibracket{b(s)}*{\prrvfunsym{P}}(s,s_{i-1})* \\
                    \left(
                    \begin{array}{@{}l}
                        \infsum s_{i-2}. \ibracket{b(s_{i-1})}*{\prrvfunsym{P}}(s_{i-1}, s_{i-2})* \\\left(
                        \begin{array}{@{}l}
                            \vdots * \\
                            \left(
                            \infsum s_0. \ibracket{b(s_1)}* {\prrvfunsym{P}}(s_1, s_0) * \ibracket{\lnot b(s_0)}*\ibracket{s'=s_0} 
                            \right)
                        \end{array}\right)
                    \end{array}
                    \right)
                \end{array}\right) \\
                + \left({\ibracket{\lnot b(s)}}\right) * {\ibracket{s'=s}}
            \end{array} 
            \right) 
        } 
    \end{align*}
    We note that the part marked with (\textnormal{\color{red} diff}) is the only difference of this $\left(\lfun_P^b\right)^n(\ufone)$ from $\left(\lfun_P^b\right)^n(\ufzero)$. 
\end{thm}

\begin{proof}
    We show below that $\left(\lfun_P^b\right)^n(\ufone)$ for $n=0$ to $3$ satisfies the theorem.
    \begin{align*}
        &\left(\lfun_P^b\right)^0(\ufone) = \lambda (s, s'). 0 = \ufone \\
        &\left(\lfun_P^b\right)^1(\ufone) \\
        = & \cmt{Defintion~\ref{def:lfun}} \\
        & \pcchoice{b}{\left(\pseq{P}{\ufone}\right)}{\pskip}\\ 
        = & \cmt{Law~\ref{thm:lfun_altdef}} \\
        & \rvprfunsym{{\ibracket{b}}*\prrvfunsym{\left(\pseq{P}{\ufone}\right)} + \ibracket{\lnot b} * {\ibracket{\II}}} \\
        = & \cmt{$\pseq{P}{\ufone}= \rvprfunsym{{\infsum v_0 @ \prrvfunsym{P}[v_0/\vv']}}$, see the proof of Theorem~\ref{thm:prog_seq_comp} Law~\ref{thm:pseq_one} } \\
        & \rvprfunsym{{\ibracket{b}}*\prrvfunsym{\left(\rvprfunsym{{\infsum v_0 @ \prrvfunsym{P}[v_0/\vv']}}\right)} + \ibracket{\lnot b} * {\ibracket{\II}}} \\
        = &\cmt{Expand as a function form, use $s$ and $s'$ for initial and final observation states} \\
        &\cmt{Substitution and Definition~\ref{def:uskip}, Theorem~\ref{thm:prrvfun_inverse}, and Law~\ref{thm:summation_cmult_left} } \\
         &\lambda (s, s'). ~\rvprfunsym{\infsum s_0. \ibracket{b(s)}* {\prrvfunsym{P}}(s, s_0) + \ibracket{\lnot b(s)}*\ibracket{s'=s}}\\
        &\left(\lfun_P^b\right)^2 (\ufone) \\
        = &\cmt{$\lfun^2(\ufone) = \lfun(\lfun^1(\ufone))$ and same as previous proof} \\
        & \lambda (s, s'). \rvprfunsym{
            \left(
            \begin{array}[]{@{}l}
                {\ibracket{b(s)}}*\prrvfunsym{\left(\rvprfunsym{\left(
                    \begin{array}[]{@{}l}
                        \infsum s_1 @ {{\prrvfunsym{P}}}(s,s_1) * \\
                        \prrvfunsym{\left(
                            \rvprfunsym{
                                \begin{array}[]{@{}l}
                                    \infsum s_0. \ibracket{b(s_1)}* {\prrvfunsym{P}}(s_1, s_0) 
                                    + {\ibracket{\lnot b(s_1)}} * {\ibracket{s'=s_1}}
                                \end{array} }\right)}
                            \end{array} \right)}\right)} \\
                            + \left({\ibracket{\lnot b(s)}}\right) * {\ibracket{s'=s}}
                        \end{array}
                        \right)
                    } \\
        = &\cmt{Theorem~\ref{thm:prrvfun_inverse} where the proof of $\isprob$ is omitted} \\
          &\cmt{Law~\ref{thm:summation_cmult_left} and proof of summable is omitted } \\
        & \lambda (s, s'). \rvprfunsym{
            \left(
            \begin{array}[]{@{}l}
                \infsum s_1 @ \ibracket{b(s)}*{{\prrvfunsym{P}}}(s,s_1) * 
                \left( \infsum s_0. \ibracket{b(s_1)}* {\prrvfunsym{P}}(s_1, s_0) \right) + \\
                \infsum s_1 @ \ibracket{b(s)}*{{\prrvfunsym{P}}}(s,s_1) * 
                \left( {\ibracket{\lnot b(s_1)}} * {\ibracket{s'=s_1}} \right) + \\
                {\ibracket{\lnot b(s)}} * {\ibracket{s'=s}}
            \end{array}
            \right)
        } 
\end{align*}
Assume $(\lfun_P^b)^n(\ufone)$ satisfies the theorem, we show below that $(\lfun_P^b)^{n+1}(\ufone)$ also satisfies the theorem.
\begin{align*}
    &\left(\lfun_P^b\right)^{n+1} (\ufone) \\
    = &\cmt{Assumption, $\lfun^{n+1}(\ufone) = \lfun(\lfun^n(\ufone))$, and same as previous proof} \\
    & \lambda (s, s'). \rvprfunsym{
        \left(
        \begin{array}[]{@{}l}
            {\ibracket{b(s)}}*\\
            \prrvfunsym{\left(\rvprfunsym{\left(
                \begin{array}[]{@{}l}
                    \infsum s_{n} @ {{\prrvfunsym{P}}}(s,s_{n}) * \\
            { \left(
            \begin{array}[]{@{}l} 
                \left(
                \begin{array}{@{}l}
                    \infsum s_{n-1}.  
                    \ibracket{b(s_n)}*{\prrvfunsym{P}}(s_n,s_{n-1})* \\
                    \left(
                    \begin{array}{@{}l}
                        \infsum s_{n-2}. \ibracket{b(s_{n-1})}*{\prrvfunsym{P}}(s_{n-1}, s_{n-2})* \\\left(
                        \begin{array}{@{}l}
                            \vdots * \\
                            \left(
                            \infsum s_0. \ibracket{b(s_1)}* {\prrvfunsym{P}}(s_1, s_0) 
                            \right)
                        \end{array}\right)
                    \end{array}
                    \right)
                \end{array}\right) + \\
                \sum\limits_{i=1}^{n-1} \left(
                \begin{array}{@{}l}
                    \infsum s_{i-1}.  
                    \ibracket{b(s_n)}*{\prrvfunsym{P}}(s_n,s_{i-1})* \\
                    \left(
                    \begin{array}{@{}l}
                        \infsum s_{i-2}. \ibracket{b(s_{i-1})}*{\prrvfunsym{P}}(s_{i-1}, s_{i-2})* \\\left(
                        \begin{array}{@{}l}
                            \vdots * \\
                            \left(
                            \infsum s_0. \ibracket{b(s_1)}* {\prrvfunsym{P}}(s_1, s_0) * \ibracket{\lnot b(s_0)}*\ibracket{s'=s_0} 
                            \right)
                        \end{array}\right)
                    \end{array}
                    \right)
                \end{array}\right) \\
                + \left({\ibracket{\lnot b(s)}}\right) * {\ibracket{s'=s}}
            \end{array} 
            \right) 
        }
                \end{array} \right)}\right)} \\
                + \left({\ibracket{\lnot b(s)}}\right) * {\ibracket{s'=s}}
            \end{array}
            \right)
        } \\
        = &\cmt{ Expand finite summation, Law~\ref{thm:summation_cmult_left}, same as previous proof } \\
        & \lambda (s, s'). \rvprfunsym{
            \left(
            \begin{array}[]{@{}l} 
                \left(
                \begin{array}{@{}l}
                    \infsum s_{n}.  
                    \ibracket{b(s)}*{\prrvfunsym{P}}(s,s_{n})* \\
                    \left(
                    \begin{array}{@{}l}
                        \infsum s_{n-1}. \ibracket{b(s_{n})}*{\prrvfunsym{P}}(s_{n}, s_{n-1})* \\\left(
                        \begin{array}{@{}l}
                            \vdots * \\
                            \left(
                            \infsum s_0. \ibracket{b(s_1)}* {\prrvfunsym{P}}(s_1, s_0) 
                            \right)
                        \end{array}\right)
                    \end{array}
                    \right)
                \end{array}\right) + \\
                \sum\limits_{i=1}^{n} \left(
                \begin{array}{@{}l}
                    \infsum s_{i-1}.  
                    \ibracket{b(s)}*{\prrvfunsym{P}}(s,s_{i-1})* \\
                    \left(
                    \begin{array}{@{}l}
                        \infsum s_{i-2}. \ibracket{b(s_{i-1})}*{\prrvfunsym{P}}(s_{i-1}, s_{i-2})* \\\left(
                        \begin{array}{@{}l}
                            \vdots * \\
                            \left(
                            \infsum s_0. \ibracket{b(s_1)}* {\prrvfunsym{P}}(s_1, s_0) * \ibracket{\lnot b(s_0)}*\ibracket{s'=s_0} 
                            \right)
                        \end{array}\right)
                    \end{array}
                    \right)
                \end{array}\right) \\
                + \left({\ibracket{\lnot b(s)}}\right) * {\ibracket{s'=s}}
            \end{array} 
            \right) 
        } 
\end{align*}
This concludes the proof.
\end{proof}

\begin{thm}
    \label{thm:iterdiff_eq}
    \begin{align*}
       \forall n:\nat | n \geq 1 \bullet  
\iterdiff(n, b, P) = \lambda (s, s'). \rvprfunsym{
\left(
                \begin{array}{@{}l}
                    \infsum s_{n-1}.  
                    \ibracket{b(s)}*{\prrvfunsym{P}}(s,s_{n-1})* \\
                    \left(
                    \begin{array}{@{}l}
                        \infsum s_{n-2}. \ibracket{b(s_{n-1})}*{\prrvfunsym{P}}(s_{n-1}, s_{n-2})* \\\left(
                        \begin{array}{@{}l}
                            \vdots * \\
                            \left(
                            \infsum s_0. \ibracket{b(s_1)}* {\prrvfunsym{P}}(s_1, s_0) 
                            \right)
                        \end{array}\right)
                    \end{array}
                    \right)
                \end{array}
            \right)}
    \end{align*}
\end{thm}

\begin{proof}
   If $n=1$, then  
   \begin{align*}
      & \iterdiff(1, b, P) = \\ 
     = &\cmt{Defintions~\ref{def:iterdiff} and $\lfundiff$} \\
     & \pcchoice{b}{\left(\pseq{P}{\ufone}\right)}{\ufzero} \\
     = &\cmt{Theorem~\ref{thm:prog_cond_choice} Law~\ref{thm:cchoice_pchoice} and Theorem~\ref{thm:prob_prob_choice} Law~\ref{thm:pchoice_altdef}, and Theorem~\ref{thm:prrvfun_inverse} } \\
     & \rvprfunsym{{\ibracket{b}}*\prrvfunsym{\left(\pseq{P}{\ufone}\right)}}\\
     = & \cmt{$\pseq{P}{\ufone}= \rvprfunsym{{\infsum v_0 @ \prrvfunsym{P}[v_0/\vv']}}$ } \\
       &\cmt{Expand as a function form, use $s$ and $s'$ for initial and final observation states} \\
       &\lambda (s, s').~\rvprfunsym{\infsum s_0. \ibracket{b(s)}* {\prrvfunsym{P}}(s, s_0)}
   \end{align*}
Assume $\iterdiff(n, b, P)$ satisfies the theorem, we show below that $\iterdiff(n+1, b, P)$ also satisfies the theorem.
   \begin{align*}
      & \iterdiff(n+1, b, P) = \\ 
     = &\cmt{Defintions~\ref{def:iterdiff} and $\lfundiff$} \\
     & \pcchoice{b}{\left(\pseq{P}{\iterdiff(n, b, P)}\right)}{\ufzero} \\
     = &\cmt{Theorem~\ref{thm:prog_cond_choice} Law~\ref{thm:cchoice_pchoice} and Theorem~\ref{thm:prob_prob_choice} Law~\ref{thm:pchoice_altdef}, and Theorem~\ref{thm:prrvfun_inverse} } \\
     & \rvprfunsym{{\ibracket{b}}*\prrvfunsym{\left(\pseq{P}{\iterdiff(n, b, P)}\right)}}\\
     = &\cmt{Definition~\ref{def:prog_seq}} \\
       & \rvprfunsym{{\ibracket{b}}*\prrvfunsym{\left(\rvprfunsym{\left({\infsum v_0 @ {{\prrvfunsym{P}}}[v_0/\vv'] * {{\prrvfunsym{\left(\rvprfunsym{\iterdiff(n, b, P)}\right)}}}[v_0/\vv]}\right)}\right)}} \\
     = &\cmt{Theorem~\ref{thm:prrvfun_inverse} where the proof of $\isprob$ is omitted} \\
       & \rvprfunsym{{\ibracket{b}}*{{\left({\infsum v_0 @ {{\prrvfunsym{P}}}[v_0/\vv'] * {{{\left({\iterdiff(n, b, P)}\right)}}}[v_0/\vv]}\right)}}} \\
     = &\cmt{Law~\ref{thm:summation_cmult_left} and proof of summable is omitted } \\
       & \rvprfunsym{{\left( {{\infsum v_0 @ {\ibracket{b}}*{{\prrvfunsym{P}}}[v_0/\vv'] * {{{\left({\iterdiff(n, b, P)}\right)}}}[v_0/\vv]}}\right)}} \\
     = &\cmt{Assumption, expand as a function form, use $s$ and $s'$ for initial and final observation states} \\
      &\lambda (s, s'). \rvprfunsym{ \left(
                \begin{array}{@{}l}
                    \infsum s_{n}.  
                    \ibracket{b(s)}*{\prrvfunsym{P}}(s,s_{n})* \\
                    \left(
                    \begin{array}{@{}l}
                        \infsum s_{n-1}. \ibracket{b(s_{n})}*{\prrvfunsym{P}}(s_{n}, s_{n-1})* \\\left(
                        \begin{array}{@{}l}
                            \vdots * \\
                            \left(
                            \infsum s_0. \ibracket{b(s_1)}* {\prrvfunsym{P}}(s_1, s_0) 
                            \right)
                        \end{array}\right)
                    \end{array}
                    \right)
                \end{array}
            \right)}
   \end{align*}
   This concludes the proof of Theorem~\ref{thm:iterdiff_eq}.
\end{proof}

We now show the proof of Theorem~\ref{thm:iterdiff}:
    $\forall n:\nat \bullet {\lfunbp}^n(\ufone) - {\lfunbp}^n(\ufzero) = \iterdiff(n, b, P)$
\begin{proof}
    For $n=0$, 
    \begin{align*}
        & {\lfunbp}^0(\ufone) - {\lfunbp}^0(\ufzero) \\ 
        = &\cmt{Theorems~\ref{thm:iterdiff_bot} and \ref{thm:iterdiff_top}} \\
        & \ufone - \ufzero \\
        = &\cmt{Theorem~\ref{thm:top_bot}} \\
        & \ufone \\
        = &\cmt{Definition~\ref{def:iterdiff}} \\
        & \iterdiff(0, b, P)
    \end{align*}
    For $n=1$, 
    \begin{align*}
        & {\lfunbp}^1(\ufone) - {\lfunbp}^1(\ufzero) \\ 
        = &\cmt{Theorems~\ref{thm:iterdiff_bot} and \ref{thm:iterdiff_top}} \\
        & \lambda (s, s'). ~\rvprfunsym{\infsum s_0. \ibracket{b(s)}* {\prrvfunsym{P}}(s, s_0) + \ibracket{\lnot b(s)}*\ibracket{s'=s}} - \lambda (s, s').~\rvprfunsym{\ibracket{\lnot b(s)}*\ibracket{s'=s}} \\
        = &\cmt{Definition~\ref{def:urf_pointwise}} \\
        & \lambda (s, s'). \left(\rvprfunsym{\infsum s_0. \ibracket{b(s)}* {\prrvfunsym{P}}(s, s_0) + \ibracket{\lnot b(s)}*\ibracket{s'=s}} - \rvprfunsym{\ibracket{\lnot b(s)}*\ibracket{s'=s}}\right) \\
        = &\cmt{Definition~\ref{def:bounded_plus_minus}} \\
        & \lambda (s, s'). \ru{\left(\umax\left(0, \ur{\left(\rvprfunsym{\infsum s_0. \ibracket{b(s)}* {\prrvfunsym{P}}(s, s_0) + \ibracket{\lnot b(s)}*\ibracket{s'=s}}\right)} - \ur{\left(\rvprfunsym{\ibracket{\lnot b(s)}*\ibracket{s'=s}}\right)}\right)\right)}\\
        = &\cmt{Theorem~\ref{thm:prrvfun_inverse} where the proof of $\isprob$ is omitted} \\
        & \lambda (s, s'). \ru{\left(\umax\left(0, {\left({\infsum s_0. \ibracket{b(s)}* {\prrvfunsym{P}}(s, s_0) + \ibracket{\lnot b(s)}*\ibracket{s'=s}}\right)} - {\left({\ibracket{\lnot b(s)}*\ibracket{s'=s}}\right)}\right)\right)}\\
        = &\cmt{$\forall s \bullet \infsum s_0. \ibracket{b(s)}* {\prrvfunsym{P}}(s, s_0) \geq 0$ because $P$ is a distribution, Theorem~\ref{thm:final_distribtion}, and Definition~\ref{def:isprob} } \\
        & \lambda (s, s').~\ru{\infsum s_0. \ibracket{b(s)}* {\prrvfunsym{P}}(s, s_0)}\\
        = &\cmt{Theorem~\ref{thm:iterdiff_eq}} \\
        & \iterdiff(1, b, P)
    \end{align*}
    For $n>1$, we assume ${\lfunbp}^n(\ufone) - {\lfunbp}^n(\ufzero)= \iterdiff(n, b, P)$, then 
    \begin{align*}
        & {\lfunbp}^{n+1}(\ufone) - {\lfunbp}^{n+1}(\ufzero) \\ 
        = &\cmt{Theorems~\ref{thm:iterdiff_bot} and \ref{thm:iterdiff_top}, and same as previous proof} \\
        & \lambda (s, s'). \rvprfunsym{
                \left(
                \begin{array}{@{}l}
                    \infsum s_{n}.  
                    \ibracket{b(s)}*{\prrvfunsym{P}}(s,s_{n})* \\
                    \left(
                    \begin{array}{@{}l}
                        \infsum s_{n-1}. \ibracket{b(s_{n})}*{\prrvfunsym{P}}(s_{n}, s_{n-1})* \\\left(
                        \begin{array}{@{}l}
                            \vdots * \\
                            \left(
                            \infsum s_0. \ibracket{b(s_1)}* {\prrvfunsym{P}}(s_1, s_0) 
                            \right)
                        \end{array}\right)
                    \end{array}
                    \right)
                \end{array}\right) 
        } \\ 
        = &\cmt{Theorem~\ref{thm:iterdiff_eq}} \\
        & \iterdiff(n+1, b, P)
    \end{align*}
   This concludes the proof of Theorem~\ref{thm:iterdiff}.
\end{proof}
}

{
\subsection{Proof of Theorem~\ref{thm:cflip_iterdiff_0}}
\begin{proof}
    \begin{align*}
        & \lambda n @ \prrvfunsym{\iterdiff\left(n, c=tl, cflip\right)}(s,s') \\
        = &\cmt{Theorem~\ref{thm:iterdiff_eq} where $b=(c=tl)$ and $P=cflip$} \\
        & \lambda n @ \prrvfunsym{ \left(
\rvprfunsym{\lambda (s, s').  \left(
                \begin{array}{@{}l}
                    \infsum s_{n-1}.  
                    \ibracket{b(s)}*{\prrvfunsym{P}}(s,s_{n-1})* \\
                    \left(
                    \begin{array}{@{}l}
                        \infsum s_{n-2}. \ibracket{b(s_{n-1})}*{\prrvfunsym{P}}(s_{n-1}, s_{n-2})* \\\left(
                        \begin{array}{@{}l}
                            \vdots * \\
                            \left(
                            \infsum s_0. \ibracket{b(s_1)}* {\prrvfunsym{P}}(s_1, s_0) 
                            \right)
                        \end{array}\right)
                    \end{array}
                    \right)
                \end{array}
            \right)}
    \right)}(s,s') \\
    = &\cmt{Definition~\ref{def:coin_flip} where $cstate$ only has one variable $c$ and so $s$ is replaced by $c$} \\
    &\cmt{${\prrvfunsym{cflip}}={{1/2}*{\ibracket{c' = hd}} + {1/2} * {\ibracket{c' = tl}}}$ according to Law~\ref{thm:cflip_altdef} and Theorem~\ref{thm:prrvfun_inverse_ibracket}} \\
        & \lambda n @ \prrvfunsym{ \left(
\rvprfunsym{\lambda (c, c').  \left(
                \begin{array}{@{}l}
                    \infsum s_{n-1}.  
                    \ibracket{c=tl}*({{1/2}*{\ibracket{s_{n-1} = hd}} + {1/2} * {\ibracket{s_{n-1} = tl}}})* \\
                    \left(
                    \begin{array}{@{}l}
                        \infsum s_{n-2}. 
                        \ibracket{s_{n-1}=tl}*({{1/2}*{\ibracket{s_{n-2} = hd}} + {1/2} * {\ibracket{s_{n-2} = tl}}})* \\
                        \left(
                        \begin{array}{@{}l}
                            \vdots * \\
                            \left(
                            \infsum s_0. 
                                \ibracket{s_1=tl}*({{1/2}*{\ibracket{s_0 = hd}} + {1/2} * {\ibracket{s_0 = tl}}})
                            \right)
                        \end{array}\right)
                    \end{array}
                    \right)
                \end{array}
            \right)}
    \right)}(c,c') \\
    = &\cmt{ Every summation variable such as $s_{n-1}$ only when it is equal to $tl$, then $\ibracket{s_{n-1}=tl}=1$. } \\
     &\cmt { Otherwise, it is 0 and the whole summation is also 0 because $a*\cdots*0*\cdots*b=0$.} \\
     &\cmt { All $\infsum$ are removed because of only one state $tl$ satisfying $s_i=tl$ } \\
     &\cmt { All $\ibracket{s_i=hd} = 0$ } \\
        & \lambda n @ \prrvfunsym{ \left(
\rvprfunsym{\lambda (c, c').  \left(
                \begin{array}{@{}l}
                    \ibracket{c=tl}*({{1/2}*1})* \\
                    \left(
                    \begin{array}{@{}l}
                        1*({{1/2} * 1})* \\
                        \left(
                        \begin{array}{@{}l}
                            \vdots * \\
                            \left(
                                1*({{1/2} * 1})
                            \right)
                        \end{array}\right)
                    \end{array}
                    \right)
                \end{array}
            \right)}
    \right)}(c,c') \\
    = &\cmt{ Rewrite } \\
    & \lambda n @ \prrvfunsym{ \left( \rvprfunsym{\lambda (c, c'). \ibracket{c=tl}*(1/2)^{n}} \right)}(c,c') \\
    = &\cmt{Theorem~\ref{thm:prrvfun_inverse} where $\isprob\left({\lambda (c, c'). \ibracket{c=tl}*(1/2)^{n}}\right)$} \\
    & \lambda n @ \ibracket{c=tl}*(1/2)^{n} 
\end{align*}
So if $c=hd$, then 
\begin{align*}
\left(\lambda n @ \prrvfunsym{\iterdiff\left(n, c=tl, cflip\right)}(s,s') = \lambda n @ 0\right)
\end{align*}
Otherwise, 
\begin{align*}
\left(\lambda n @ \prrvfunsym{\iterdiff\left(n, c=tl, cflip\right)}(s,s') = \lambda n @ (1/2)^n\right)
\end{align*}
We conclude that $\forall (c,c'):cstate\cross cstate$ such that 
\begin{align*}
\left(\lambda n @ \prrvfunsym{\iterdiff\left(n, c=tl, cflip\right)}(c,c')\right) \tendsto 0
\end{align*}
according to Definition~\ref{def:tendsto}.
\end{proof}
}

{
\subsection{The necessity of $\finstatesasc\left(\lambda n @ \iter\left(n, b, P, \ufzero\right) \right)$}
\label{appendix:necessity_finstatesasc}
We show below to establish the continuity Theorem~\ref{thm:continuity_lfun_bot}, $\finstatesasc\left(\lambda n @ \iter\left(n, b, P, \ufzero\right) \right)$ is required. Here we omit the details about type conversion to make the discussion clearer.

We aim to prove a continuity theorem below and then our proof later shows that without the premise, we cannot prove such theorem.
\begin{thm} \label{thm:continuity}
	If the final state of $P$ is a distribution, then $\lfunbp(X)$ is continuous. That is, 
    for an non-empty countable increasing chain $S_0 \leq S_1 \leq S_2 \leq ... $ of type $[s]prfun$ and bound above (so has a supremum), then
	\begin{align*}
		\lfunbp\left(\thnsup n @S_n\right) = 
        \thnsup n @ \lfunbp(S_n) \tag*{($\lfunbp$ continuous)} \label{thm:Fbp_continuous}
	\end{align*}
\end{thm}

\begin{proof}
    We define  
	\begin{align*}
        & F_1(X) \defs \pcchoice{b}{X}{\pskip} \\ 
        & F_2(X) \defs \pseq{P}{X}
	\end{align*}
    So $\lfunbp(X) \defs~F_1(F_2(X))$. 
    According to \cite[Lemma 5.3]{Nielson2007}, if both $F_1$ and $F_2$ are continuous, then $\lfunbp(X)$ is also continuous. It is trivial to show that $F_1$ is continuous, so our goal is to prove $F_2$ is continuous. That is, 
	\begin{align*}
		F_2\left(\thnsup n @S_n\right) = \thnsup n @ F_2(S_n) \tag*{($F_2$ continuous)} \label{thm:F2_continuous}
	\end{align*}

    Because $S_n$ is an increasing chain and bound above, according to the monotone sequence theorem, the limit of $S_n$ is just its supremum. That is,
	\begin{align*}
        \forall (s,s') @ \left(\lambda n @ S_n (s,s')\right)  \tendsto \left(\thnsup n @S_n\right)(s,s')
	\end{align*}
where $(s,s')$ denotes the initial and final observations. 
    According to Definition~\ref{def:tendsto}, 
    \begin{align*}
        &\forall (s,s') @ \forall \epsilon:\real > 0 \bullet \exists M:\nat \bullet \forall l \geq M \bullet \left({\left(\thnsup n @S_n\right)(s,s')} - {S_l(s,s')}\right)< \epsilon \tag*{($S_n$ supremum as limit)} \label{thm:Sn_limit_sup}
    \end{align*}

    According to Theorem~\ref{thm:prog_seq_comp} Law~\ref{thm:pseq_mono}, $F_2$ is monotonic. It is trivial to show that $F_2(S_n)$ is also an increasing chain and bound above. According to the monotone sequence theorem,
	\begin{align*}
        \forall (s,s') @ \left(\lambda n @ F_2(S_n) (s,s')\right) \tendsto \left(\thnsup n @F_2(S_n)\right)(s,s')
		\tag*{($F_2(S_n)$ supremum as limit)} \label{thm:F2_Sn_limit}
	\end{align*}
    Therefore, according to the unique sequence limit theorem (if exists), to prove \ref{thm:F2_continuous}, we need to prove that 
	\begin{align*}
        \forall (s,s') @ \left(\lambda n @ F_2(S_n) (s,s')\right) \tendsto F_2\left(\thnsup n @S_n\right) (s,s')
		\tag*{($F_2\left(\thnsup n @S_n\right)$ as limit)} \label{thm:F2_Sup_Sn_limit}
	\end{align*}
    According to Definition~\ref{def:tendsto}, this is equal to prove
    \begin{align*}
        &\forall (s,s') @ \forall \varepsilon:\real > 0 \bullet \exists N:\nat \bullet \forall l \geq N \bullet |F_2(S_l)(s,s') - F_2\left(\thnsup n @S_n\right)(s,s')| < \varepsilon \\
= & \cmt{$F_2$ is monotonic and $S_l \leq \left(\thnsup n @S_n\right)$, so $F_2(S_l)(s,s') \leq F_2\left(\thnsup n @S_n\right)(s,s')$} \\
        &\forall (s,s') @ \forall \varepsilon:\real > 0 \bullet \exists N:\nat \bullet \forall l \geq N \bullet F_2\left(\thnsup n @S_n\right)(s,s') - F_2(S_l)(s,s') < \varepsilon
    \end{align*}

    We can rewrite the non-quantifier part above to 
    \begin{align*}
        & |F_2(S_l)(s,s') - F_2\left(\thnsup n @S_n\right)(s,s')| < \varepsilon \\
        = & \cmt{$F_2$ is monotonic and $S_l \leq \left(\thnsup n @S_n\right)$, so $F_2(S_l)(s,s') \leq F_2\left(\thnsup n @S_n\right)(s,s')$} \\
        &F_2\left(\thnsup n @S_n\right)(s,s') - F_2(S_l)(s,s') < \varepsilon \\
        = & \cmt{Definitions of $F_2$ and $\pseq{}{}$~\ref{def:prog_seq}} \\
        &{\infsum s_0 @ {P}(s, s_0) * {\left(\thnsup n @S_n\right)(s_0,s')}} - {\infsum s_0 @ {P}(s, s_0) * {S_l(s_0,s')}} < \varepsilon \\
        = & \cmt{Theorem~\ref{thm:summation} Law~\ref{thm:summation_minus} and proof of summable is omitted } \\
          &{\infsum s_0 @ {P}(s, s_0) * \left({\left(\thnsup n @S_n\right)(s_0,s')} - {S_l(s_0,s')}\right)} < \varepsilon
    \end{align*}

    So our goal is to prove 
    \begin{align*}
        &\forall (s,s') @ \forall \varepsilon:\real > 0 \bullet \exists N:\nat \bullet \forall l \geq N \bullet \left({\infsum s_0 @ {P}(s, s_0) * \left({\left(\thnsup n @S_n\right)(s_0,s')} - {S_l(s_0,s')}\right)}\right) < \varepsilon
    \end{align*}
    The key step in proving the goal above is to supply a witness for $N$. Based on \ref{thm:Sn_limit_sup}, we can choose a $N$ to make $\left({\left(\thnsup n @S_n\right)(s_0,s')} - {S_l(s_0,s')}\right)$ any small (say $\epsilon(s_0)$), but this cannot guarantees 
    \begin{align*}
    \left({\infsum s_0 @ {P}(s, s_0) * \epsilon(s_0)}\right) < \varepsilon
    \end{align*}
    because this is an infinite sum. In other words, the question is to find a $N$ such that this summation converges to a value less than any number $\varepsilon$. The approach we use in this paper is to assume $\finstatesasc\left(\lambda n @ \iter\left(n, b, P, \ufzero\right) \right)$, and then we can construct such $N$.

    We omit further details of this proof.

\end{proof}
}

\bibliographystyle{elsarticle-num} 
\bibliography{main}

\ifdefined \CHANGES \indexprologue{%
  This index lists for each comment the pages where the text has been modified to address the comment. Since the same page may contain multiple changes, the page number contains the index of the change in superscript to identify different changes. Finally, the page number contains a hyperlink that takes the reader to corresponding change.%
}%
\printindex[changes] \fi

\end{document}